\documentclass[12pt]{article}
\usepackage{stmaryrd}
\usepackage{bbding}
\usepackage{mathrsfs}
\usepackage{amsfonts}
\usepackage{amssymb}
\usepackage{amsthm}
\usepackage{amsmath}
\usepackage{algorithm}
\usepackage[noend]{algpseudocode}

\usepackage{newtxtext}
\usepackage{newtxmath}
\usepackage{tabularx}
\usepackage{rotating}
\usepackage{adjustbox}
\usepackage{multirow}
\usepackage{booktabs}
\usepackage{makecell}
\usepackage{threeparttable}
\usepackage{lscape}

\usepackage{graphicx}
\usepackage{epsfig}
\usepackage{epstopdf}
\usepackage{subcaption}
\usepackage{caption}
\usepackage{color}
\usepackage{eurosym}
\usepackage{natbib}
\usepackage{hyperref}
\hypersetup{
    colorlinks=true,
    linkcolor=blue,
    urlcolor=blue,
    citecolor=blue,
}
\usepackage{geometry}
\usepackage{enumitem}
\usepackage{longtable}

\numberwithin{equation}{section}
\allowdisplaybreaks[4]

\newtheorem{lemma}{Lemma}[section]
\newtheorem{example}{Example}[section]
\newtheorem{definition}{Definition}
\newtheorem{assumption}{Assumption}[section]
\newtheorem{theorem}{Theorem}[section]
\newtheorem{remark}{Remark}[section]

\def\EE{{\mathbb{E}}}

\def\cov{{\mathbb{C}ov}}

\newcommand{\NN}{\mathbb{N}}

\def\##1\#{\begin{align}#1\end{align}}
\def\$#1\${\begin{align*}#1\end{align*}}

\begin{document}	
%\title{Estimating the Number of Factors via Adaptive Eigenvalue Thresholding with Missing Data}

\title{\Large Missingness-Adaptive Factor Identification in High-Dimensional Data} 
% Identifying the Latent Structure: Factor Discovery in High-Dimensional Data with Missing Observations}
%Missingness-Adaptive Thresholding Estimator for the\\ Number of Factors}%option 2

\author{Ping Zeng\thanks{School of Finance and Statistics, Hunan University, Changsha, China; email: \texttt{\color{blue} zengpingapple@gmail.com}}\and 
Yicheng Zeng\thanks{School of Science, Sun Yat-sen University, Shenzhen, China; email: \texttt{\color{blue} zengyicheng@mail.sysu.edu.cn}}\ \footnotemark[4]\and 
Lixing Zhu\thanks{Department of Statistics, Beijing Normal University at Zhuhai, Zhuhai, China; email: \texttt{\color{blue} lzhu@bnu.edu.cn}}\  
\thanks{Yicheng Zeng and Lixing Zhu are co-corresponding authors.}
}

\maketitle

\begin{abstract}

Determining the number of factors in high-dimensional factor models remains a fundamental challenge, particularly when data are incomplete. This paper introduces the concept of identifiable factors—those that can be reliably recovered despite missing observations—and proposes the Missingness-Adaptive Thresholding Estimator (MATE). To our knowledge, MATE is the first missingness-adaptive framework for factor number determination that accommodates both homogeneous and heterogeneous missingness without imposing restrictive assumptions on factor strength. Notably, it operates without data imputation, circumventing the computational burden associated with most existing approaches. We establish a rigorous theoretical foundation for MATE, proving its consistency under a range of structural conditions. Extensive simulations and real-world applications demonstrate that MATE consistently outperforms state-of-the-art methods, exhibiting superior robustness in settings with high missingness rates and weak factor signals.

\end{abstract}


\newpage
\section{Introduction}\label{Introduction}

Over the past three decades, high-dimensional data—where the number of features $d$ is comparable to or exceeds the sample size $n$—have become ubiquitous across fields ranging from bioinformatics and image processing to finance and social networks. While high dimensionality offers rich information, it also introduces substantial statistical challenges, primarily due to noise accumulation that can degrade conventional estimators. Dimension reduction is therefore essential, with latent factor models serving as a cornerstone for tasks such as asset return analysis, macroeconomic forecasting, and policy evaluation \citep{stock1989new, hsiao2012panel}. A fundamental problem in factor modeling is determining the number of latent factors. Existing literature generally falls into two categories: information criteria-based methods \citep{bai2002determining, su2017time} and eigenvalue-ratio or hypothesis-testing approaches \citep{ahn2013eigenvalue, lam2012factor}. Most of these methods rely on a strong factor assumption, where “spiked” eigenvalues (driven by factors) are significantly larger than the noise floor. However, in many modern applications, factors may be “weak”, explaining only a small fraction of the total variance. In such cases, existing methods often fail to distinguish signal from noise, leading to systematic underestimation of the factor rank.

A third, more recent line of research leverages Random Matrix Theory (RMT) to address both strong and weak factor regimes by identifying specific eigenvalue thresholds \citep{onatski2010determining, zeng2022order, zeng2023order, fan2022estimating, ke2023estimation}. These methods typically assume the high-dimensional asymptotic regime where $d/n \to \gamma \in (0, \infty)$, a setting in which the sample covariance matrix may not consistently estimate its population counterpart \citep{johnstone2009consistency}. While RMT-based methods are robust to varying factor strengths, they almost universally assume complete data. In practice, missing data are pervasive and often unavoidable. In longitudinal surveys such as the China Family Panel Studies, missingness arises from participant attrition; in finance, it stems from mixed-frequency data or non-overlapping trading hours. In high dimensions, simple strategies like listwise deletion are catastrophic: for a $d \times n$ matrix with a mere $1\%$ missing rate, approximately $94\%$ of rows contain no missing values when $d=6$, yet the probability of a row being complete drops below $1\%$ as $d$ approaches $460$. Despite the ubiquity of this problem, methods for determining the factor rank under missingness remain underdeveloped. This paper bridges the gap between RMT-based factor estimation and incomplete data analysis.

Our contributions are fourfold. First, we introduce a conceptual framework by defining identifiable factors—latent components that can be theoretically recovered under specific missingness constraints. Second, on the methodological front, we propose the Missingness-Adaptive Thresholding Estimator (MATE), an imputation-free, eigenvalue-based procedure that automatically adapts to both homogeneous and heterogeneous missingness patterns and offers improved computational efficiency. Third, we establish theoretical rigor by proving the consistency of MATE in the $d/n \to \gamma$ regime, demonstrating its effectiveness across the full spectrum of factor strengths. Fourth, we provide empirical validation through extensive simulations and real-data applications, showing that MATE consistently outperforms state-of-the-art approaches—particularly in the challenging “double-threat” scenario characterized by high missing rates and weak factor signals.

\paragraph{Organization} The remainder of the paper is organized as follows. Section~\ref{model} introduces the model setting under the MCAR mechanism. Section~\ref{method} presents MATE and its theoretical properties. Section~\ref{Simulation} reports simulation studies, while Section~\ref{real_data} illustrates real-data applications. Section~\ref{Conclusion} concludes with discussions. Auxiliary lemmas and technical proofs are provided in the supplementary material.

\section{Model setting}\label{model}

\paragraph{Notations} 
For $n\in \NN^+$, let $[n]=\{1,\cdots,n\}$, so that $\{m+1,\cdots,n\}=[n]\backslash [m]$. Denote $a\wedge b = \min \{a,b\}$ and $a \vee b=\max\{a,b\}$. For sequences $\{a_n\}$ and $\{b_n\}$, write $a_n\sim b_n$ if $a_n/b_n\rightarrow 1$ in probability\footnote{Including the degenerate case of deterministic convergence.}. 
For any $n\times n$ matrix $A$, let $\lambda_i(A)$ denote its $i$-th largest eigenvalue and ${\rm Spec}(A)$ its spectrum, and define $\beta_k(A)=\frac{1}{n}{\rm tr}(A^k)$ for $k\in \mathbb{N}$. 
For any set $A$, let $\sharp A$ denote its cardinality. For any probability measure $\rho$, let ${\rm supp}(\rho)$ denote its support.  
In general, the superscript $^o$ denotes quantities under missingness; for example, $S^o_n$ is the sample covariance matrix with missing data. 

\subsection{High-dimensional factor model}
Consider a dataset $\{x_i\}_{1\le i\le n}\subset \mathbb{R}^d$ with $x_i$ generated from the factor model 
\begin{equation}\label{factor_model}
x_i = \Lambda f_i + \varepsilon_{i}\in\mathbb{R}^d,\ i\in [n],
\end{equation}
where  $x_i$ is a $d$-dimensional observed data point, $f_i$ is an $r$-dimensional latent common factor, $\Lambda\in \mathbb{R}^{d\times r}$ is the loading matrix, and $\varepsilon_{i}$ is a $d$-dimensional error independent of $\{f_i\}$. Let $F=(f_1,\cdots,f_n) \in \mathbb{R}^{r\times n}$ be the factor matrix, $X=(x_1,\cdots,x_n) \in \mathbb{R}^{d\times n}$ the data matrix and $\varepsilon=(\varepsilon_1,\cdots,\varepsilon_n) \in \mathbb{R}^{d\times n}$ the noise matrix. In matrix form, \eqref{factor_model} can be written as 
\begin{equation}\label{factor_model_matrix}
X = \Lambda F + \varepsilon.
\end{equation}
We first provide some regularity conditions to  specify the high-dimensional factor model.

\begin{assumption}\label{assumption_HD}
The ratio $\gamma_n:=d/n \to \gamma \in (0,\infty)$ as $n \to \infty $.
\end{assumption}

Hereafter, we use $\gamma$ to denote the aspect ratio and do not discuss the convergence rate of $\gamma_n$ to $\gamma$.

\begin{assumption}\label{assumption_factor}
The loading matrix $\Lambda$ has fixed rank $r$ with $r<d$. The entries $f_{ji}$'s of the factor matrix $F\in \mathbb{R}^{r\times n}$ are i.i.d. and $f_{ji}\sim \mathcal N (0,1)$ for $j\in [r]$ and $i\in [n]$.  
\end{assumption}

\begin{assumption}\label{assumption_varepsilon}
The noise vectors $\{\varepsilon_i\}_{1\le i\le n}$ are i.i.d. Write the covariance matrix $\Sigma_\varepsilon:=\cov(\varepsilon_i)\in \mathbb{R}^{d\times d}$. 
There exists a random matrix $e= (e_1,\cdots,e_n)=[e_{ji}]\in\mathbb{R}^{d\times n}$ such that $\varepsilon_i=\Sigma_\varepsilon^{1/2}e_i$ and $e_i\sim \mathcal N(0,I_d)$ for $i\in [n]$.  %\zcomment{do we really need Gaussian?}
\end{assumption}
    
With model \eqref{factor_model} and Assumptions~\ref{assumption_factor} and \ref{assumption_varepsilon}, the population covariance matrix $\Sigma\in\mathbb{R}^{d\times d}$ of  $x_i$ has the formula as 
\begin{equation}\label{data_covariance}
\Sigma=\Lambda\Lambda^\top+\Sigma_\varepsilon=U D U^\top, 
\end{equation}
where $U\in \mathbb{R}^{d\times d}$ is an orthogonal matrix ($U^\top U=I_d$), and $D={\rm diag}(\lambda_1,\cdots,\lambda_d)\in \mathbb{R}^{d\times d}$ with $\lambda_i$ being the $i$-th largest eigenvalue of $\Sigma$. 

\begin{definition}
For any {\it Hermitian} matrix $A\in \mathbb{R}^{d\times d}$, its \emph{empirical spectral distribution (ESD)} is $F_d^A(x)=d^{-1}\sum\nolimits_{i=1}^d\mathbf 1\{\lambda_i(A)\le x\}$ for all $x\in \mathbb{R}$, where $\lambda_i(A)$ denotes the $i$-th largest eigenvalue of $A$, and $\mathbf 1\{\cdot\}$ denotes the indicator function. 
\end{definition}
    
The following assumption on $\Sigma_\varepsilon$ is standard in the random matrix theory (RMT) literature.

\begin{assumption}\label{assumption_Sigma_x}
The ESD of $\Sigma_\varepsilon$ weakly converges to a nonrandom probability measure $\pi_\varepsilon(t)$ as $n \to \infty$, which is referred to as the limiting spectral distribution (LSD) of $\Sigma_\varepsilon$.
\end{assumption}

\subsection{Spiked structure}
Inspired by \cite{cai2020limiting} and with Assumption~\ref{assumption_varepsilon}, we can rewrite the factor model \eqref{factor_model_matrix} as 
\begin{equation}\label{x_decomposition}
X = \Lambda F +\Sigma_\varepsilon^{1/2} e
= \Gamma \widetilde Y 
\end{equation}
with $\Gamma := (\Lambda,\Sigma_\varepsilon^{1/2})\in\mathbb{R}^{d\times (r+d)}$ and $\widetilde Y := (F^\top, e^\top)^\top\in\mathbb{R}^{(r+d)\times n}$. 
We observe that $\Sigma=\Gamma \Gamma^\top$. Therefore, $\Gamma\in\mathbb{R}^{d\times (r+d)}$ admits the singular value decomposition  
\begin{equation}\label{gamma_decomposition}
\Gamma= U D^{1/2} V^\top,
\end{equation}
where $V\in \mathbb{R}^{(r+d)\times d}$ is an orthogonal matrix ($V^\top V=I_d$). 

Since columns $x_i$'s of $X$ have covariance matrix $\Sigma$, we can express $X$ with a separable structure as  
\begin{equation}\label{X_decomposition}
X=\Sigma^{1/2} (\Sigma^{-1/2}X)=: \Sigma^{1/2}  Y,
\end{equation}
where $Y =U V^\top \widetilde Y\in\mathbb{R}^{d\times n}$ by \eqref{x_decomposition} and \eqref{gamma_decomposition}. The matrix $Y$ has i.i.d. entries $y_{ij}{\sim }\mathcal N(0,1)$, which follows from the orthogonality of $UV^\top$ and the Gaussianity of $\widetilde Y$ assumed in Assumptions~\ref{assumption_factor} and \ref{assumption_varepsilon}. 

Consider the decomposition 
\$
\Sigma=\Lambda\Lambda^\top+\Sigma_\varepsilon=U_1 D_{\Lambda} U_1^\top + U_2 D_{\varepsilon} U_2^\top,
\$
where $D_{\Lambda}={\rm diag}(\rho_1,\cdots,\rho_r,0, \cdots,0)\in \mathbb{R}^{d\times d}$, $D_{\varepsilon}={\rm diag}(\sigma_1^2,\cdots,\sigma_d^2)\in\mathbb{R}^{d\times d}$, and $U_1\in \mathbb{R}^{d\times d}$ and $U_2\in \mathbb{R}^{d\times d}$ are orthogonal matrices. When $\Sigma_\varepsilon$ is isotropic, i.e., $\Sigma_{\varepsilon}=\sigma^2 I_d$ for some $\sigma^2>0$, by \eqref{data_covariance}, the spectrum of $\Sigma$ becomes  
\$
{\rm Spec}(\Sigma)=\{\lambda_1,\cdots,\lambda_{r},\sigma^2,\cdots,\sigma^2\},
\$ 
where $\lambda_i=\rho_i+\sigma^2$ for $1\le i\le r$. This coincides with the {\it spiked model} proposed by \cite{johnstone2001distribution}. When $\Sigma_\varepsilon$ is anisotropic, if $\Lambda \Lambda^\top$ and $\Sigma_\varepsilon$ are simultaneously diagonalizable in the sense that $U_1=U_2$, then $\lambda_i=\rho_i+\sigma_i^2$ for $1\le i\le r$ and $\lambda_i=\sigma_i^2$ for $r+1\le i \le d$. If the top eigenvalues $\{\lambda_i\}_{1\le i\le r}$ exceed some phase transition threshold determined by the LSD of $\Sigma_\varepsilon$, $\Sigma$ follows a {\it generalized spiked model} \citep{bai2012sample}.

\subsection{Identifiable factors}

Let $S_n \in \mathbb{R}^{d\times d}$ denote the sample covariance matrix for the data matrix $X$, defined by 
\$
S_n=\frac{1}{n}XX^{\top}=\frac{1}{n}\Sigma^{1/2}{Y}{Y}^{\top}\Sigma^{1/2},
\$
with eigenvalues $\hat{\lambda}_1 \ge \cdots \ge \hat{\lambda}_d$.  

Under the high-dimensional regime in Assumption \ref{assumption_HD}, \cite{marchenko1967distribution} established that the ESD of $S_n$ weakly converges to the {\it Marchenko-Pastur} (M-P) distribution $F_{\gamma,\sigma^2}$ when $\Sigma=\sigma^2 I_d$. This convergence remains valid when $\Sigma = \Lambda \Lambda^\top + \sigma^2 I_d$ with $\Lambda$ satisfying Assumption~\ref{assumption_factor}, since the finite-rank perturbation $\Lambda \Lambda^\top$ does not affect the LSD of $S_n$. Specifically, almost surely, $F_d^{S_n}(x) \to F_{\gamma,\sigma^2}(x),\ \forall x\in\mathbb{R}$. 
When $0< \gamma \le 1$, the limiting distribution $F_{\gamma,\sigma^2}$ has density 
\$
f_{\gamma,\sigma^2}(x)=\frac{1}{2\pi x \gamma \sigma^2} \sqrt{(\lambda_+ -x)(x-\lambda_-)}, \  \forall \lambda_-< x <  \lambda_+,
\$ 
where $\lambda_-=\sigma^2(1-\sqrt{\gamma})^2, \lambda_+=\sigma^2(1+\sqrt{\gamma})^2$. For $\gamma>1$, the distribution includes an additional point mass $1-\gamma^{-1}$ at zero, since the integral of $f_{\gamma,\sigma^2}$ over $(\lambda_-,\lambda_+)$ equals $\gamma^{-1}$ (e.g., see \cite{yao2015sample}).  

For the extreme eigenvalues of $S_n$, \cite{baik2006eigenvalues} show that, almost surely,  
    \begin{equation}\label{baik_2006}
    \hat \lambda_i 
    \rightarrow 
        \begin{cases}
        \sigma^2\psi(\lambda_i/\sigma^2), & i\in [r_0],\\
        \sigma^2(1+\sqrt{\gamma})^2, & i\in [C]\backslash [r_0],
        \end{cases}
    \end{equation}
for any fixed $r_0< C <d$, where $\psi(\alpha):=\alpha+\gamma \alpha(\alpha-1)^{-1}$ for $\alpha \neq 1$, and $r_0$ represents the number of {\bf identifiable spikes}, defined by
\begin{equation}\label{identifiable_spikes}
r_0 =\sharp\left\{i\in [d]: \lambda_i>\sigma^2(1+\sqrt{\gamma})\right\}.
\end{equation}
The {\bf identification threshold} $\sigma^2(1+\sqrt{\gamma})$ is also known as the {\it BBP phase transition threshold} \citep{baik2005phase}, which specifies the minimum signal strength required for a population spike to be detectable in the sample spectrum.  

\subsection{Missing data}
There are three standard types of missingness: \emph{missing completely at random} (MCAR), \emph{missing at random} (MAR), and \emph{missing not at random} (MNAR) \citep{little2019statistical,seaman2013meant}. 
Under MCAR, missingness is independent of the data values; under MAR, it depends only on observed components; under MNAR, it depends on unobserved components. 
MCAR is widely assumed for missingness in factor-model studies \citep{jin2021factor,stock2016dynamic,bai2015unbalanced}. We study this scenario throughout the paper. 

\subsubsection{Formulation of MCAR}

Let $\Omega = \{(i,j)\in[d] \times [n] : x_{ij} ~\text{is observed}\}$ denote the index set of observed entries of $X=[x_{ij}]\in \mathbb{R}^{d\times n}$.   
Define $B=[b_{ij}]\in \mathbb{R}^{d\times n}$ with $b_{ij}=\mathbf 1\{(i,j)\in \Omega\}$. 
The partially observed data matrix is 
\begin{equation}\label{function_X^o}
X^o:= X \circ B,
\end{equation}
where the notation `` $\circ$ " denotes the {\it Hadamard} product, so $x^o_{ij}=x_{ij}b_{ij}$ for $(i,j)\in[d] \times [n]$. Assume $b_{ij}\sim {\rm Bernoulli}(p_{ij})$, independently of $\{x_{ij}\}$, $\{f_{ik}\}$, $\{\lambda_{jk}\}$ and $\{\varepsilon_{ij}\}$. If $p_{ij}=p$ for some $p\in (0,1]$, the missingness is {\it homogeneous}; otherwise, {\it heterogeneous}.  

In datasets with heterogeneous missingness, two types are common: sample-specific and feature-specific. For example, in recommendation systems with consumers (samples) as columns and items (features) as rows, sample-specific missingness occurs when some consumers rate few or no items, while feature-specific missingness arises when certain items lack evaluations across multiple consumers. 

Heterogeneous missingness can be formalized by partitioning $B$ in two ways: 
\begin{align}\label{B_block}
B&=(B_{1 \cdot}^\top, \cdots, B_{K \cdot}^\top)^\top\     
{\rm or}\ 
B =( B_{\cdot 1}, \cdots, B_{\cdot L} ) ,
\end{align}
where $B_{k \cdot} \in \mathbb{R}^{d_k\times n}$ with $\sum_{k=1}^K d_k=d$ and $B_{\cdot \ell} \in \mathbb{R}^{d \times n_\ell}$ with $\sum_{\ell=1}^L n_\ell =n$. The blocks $\{B_{k \cdot}\}$ and $\{B_{\cdot \ell}\}$ correspond to feature-specific and sample-specific missingness, respectively. We assume that the blocks are independent, with   
\begin{equation}\label{b_distribution}
b_{ij}\stackrel{i.i.d.}{\sim} {\rm Bernoulli}(p_k),\ b_{ij}\in B_{k \cdot},\ k\in [K],\  {\rm or}\    
b_{ij}\stackrel{i.i.d.}{\sim} {\rm Bernoulli}(q_\ell),\ b_{ij}\in B_{\cdot \ell },\ \ell\in [L]. 
\end{equation}
A homogeneous setting corresponds to $K=1$ or $L=1$. 

\subsubsection{Existing methods for handling missing data}\label{missing_data_review}

A straightforward strategy for handling missing data is to delete rows or columns containing missing values, but this greatly reduces the effective sample size and harms factor estimation in high dimensions \citep{ludvigson2007empirical}. 
Imputation offers an alternative, yet sample mean or median filling often introduces bias. Recent factor-based imputation methods, such as TALL-WIDE \citep{bai2021matrix}, TALL-PROJECT \citep{cahan2023factor}, and weighted covariance imputation \citep{xiong2023large}, require a consistent estimator of the number of factors; otherwise, the factor model cannot be identified. 
Another line of work jointly estimates the factor model and imputes missing values, including {\it Kalman Filtering} \citep{giannone2008nowcasting,doz2011two} and EM-based approaches \citep{stock2002macroeconomic,banbura2014maximum,jin2021factor}. 
These two methods are computationally intensive and become infeasible in high dimensions, causing estimation accuracy reduction for weak factors.

\section{Estimating the number of factors}\label{method}

In this section, we consider multiple settings involving both the noise $\varepsilon$ and the missingness $B$. Under the factor model \eqref{X_decomposition} and the random missingness \eqref{function_X^o}, the observed data is  
\begin{equation}\label{X^o_decomposition}
X^o = X \circ B 
= (\Sigma^{1/2} Y)\circ B.
\end{equation}
Let $D_1=(\lambda_1, \cdots, \lambda_d)$. When $\Sigma_{\varepsilon}=\sigma^2 I_d$, \eqref{data_covariance} implies that 
\begin{equation}\label{Sigma_block}
\Sigma=U {\rm Diag}\left(D_1, \sigma^2 I_{d-r}\right)U^{\top}.
\end{equation}

\subsection{Isotropic noise}
We begin with the isotropic case $\Sigma_\varepsilon =\sigma^2 I_d$ and consider both homogeneous and heterogeneous missingness. In the homogeneous setting, $K=1$ and $L=1$ are equivalent, so we set $K=1$. For the heterogeneous setting, we consider both feature-specific and sample-specific missingness. 

\subsubsection{A warm-up case: homogeneous missingness}
Homogeneous missingness corresponds to $K=1$ in \eqref{B_block}. For a diagonal $\Sigma$, using \eqref{X^o_decomposition} and Lemma~\ref{lemma_inner_hadamard}, we write 
\begin{equation}\label{X^o_decomposition2}
X^o = \Sigma^{1/2} (Y\circ B)
= : \Sigma^{1/2}  Y^o. 
\end{equation}
When $K=1$, \eqref{b_distribution} gives $b_{ij}\stackrel{i.i.d.}{\sim} {\rm Bernoulli}(p)$ for some $p\in (0,1]$, so $\EE(b_{ij})=p$ and ${\rm var}(b_{ij})=p(1-p)$. Independence of $b_{ij}$ implies that the entries of $Y^o$ are i.i.d. with mean zero and variance $p$. The sample covariance matrix associated with $X^o$ is 
\begin{equation}\label{spectrum_equivalence}
S_n^o = \frac{1}{n}X^o X^{o\top}
= \frac{1}{n} \Sigma^{1/2} Y^o Y^{o\top} \Sigma^{1/2}.
\end{equation}

\begin{remark}
The spectra of sample covariance and inner product matrices under missing data have been studied in prior works. For example, \cite{jurczak2017spectral} study the eigenvalues of sample covariance matrices under missingness; \cite{nadakuditi2014optshrink} and \cite{dobriban2020optimal} treat missing data as a special case of the signal-plus-noise model; and \cite{yu2025theory} study a spectrum shift phenomenon under MCAR and propose an eigenvalue correction method for inner product matrices. 
\end{remark}

We define the number of identifiable spikes under missing data as $r_1$. From \eqref{X^o_decomposition2} and \eqref{spectrum_equivalence}, $p^{-1}S_n^o$ is the sample covariance of $p^{-1/2}\Sigma^{1/2}Y^o$ and thus has the same LSD and spike limits as $S_n$. 
In particular, homogeneous missingness does not alter the identifiability of spikes of $\Sigma$: 
\#\label{homogeneous_identifiable_spikes}
r_1=\sharp \{i\in [d]: p \lambda_i > p\sigma^2 (1+\sqrt \gamma)\}
\# 
Therefore, $r_1=r_0$, where $r_0$ is defined in \eqref{identifiable_spikes}, implying that $S_n^o$ and $S_n$ identify the same spikes. 

Using the scaling $p^{-1} S_n^o \mapsto S_n^o$ and results from \eqref{baik_2006} and \eqref{spectrum_equivalence}, the eigenvalues of $S_n^o$ satisfy  
\$
\hat \lambda_i^o \overset{\rm a.s.}{\longrightarrow} 
\begin{cases}
p\sigma^2\psi(\lambda_i/\sigma^2), & i\in [r_1],\\
p\sigma^2(1+\sqrt{\gamma})^2, & i\in [C]\backslash [r_1],
\end{cases}
\$ 
where $\hat \lambda_i^o$ is the $i$-th largest eigenvalue of $S_n^o $, $\psi(\cdot)$ is defined in \eqref{baik_2006}, and $C\ge r_1$ is fixed.  

We suggest a thresholding-based approach to determine the number of factors. Note that classical scree-plot methods \citep{cattell1966scree} and their high-dimensional counterparts \citep{fan2022estimating, ke2023estimation} identify the factor count by locating a ``gap" or significant drop in ordered empirical eigenvalues. These techniques are fundamentally designed for complete datasets. They typically fail to account for the structural distortions introduced by missing observations.  A standard alternative for handling incomplete blocks is to first reconstruct the data via imputation before applying traditional estimators. However, as detailed in Section~\ref{missing_data_review}, imputation often imposes a heavy computational burden and risks propagating reconstruction errors into the downstream eigenvalue analysis. 
To address these issues, we introduce a missingness-adaptive identification threshold. Unlike traditional fixed-cutoff methods, our approach explicitly incorporates the missingness structure into the threshold calculation.
Therefore, our method is imputation-free; it operates directly on the incomplete data to maintain both computational efficiency and estimation integrity. Motivated by these requirements, we propose the Missingness-Adaptive Thresholding Estimator (MATE) for estimating $r_1$:
\begin{equation}\label{MATE}
\hat{r}(v,\epsilon_n):=\sharp \left\{i\in [d]: \hat{\lambda}_i^o > v +\epsilon_n\right\},
\end{equation}
where $v$ is referred to as the {\bf missingness-adaptive identification threshold}, and $\epsilon_n=o(1)$ provides robustness against estimation error. 

The following theorem specifies how to choose $(v,\epsilon_n)$ and establishes the consistency of $\hat{r}(v,\epsilon_n)$. 

\begin{theorem}\label{consistency}
Consider the factor model in \eqref{factor_model_matrix}, \eqref{X^o_decomposition2} and Assumptions \ref{assumption_HD}-\ref{assumption_Sigma_x}. Consider $K=1$ and $r_1$ in \eqref{homogeneous_identifiable_spikes}. Then, for any $\epsilon_n=o(1)$ satisfying $\epsilon_n n^{2/3}\to \infty$, $\lim\limits_{n\to \infty}\mathbb{P} \{\hat{r}(p\sigma^2(1+\sqrt \gamma)^2,\epsilon_n)=r_1\} = 1$.  
\end{theorem}

\begin{remark}\label{optimal_threshold}
$v=p \sigma^2(1+\sqrt \gamma)^2$ is the optimal threshold in the following sense: for any $\delta>0$ and $\epsilon_n=o(1)$, there exists some covariance matrix $\Sigma$ such that $\lim_{n\rightarrow \infty}\mathbb{P} \{\hat r(p\sigma^2(1+\sqrt \gamma)^2-\delta,\epsilon_n)=r_1\}<1$ and $\lim_{n\rightarrow \infty}\mathbb{P} \{\hat r(p\sigma^2(1+\sqrt \gamma)^2+\delta,\epsilon_n)=r_1\}<1$.  
\end{remark}

\subsubsection{Feature-specific heterogeneous missingness}\label{subsec_feature-specific_hete}

Consider feature-specific heterogeneous missingness:
\begin{equation*}
\mathbb{P} (b_{ij})=p_k,\ \forall i\in \mathcal I_k,\ k\in [K], 
\end{equation*}
where the index sets $\{\mathcal I_k\}_{1\le k\le K}$ form a partition of $\mathcal I =[d]$ according to the row blocks $\{B_{k \cdot} \}_{1\le k\le K}$, that is, $\mathcal I = \mathcal I_1\cup \cdots \cup \mathcal I_K$. The cardinality of each block is denoted by $d_k=|\mathcal I_k|$. Without loss of generality, we take $\mathcal I_k = \{ d_1+\cdots+d_{k-1}+1,\cdots,d_1+\cdots+d_{k}\}$, $1\le k\le K$. 
Define the block-diagonal matrix 
$
P_d = {\rm Diag}(p_1 I_{d_1},\cdots, p_K  I_{d_K})\in \mathbb{R}^{d\times d}. 
$ 
We impose the following assumption on $P_d$.
\begin{assumption}\label{assumption_P_d}
The ESD of $\Sigma P_d$ weakly converges to a nonrandom probability measure $\pi_{\Sigma,P}$ as $d \to \infty$.  
\end{assumption}

Given these notations and formula \eqref{X^o_decomposition2}, the observed data matrix $X^o$ can be expressed as 
\begin{equation}\label{feature_specific_data}
X^o=\Sigma^{1/2} Y^o = \Sigma^{1/2} P_d^{1/2}P_d^{-1/2}  Y^o,
\end{equation} 
and the associated sample covariance matrix is 
\$
S^o_n = \frac{1}{n} \Sigma^{1/2} P_d^{1/2}(P_d^{-1/2}  Y^o)( P_d^{-1/2} Y^o)^{\top}P_d^{1/2}\Sigma^{1/2}.  
\$ 
It can be verified that the entries of $P_d^{-1/2}  Y^o$ are i.i.d. with zero mean and unit variance. 

Under Assumptions \ref{assumption_HD} and \ref{assumption_P_d}, Theorem~2.14 of \cite{yao2015sample} implies that the ESD of ${S_n^o}$ weakly converges to a nonrandom probability measure $\rho_{\gamma}$, and its {\it Stieltjes transform} $s_\gamma(z)=\int(x-z)^{-1}\rho_\gamma({\rm d}x)$, $z\in \mathbb C^+$, solves the equation 
\$
s_\gamma = - \left(z- \gamma \int \frac{t}{1+t s} {\rm d} \pi_{\Sigma,P} (t) \right)^{-1}.
\$
Furthermore, \cite{bai2012sample} show that ${\rm supp}(\rho_\gamma)=\{\psi(\alpha):\alpha\neq 0, \alpha\notin {\rm supp}(\pi_{\Sigma,P}), \psi'(\alpha)>0\}$, where $\psi'(\alpha)$ denotes the derivative of 
$
\psi (\alpha) = \alpha + \gamma \int \alpha t(\alpha -t)^{-1} {\rm d}\pi_{\Sigma,P}(t).
$ 

In the following, we take $(K,\sigma^2)=(2,1)$ to illustrate the rightmost edge $\lambda_+$ of ${\rm supp}(\rho_\gamma)$. 

\begin{example}\label{example1}
Take $K=2$, $d_1=d_2=d/2$, $\gamma=\gamma_1=0.5$ and $\Sigma_\varepsilon=I_d$. In this setting, the LSD of $\Sigma P_d$ is  
$
\pi_{\Sigma,P}=1/2\left(\delta_{\{p_1\}}+\delta_{\{p_2\}}\right),
$ 
and we have $\psi(\alpha) =\alpha+ \alpha/4[p_1(\alpha-p_1)^{-1}+p_2(\alpha-p_2)^{-1}]$, $
\psi^{'}(\alpha) =1-1/4[ p_1^2(\alpha-p_1)^{-2} + p_2^2(\alpha-p_2)^{-2}]$. Solve $\psi^{'}(\alpha)=0$ and then take the maximum value of $\psi (\alpha)$ among the solutions, which gives the rightmost edge $\lambda_+^{(1)}$ of ${\rm supp}(\rho_{\gamma_1})$. In the left panel of Figure~\ref{example1_picture}, $\lambda_+^{(1)}$ is plotted as a function of $(p_1, p_2)\in [0.1,1]^2$, showing a clearly nonplanar surface. The purple contour lines represent the level sets of $\lambda_+^{(1)}$, and the orange line corresponds to the homogeneous case $p_1=p_2$, along which $\lambda_+^{(1)}$ increases linearly with $p_1$ and $p_2$. Notably, when $(p_1,p_2)=(0.8,0.4)$, the overall nonmissing rate\footnote{non-missing rate refers to the rate of observed values.} is $p=(p_1+p_2)/2=0.6$, yet $\lambda_+^{(1)}$ differs from its value at $(p_1,p_2)=(0.6,0.6)$. The right panel presents the same relationship in two dimensions, with white contour lines marking the level sets. These results illustrate how heterogeneous missingness alters $\lambda_+^{(1)}$, distinguishing it from the homogeneous case.
\vspace{-0.5em}   
\begin{figure}[h]
\centering
\caption{The left panel shows the rightmost edge $\lambda_+^{(1)}$ as a function of $(p_1,p_2)$ in $\mathbb{R}^3$, whereas the right panel depicts the same relationship between $\lambda_+^{(1)}$ and $(p_1, p_2)$ in 2D.}
\label{example1_picture}
\begin{minipage}{0.5\textwidth}
\centering
\includegraphics[width=1.15\textwidth]{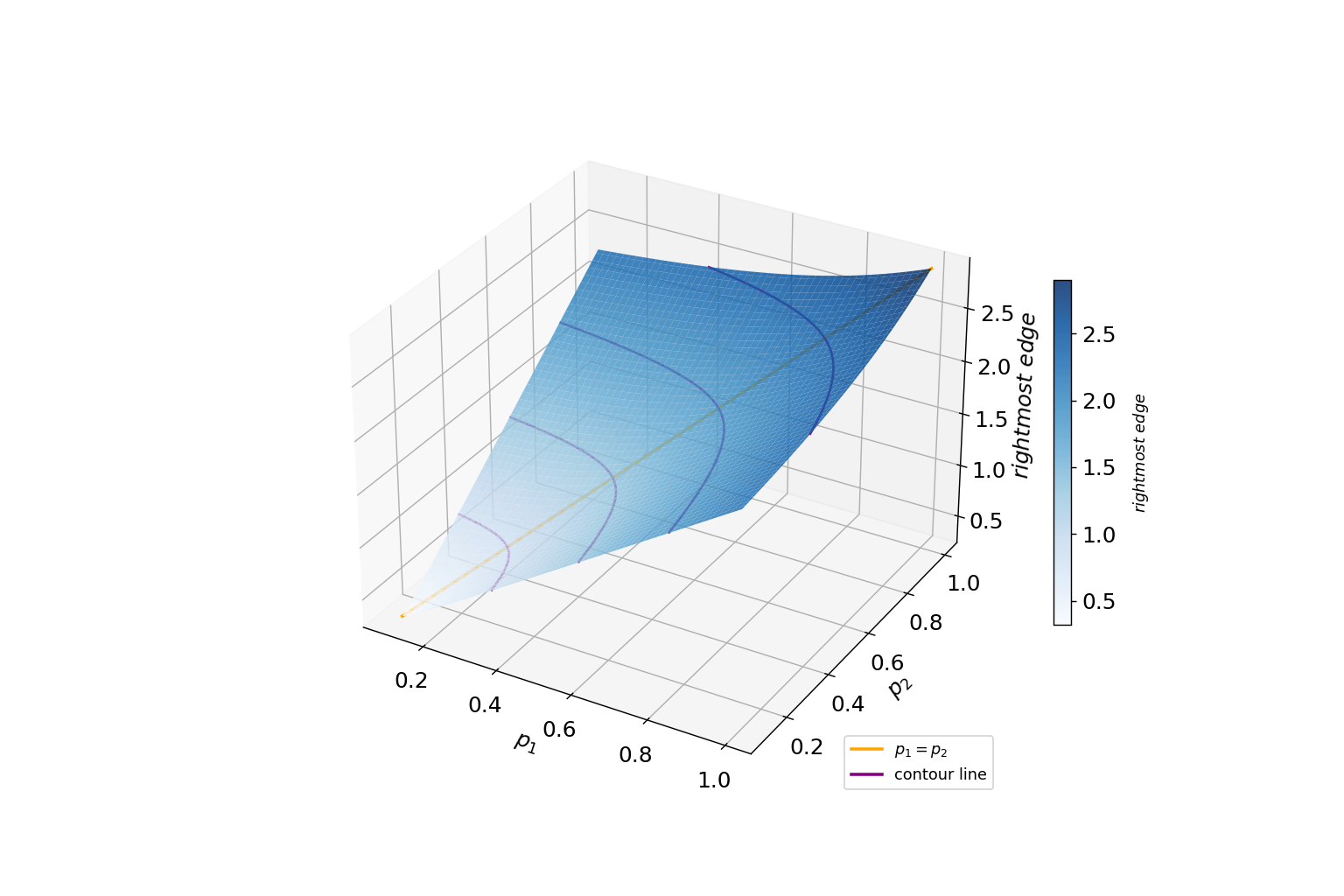}
\end{minipage}
\begin{minipage}{0.48\textwidth}
\centering
\includegraphics[width=\textwidth]{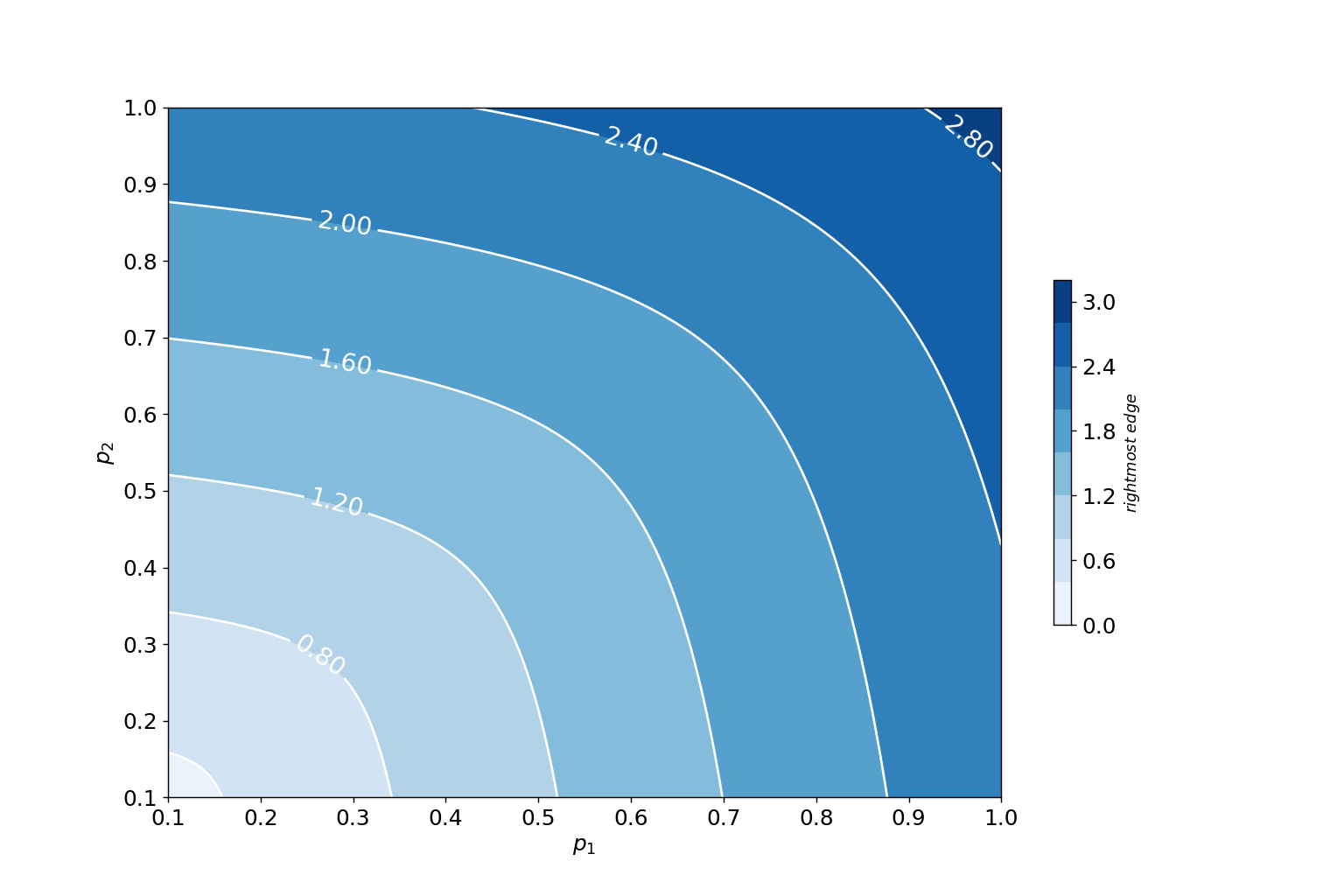}
\end{minipage}
\vspace{-1em}
\end{figure}
\end{example}

We establish the identification condition at the population level:
\#\label{heterogeneous_identifiable_spikes1}
r_1 =\sharp\left\{i\in [d]: \tilde \lambda_i > \alpha_+ \right\}.
\#
where $\tilde \lambda_i$ is the $i$-th largest eigenvalue of $\Sigma P_d$, and $\alpha_+$ is the maximal solution to $\psi'(\alpha)=0$. In general, it is difficult to explicitly express $\tilde \lambda_i$ in terms of $\lambda_i$. At the sample level, we have
\#\label{rightmost_edge_K>1}
\hat \lambda_i^o\overset{P}{\rightarrow}
\begin{cases}
\psi(\tilde \lambda_i),\quad & i\in [r_1],\\
\psi(\alpha_+), \quad & [C]\backslash [r_1],
\end{cases}
\#
We then consider the MATE in \eqref{MATE}. The following theorem specifies the choice of $(v,\epsilon_n)$ under heterogeneous missingness and establishes the consistency of $\hat{r}(v,\epsilon_n)$. 

\begin{theorem}\label{consistency2}
Consider the factor model in \eqref{factor_model_matrix}, \eqref{feature_specific_data} and  Assumptions \ref{assumption_HD}-\ref{assumption_Sigma_x}. Consider $K\ge 2$ and $r_1$ in \eqref{heterogeneous_identifiable_spikes1}. For any $\epsilon_n=o(1)$ satisfying $\epsilon_n n^{2/3}\to \infty$, $\lim\limits_{n\to \infty}\mathbb{P} \{\hat{r}(\psi(\alpha_+),\epsilon_n)=r_1\} = 1$. 
\end{theorem}

$\psi(\alpha_+)$ is the optimal threshold, as characterized in Remark~\ref{optimal_threshold}. The following example shows that heterogeneous missingness may lead to $r_1 < r_0$.

\begin{example}\label{example3}
Take $\gamma=0.25$, and $\Sigma={\rm diag}(5,4,2.4,1,\cdots, 1)$ with $r=3$. Let $K=2$ and $P_d={\rm Diag}(p_1 I_{{d}/{2}},p_2 I_{{d}/{2}})$. When $(p_1,p_2)=(1,1)$, we have $r_0=3$ with identification threshold $1+\sqrt{\gamma}=1.50$. When $(p_1,p_2)=(0.4,0.9)$, the adjusted top eigenvalues are $(\tilde \lambda_1,\tilde \lambda_2,\tilde \lambda_3)=(2,1.6,0.96)$, yielding $r_1=2$ with threshold $\alpha_+=1.223$. 
We then use MATE in \eqref{MATE} to estimate the number of factors in simulations. The results, summarized in Table \ref{example3_figure}, confirm that heterogeneous missingness can reduce the number of identifiable factors. 
\begin{table}[h]
\setlength\tabcolsep{3pt} 
\centering
\footnotesize
\caption{Means (with MSEs) of MATE $\hat r(v,\epsilon_n)$ over $500$ simulations, using $(v,\epsilon_n)=( (1+\sqrt \gamma)^2,0.01)$ for $(p_1,p_2)=(1,1)$ and $(v,\epsilon_n)=(\psi(\alpha_+),0.01)$ for $(p_1,p_2)=(0.4,0.9)$.}\label{example3_figure}
\begin{tabular}{cccc}
\hline
\makebox[0.1\textwidth][c]{$(d,n,p_1,p_2)$}   & \makebox[0.1\textwidth][c]{Mean} & \makebox[0.25\textwidth][c]{Underestimation rate ($\hat r<r$)} & \makebox[0.25\textwidth][c]{Overestimation rate  ($\hat r>r$)} \\
\hline  
$(250,1000,1,1)$     & $2.97_{(0.03)}$   & $3\%$ & $0\%$    \\
\hline 
$(250,1000,0.4,0.9)$ & $2.01_{(0.99)} $  & $99\%$  & $0\% $  \\
\hline
\end{tabular}
\end{table}
\end{example}

\vspace{-2em}

\subsubsection{Sample-specific heterogeneous missingness}\label{subsec_Sample-specific_hete}

Similar to feature-specific heterogeneous missingness in Section \ref{subsec_feature-specific_hete}, we consider a sample-specific heterogeneous missingness mechanism:
    \begin{equation*}
    \mathbb{P} (b_{ij})=q_{\ell},\ \forall j\in \mathcal J_\ell, \ \ell\in [L],
    \end{equation*}
where $\{\mathcal J_\ell \}_{1\le \ell\le L}$ partition $\mathcal J=[n]$ according to $\{ B_{\cdot \ell}\}_{1 \le \ell \le L}$, with $\mathcal J_\ell = \{ n_1+\cdots+n_{\ell-1}+1,\cdots,n_1+\cdots+n_{\ell}\}$. 
Define 
$ 
Q_n = {\rm Diag}(q_1 I_{n_1},\cdots, q_L  I_{n_L})\in \mathbb{R}^{n\times n}.
$  
We impose the following assumption on $Q_n$. 

\begin{assumption}\label{assumption_Q_n}
The ESD of $ Q_n$ weakly converges to a nonrandom probability measure $\pi_Q$ as $n \to \infty$, and the largest eigenvalue of $Q_n$ converges to the rightmost edge of ${\rm supp}\{\pi_Q\}$. 
\end{assumption}

With these notations, the observed data matrix can be written as
\#\label{sample_specific_data}
X^o=\Sigma^{1/2} Y^o = \Sigma^{1/2} Y^o Q_n^{-1/2}Q_n^{1/2}, 
\#
where the entries of $Y^o Q_n^{-1/2}$ are i.i.d. with zero mean and unit variance. 
Then we rewrite $S_n^o $ as
\$
S_n^o = \Sigma^{1/2}(Y^o Q_n^{-1/2}) Q_n (Y^o Q_n^{-1/2})^\top \Sigma^{1/2}.
\$

Let $\pi_\Sigma$ denote the LSD of $\Sigma$. 
In the isotropic noise case, $\pi_\Sigma = \pi_\varepsilon=\delta_{ \{\sigma^2\}}$. Define $(s_1,s_2)$ as the unique solution to the self-consistent equations: 
\#\label{Stieltjes_transforms}
s_{1}(z)& = \gamma \int \frac{x}{-z\{ 1+x s_{2}(z)\} }\pi_\Sigma ({\rm d}x),\ 
s_{2}(z) =\int \frac{x}{-z\{1+x s_{1} (z)\} } \pi_Q ({\rm d}x).
\#
According to \cite{ding2021separable}, 
\$
s_\gamma (z)=\int \frac{x}{-z\{ 1+x s_{2}(z)\} }\pi_\Sigma ({\rm d}x), 
\$
and $s_\gamma(z)$ determines $\rho_\gamma$ via the inverse formula: 
$
\rho_\gamma (E) =\lim_{\eta\rightarrow0+ }\pi^{-1} {\rm Im}\ s_\gamma (E+i \eta),\ \forall z=E+i\eta \in \mathbb C^+.
$

Define 
\#\label{heter_f}
f(z,s)=-s +\int \frac{x}{-z+x \gamma \int \frac{t}{1+t s}\pi_\Sigma({\rm d}t) }\pi_Q ({\rm d}x).
\#
Then $s_2(z)$ is the unique solution to the equation $f(z,s)=0$ satisfying ${\rm Im}\ s>0$ for $z\in \mathbb C^+$. 

The following result characterizes the support of $\rho_\gamma$. 

\begin{lemma}[Lemma~2.5 in \cite{ding2021separable}]
The support of $\rho_\gamma$ is a union of intervals: 
$
{\rm supp}\ \rho_\gamma \cap (0,\infty) = \bigcup_{i =1}^N [a_{2 i},a_{2 i -1}]\cap (0,\infty),
$ 
where $N\in \NN $ depends on $\pi_\Sigma$ and $\pi_Q$. Moreover, $(x,s)=(a_i,s_2(a_i))$, $1\le i\le 2N$, are real solutions to the equations: $
f(x,s)=0$ and $\frac{\partial f}{\partial s} (x,s)=0$. 
\end{lemma}

The rightmost edge $\lambda_+$ of the support of $\rho_\gamma$ is $a_1$, and the phase transition condition for the emergence of outliers among the top eigenvalues $\{\lambda_i\}_{1\le i\le r}$ of $\Sigma$ is $\lambda_i > s_{2}^{-1}(a_1)$. That is,
\#\label{heterogeneous_identifiable_spikes2}
r_1=\sharp \{i\in [d]:\lambda_i>s_{2}^{-1}(a_1)\},
\#
where $s_2$ is the Stieltjes transform obtained from the system \eqref{Stieltjes_transforms} and extended to $\mathbb R\backslash {\rm supp}(\rho_\gamma)$. More precisely, \citep[Theorem~3.6]{ding2021separable} shows that 
\$
\hat \lambda_i^o\overset{P}{\rightarrow}
\begin{cases}
s_2^{-1}(-\lambda_i^{-1}),\quad & i\in [r_1],\\
a_1, \quad & [C]\backslash [r_1],
\end{cases}
\$
where $s^{-1}_2$ is the inverse function of $s_2$, and $r_1<C<d$ is any fixed integer. 

In the following, we take $(L,\sigma^2)=(2,1)$ to illustrate the rightmost edge of ${\rm supp}(\rho_\gamma)$. 

\begin{example}\label{example2}
Take $L=2$, $n_1=n_2=n/2$, $\gamma=\gamma_2=2$, and $\Sigma_\varepsilon=I_d$. By \eqref{heter_f}, we have 
\$
f(x, s)&=-s+\frac{1}{2}\frac{q_1(1+s)}{-x(1+s)+2 q_1}+\frac{1}{2}\frac{q_2(1+s)}{-x(1+s)+2 q_2}, \\
\frac{\partial f}{\partial s} (x,s)& =-1+\frac{q_1^2}{[-x(1+s)+2 q_1]^2}+\frac{q_2^2}{[-x(1+s)+2 q_2]^2}.
\$ 
Take the maximum $x$ satisfying $f(x, s)=0$ and $\frac{\partial f}{\partial s} (x,s)=0$, which gives the rightmost edge $\lambda_+^{(2)}$ of ${\rm supp}(\rho_{\gamma_2})$. We display $\lambda_+^{(2)}$ as a function of $(q_1,q_2)\in [0.1,1]^2$ in Figure~\ref{example2_picture}. In the left panel, the purple contour lines represent level sets of $\lambda_+^{(2)}$, while the orange line corresponds to the homogeneous case $q_1=q_2$. In the right panel, the white contour lines represent the level sets of $\lambda_+^{(2)}$. 
These plots illustrate how heterogeneous missingness affects $\lambda_+^{(2)}$, distinguishing it from the homogeneous case. Moreover, the right panels of Figures~\ref{example1_picture} and \ref{example2_picture} show that $\lambda_+^{(1)}=\gamma_2^{-1}\lambda_+^{(2)}$, with $\lambda_+^{(1)}$ under feature-specific missingness proportional to $\lambda_+^{(2)}$ under sample-specific missingness. This relationship holds in the isotropic noise case by swapping the roles of dimension $d$ and sample size $n$ of $X^o$, since $\gamma_2=\gamma_1^{-1}$. 
\end{example}

\vspace{-1em}
\begin{figure}[h]
\centering
\caption{The left panel shows the rightmost edge $\lambda_+^{(2)}$ as a function of $(q_1,q_2)$ in $\mathbb{R}^3$, whereas the right panel depicts the same relationship between $\lambda_+^{(2)}$ and $(q_1,q_2)$ in 2D.}
\label{example2_picture}
\begin{minipage}{0.48\textwidth}
\centering
\includegraphics[width=\textwidth]{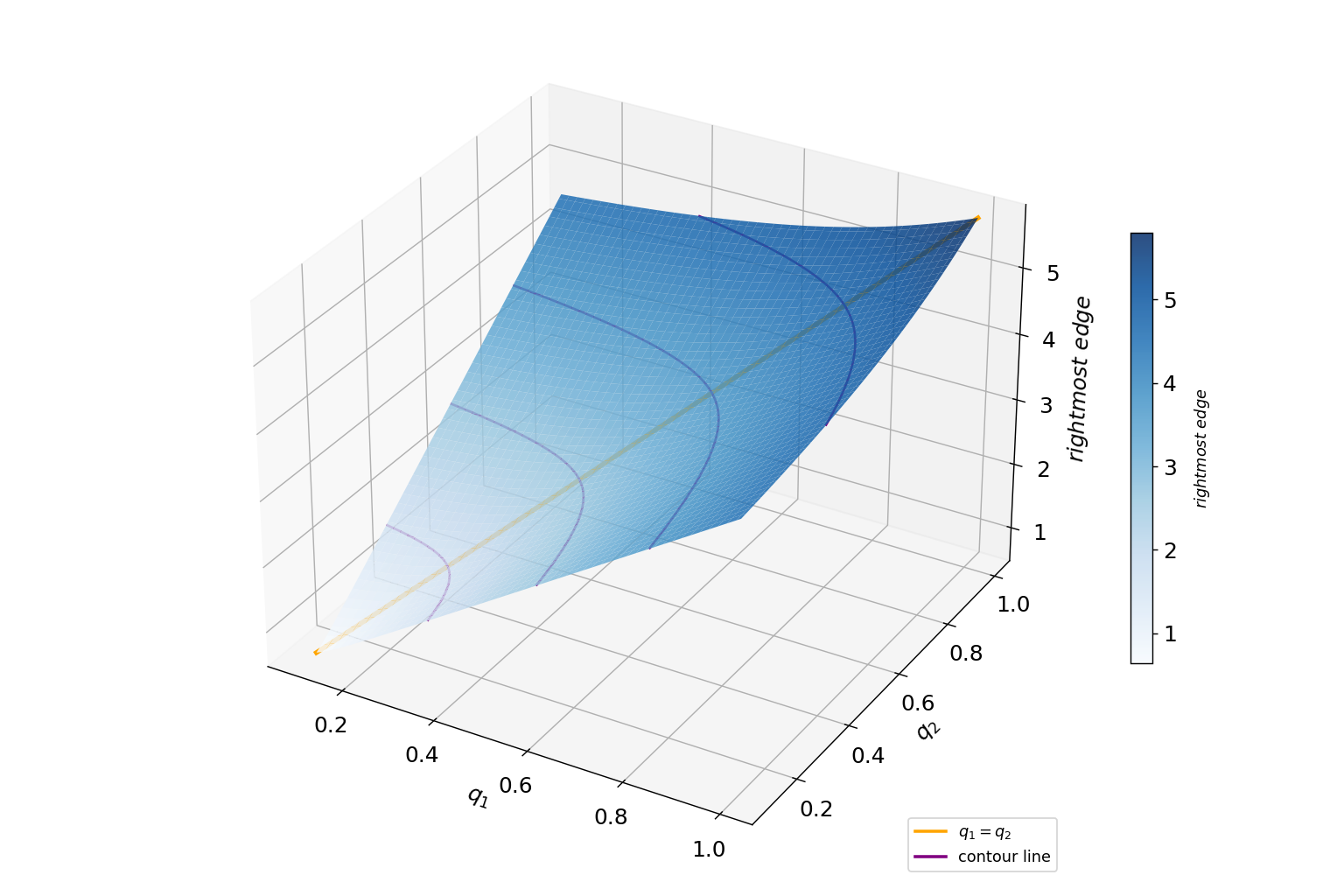}
\end{minipage}
\begin{minipage}{0.48\textwidth}
\centering
\includegraphics[width=\textwidth]{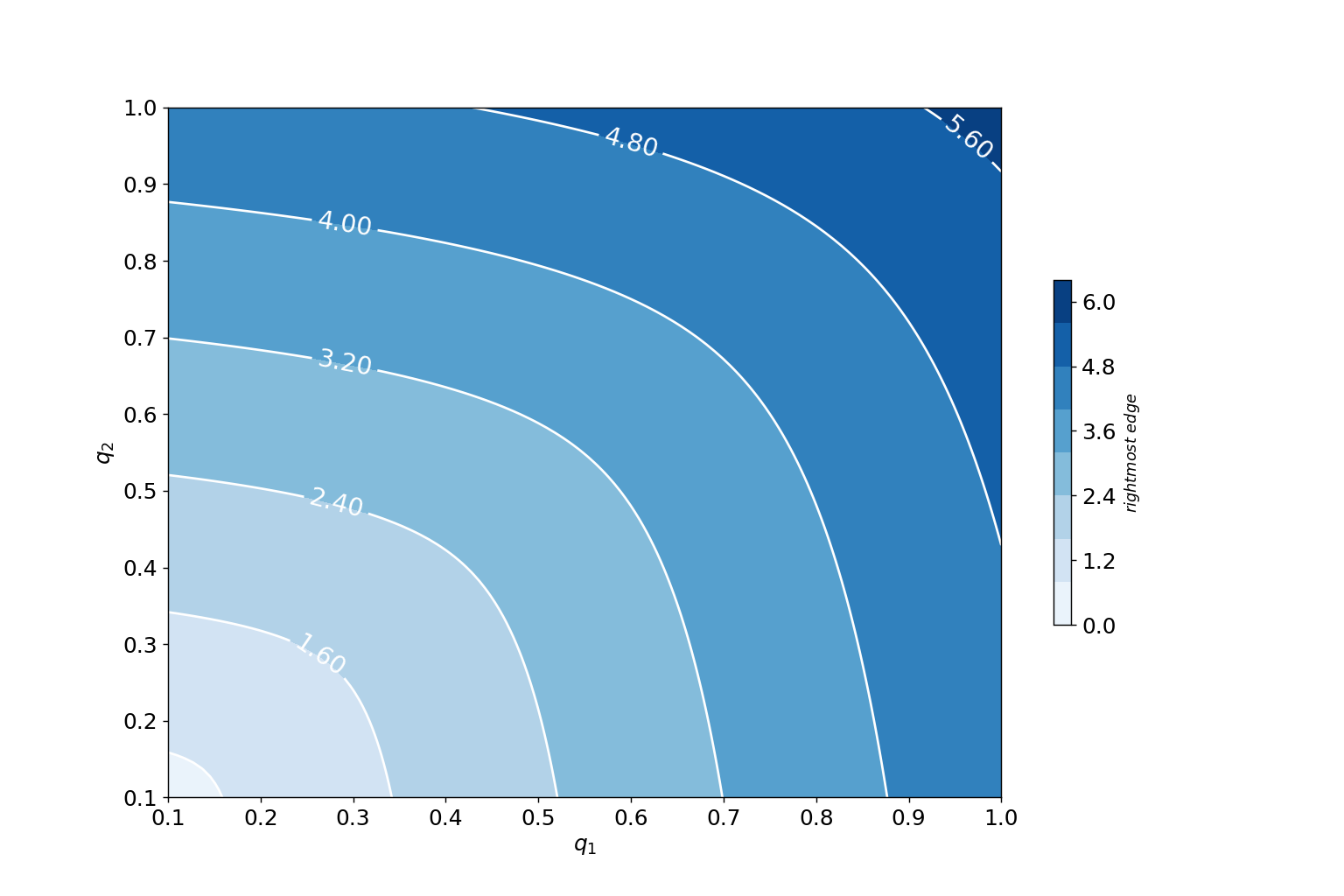}
\end{minipage}
\end{figure}
\vspace{-1em}    
The next theorem specifies the choice of $(v,\epsilon_n)$ and establishes the consistency of $\hat{r}(v,\epsilon_n)$. 

\begin{theorem}\label{consistency3}
Consider the factor model in \eqref{factor_model_matrix}, \eqref{sample_specific_data} and  Assumptions \ref{assumption_HD}-\ref{assumption_Sigma_x}. Consider $L\ge 2$ and $r_1$ in \eqref{heterogeneous_identifiable_spikes2}. Then, for any $\epsilon_n=o(1)$ satisfying $\epsilon_n n^{2/3}\to \infty$, $\lim\nolimits_{n\to \infty}\mathbb{P} \{\hat{r}(a_1,\epsilon_n)=r_1\} = 1$. 
\end{theorem}

Notably, $a_1$ is the optimal threshold, as characterized in Remark~\ref{optimal_threshold}.

\subsection{Anisotropic noise}

For general covariance $\Sigma_\varepsilon$ under both feature-specific and sample-specific heterogeneous missingness, the previous characterizations of the number $r_1$ of identifiable spikes (see \eqref{heterogeneous_identifiable_spikes1} and \eqref{heterogeneous_identifiable_spikes2}), together with the validity and consistency of our proposed MATE (see Theorems~\ref{consistency2} and \ref{consistency3}) in Sections~\ref{subsec_feature-specific_hete} and \ref{subsec_Sample-specific_hete}, remain valid. However, these MATEs depend on $\Sigma_\varepsilon$, which is generally not estimable in the high-dimensional regime except under sparsity assumptions (see e.g., \cite{bickel2008covariance}). 
This motivates us to focus on diagonal $\Sigma_\varepsilon$ as a tractable yet informative setting for studying the performance of MATE under anisotropic noise. 

\subsubsection{Diagonal noise covariance and parameter estimation}

We assume $\Sigma_{\varepsilon}={\rm diag}(\sigma^2_1,\cdots,\sigma^2_d)$. To model the high-dimensional noise variances $\{\sigma_i^2\}_{i=1}^d$, we impose a prior distribution $F_\Theta$ indexed by a parameter $\Theta\in \mathbb R^m$, a common practice in Bayesian statistics that provides flexibility and regularization in high-dimensional settings (see, e.g., \citealp{richardson1997bayesian,guan2011bayesian}). Specifically, let $\sigma_i^2\overset{i.i.d.}{\sim} F_\Theta$, where $F_\Theta$ has bounded support. We employ tools from {\it free probability} to construct a moment estimator for $\Theta$, under which the first $m$ moments $\beta_k(S_n^o)$, $1\le i\le m$, admit explicit limits under Assumptions~\ref{assumption_HD}-\ref{assumption_varepsilon}.

\subsubsection{A special case: Gamma-distributed noise variance}

We assume that $\Sigma_\varepsilon$ is diagonal with entries drawn from a Gamma distribution: 
    \begin{equation} \label{Gamma_distribution}
    \Sigma_\varepsilon = {\rm diag}(\sigma_1^2,\cdots,\sigma_d^2),\quad 
    \sigma_i^2\overset{i.i.d.}{\sim} {\rm Gamma}(\theta,\theta/\sigma^2),
    \end{equation}
where $\theta$ is the shape parameter and $\theta/\sigma^2$ is the rate parameter. The covariance matrix $\Sigma$ can then be expressed as 
    \begin{equation}\label{Gamma_Sigma_spiked_model}
    \Sigma = \Lambda \Lambda^\top + {\rm diag}(\sigma_1^2,\cdots,\sigma_d^2).
    \end{equation}
Moreover, when $r=0$, \eqref{Gamma_Sigma_spiked_model} reduces to the null model: 
    \begin{equation}\label{Gamma_Sigma_null_model}
    \Sigma={\rm diag}(\sigma_1^2,\cdots,\sigma_d^2). 
    \end{equation}
\eqref{Gamma_distribution} and \eqref{Gamma_Sigma_spiked_model} together yield a flexible covariance model that encompasses the standard spiked covariance model of \cite{johnstone2001distribution} as a special case. Indeed, as $\theta\rightarrow \infty$, ${\rm Gamma}(\theta, \theta/\sigma^2)$ converges to a point mass at $\sigma^2$, yielding $\Sigma_\varepsilon=\sigma^2 I_d$. Moreover, assuming Gamma-distributed diagonal noise variances is both realistic and practically motivated, allowing \eqref{Gamma_Sigma_spiked_model} to accommodate model misspecification and heterogeneous noise commonly encountered in applications such as gene microarray and GWAS studies \citep{ke2023estimation}.

Under this model, the eigenvalues of $\Sigma_\varepsilon$ are random with unbounded support. Consequently, the LSD of $\Sigma$ is also unbounded, and existing RMT results do not directly apply to model \eqref{Gamma_Sigma_spiked_model}. To address this issue, we approximate model \eqref{Gamma_Sigma_spiked_model} using a truncated Gamma-based spiked covariance model, following the approach of Section~4.2 in \cite{ke2023estimation}: 
\begin{equation}\label{TruncGamma_Sigma_spiked_model}
\Sigma =\Lambda\Lambda^{\top}+{\rm diag}(\sigma_1^2,\cdots,\sigma_d^2),\quad
\sigma_i^2\overset{i.i.d.}{\sim} {\rm TruncGamma}(\theta,\theta/\sigma^2,\sigma^2 T_1,\sigma^2 T_2).
\end{equation}
Here, ${\rm TruncGamma}(\theta,\theta/\sigma^2,\sigma^2T_1,\sigma^2T_2)$ denotes the Gamma distribution ${\rm Gamma}(\theta,\theta/\sigma^2)$ truncated to the interval $[\sigma^2T_1,\sigma^2 T_2]$, where $0< T_1 < T_2 < \infty$. 
Its density function is 
\begin{equation*}
f_{\rm TG}(x; \theta,\sigma^2, \sigma^2 T_1, \sigma^2 T_2)=\mathbf{1}_{[\sigma^2 T_1,\sigma^2 T_2]}(x) \frac{f_{\rm G}(x)}{F_{\rm G}(T_2)-F_{\rm G}(T_1)},\ x\in \mathbb{R}, 
\end{equation*}
where $f_{\rm G}(x)$ and $F_{\rm G}(x)$ denote the density and cumulative distribution function of ${\rm Gamma}(\theta,\theta/\sigma^2)$, respectively. 
Under model \eqref{TruncGamma_Sigma_spiked_model}, the LSD of $\Sigma$ has bounded support, which enables the application of existing results from RMT. Moreover, as shown in Section~4.2 of \cite{ke2023estimation}, the Gamma-based spiked model \eqref{Gamma_Sigma_spiked_model} can be well approximated by its truncated counterpart \eqref{TruncGamma_Sigma_spiked_model}. 

Under the two heterogeneous missingness cases, we establish the following convergence result. 

\begin{lemma}\label{lemma_heterogeneous}
Consider $\Sigma$ generated from \eqref{TruncGamma_Sigma_spiked_model}. Suppose that Assumptions~\ref{assumption_HD}-\ref{assumption_varepsilon} hold. 
(a) If the ESD of $P_d$ converges, then the ESD of $S_n^o$ converges weakly in the feature-specific heterogeneous missingness setting. 
(b) Under Assumption~\ref{assumption_Q_n}, the ESD of $S_n^o$ converges weakly in the sample-specific heterogeneous missingness setting. 
\end{lemma}

We further establish the convergence of the first two moments $\beta_k(S_n^o)$, $i=1,2$. 

\begin{lemma}\label{LSD_moments}
Consider $\Sigma$ generated from \eqref{TruncGamma_Sigma_spiked_model}. Suppose that Assumptions~\ref{assumption_HD}-\ref{assumption_varepsilon} hold. Define $\tilde \sigma^2=\EE (\sigma_i^2)$ and $\tilde \theta = \tilde \sigma^4/{\rm Var} (\sigma_i^2)$. 
\begin{itemize}
\item[(a)]For complete data $X^o=X$, we have $\beta_1(S_n) \sim \tilde \sigma^2$ and $\beta_2 (S_n) \sim \left(1+\gamma+\tilde\theta^{-1}\right)\tilde\sigma^4$.
\item[(b)] For incomplete data $X^o$ with feature-specific heterogeneous MCAR, we have 
\$
\beta_1(S_n^o) \sim \tilde\sigma^2 \left( \frac{1}{d}\sum_{k} p_k d_k  \right),\ \beta_2(S_n^o) \sim \tilde\sigma^4 (1+\tilde\theta^{-1}) \left( \frac{1}{d}\sum_{k}p_k^2 d_k\right) + \gamma \tilde \sigma^4 \left(\frac{1}{d}\sum_{k}p_k d_k \right)^2.
\$
\item[(c)] For incomplete data $X^o$ with sample-specific heterogeneous MCAR, we have 
\$
\beta_1(S_n^o) \sim \tilde\sigma^2\left( \frac{1}{n}\sum_{\ell}q_\ell n_\ell\right),\ \beta_2(S_n^o) \sim \tilde\sigma^4(1+\tilde\theta^{-1})\left(\frac{1}{n}\sum_{\ell}q_\ell n_\ell\right)^2 + \gamma\tilde\sigma^4\left(\frac{1}{n}\sum_{\ell}q_\ell^2 n_\ell\right). 
\$
\end{itemize}
\end{lemma}

Lemma~\ref{LSD_moments} directly yields the following moment estimators for $(\tilde\theta,\tilde\sigma^2)$ in two cases. 
\begin{itemize}
    \item[(1)] Feature-specific missing: 
$\hat\sigma^2=\frac{\beta_1(S_n^o)}{\frac{1}{d}\sum_k p_k d_k}$, 
$\hat\theta=\frac{\frac{1}{d}\sum_k p_k^2 d_k }{\left( \beta_1^{-2}(S_n^o) \beta_2(S_n^o) -\gamma \right) \left(\frac{1}{d}\sum_k p_k d_k \right)^2 -\frac{1}{d}\sum_k p_k^2 d_k }$; 
\item[(2)] Sample-specific missing: 
$\hat\sigma^2=\frac{\beta_1(S_n^o)}{\frac{1}{n}\sum_\ell q_\ell n_\ell}$, 
$\hat\theta=\frac{\left(\frac{1}{n}\sum_\ell q_\ell n_\ell\right)^2}{\left( \beta_1^{-2}(S_n^o) \beta_2(S_n^o) -1\right) \left(\frac{1}{n}\sum_\ell q_\ell n_\ell \right)^2 -\frac{\gamma}{n}\sum_\ell q_\ell^2 n_\ell}$. 
\end{itemize}

In practice, $\beta_1(S_n^o)$ and $\beta_2(S_n^o)$ can be replaced by 
$ 
\tilde \beta_1(S_n^o)= (d-\hat r)^{-1}\sum_{\hat r+1\le i\le d}\hat \lambda_i^o,\ 
\tilde \beta_2(S_n^o) = (d-\hat r)^{-1}\sum_{\hat r+1\le i\le d} (\hat \lambda_i^o)^2
$ 
to mitigate the effect of spiked eigenvalues, where $\hat r$ denotes a rough estimator of $r$. The consistency of $(\hat \theta, \hat\sigma^2)$ follows directly from Lemma~\ref{LSD_moments}.

\begin{remark}
The estimator $(\hat\theta,\hat\sigma^2)$ is a one-shot estimator of $(\tilde \theta, \tilde \sigma^2)$, with convergence rate $O_p(n^{-1})$, which is typical for linear spectral statistics \citep{bai2008clt}. In practice, we use it as an estimator of $(\sigma^2,\theta)$ by neglecting the truncation error in \eqref{TruncGamma_Sigma_spiked_model}.
\end{remark}

Lastly, we propose an algorithm to estimate the number of factors under anisotropic noise. From a practical viewpoint, the incomplete matrix $X^o$ can be recognized as either the row-wise or column-wise missing case. So we only present the algorithm for feature-specific missingness, as the other is completely parallel. 
Let $\{\hat \lambda_{1(m)}^*\}_{1\le m\le M}$ denote the largest eigenvalues of sample covariance matrices computed from $M$ simulated datasets, each of size $n$, generated under the null model \eqref{Gamma_Sigma_null_model} with the same missingness pattern as the observed data $X^o$, i.e., $p^*_i=\hat p_i$ for $1\le i\le d$, where $p^*_i$ and $\hat p_i$ are the non-missing rates of the $i$-th feature in the null model and $X^o$, respectively. 
Let $\hat T_\beta$ be the $(1-\beta)$-quantile of $\{\hat{\lambda}_{1(m)}^*\}_{1 \le m \le M}$, where $\beta=o(1)$. We then set $v=\hat T_\beta$ and apply MATE: 
\$
\hat{r}(\hat T_\beta,\epsilon_n)
=\sharp\{i\in [d]:\hat{\lambda}_i^o>\hat{T}_\beta+\epsilon_n\}.
\$
By \cite{bloemendal2016principal} and Section~4.3 of \cite{ke2023estimation}, the distribution of $\hat{\lambda}^o_{r_1+1}$ is close to that of $\hat{\lambda}_1^*$. 
Consequently, the number of identifiable factors can be consistently estimated by $\hat{r}(\hat T_\beta,\epsilon_n)$. The details are shown in Algorithm~\ref{Algorithm_MATE}. 
    \begin{algorithm}
        \caption{MATE for $r$ in factor models with anisotropic noise.}
        \label{Algorithm_MATE} 
        \algrenewcommand\alglinenumber[1]{\textbf{Step~#1}}
        \begin{algorithmic}[1]
        
        \renewcommand{\algorithmicrequire}{\textbf{Input:}}
        \renewcommand{\algorithmicensure}{\textbf{Output:}}
        
        \Require Observed incomplete data $X^o=(x_{ij}^o)$, $\beta\in (0,1)$, $\epsilon_n$ and $M \in \NN$.    
        \State Estimate $p_i$ as $\hat{p}_i= n^{-1}\sharp\{j: x_{ij}^o\neq 0\}$ for each $i\in [d]$. Compute eigenvalues $\{\hat{\lambda}_i^o\}_{1 \le i \le d}$ of  $X^oX^{o\top}/n$. Set $r_{\max}=10$ as a prefixed upper bound. Initialize $t=0$ and $\hat r^{(0)}=r_{\max}+1$.
        \State  Compute  $\tilde \beta_1(S_n^o)$ and $\tilde \beta_2(S_n^o)$ according to $\{\hat\lambda_i^o\}_{i\ge\hat r^{(t)}+1}$ to get $\hat{\theta}_{(t)}$ and  $\hat{\sigma}_{(t)}^2$. 
        \State Generate $M$ independent samples $\{X^*_m\}_{1\le m\le M}$ from the null model of \eqref{TruncGamma_Sigma_spiked_model} with $(\hat \theta_{(t)},\hat \sigma^2_{(t)})$ and the same missingness pattern as $X^o=(x_{ij}^o)$, and compute $\hat{T}_\beta^{(t)}$. 
        \State Update $\hat{r}^{(t+1)}=\sharp\{1\le i \le r_{\max}: \hat{\lambda}^o_i>\hat{T}_\beta^{(t)}+\epsilon_n\}$.
        
        \State Repeat Steps 2-4 until convergence. 
        %Final estimates $\hat{r}$, $\hat{\theta}$ and $\hat{\sigma}^2$ are obtained. 
        \Ensure  Estimators $\hat{r}$, $\hat{\theta}$ and $\hat{\sigma}^2$. 
        \end{algorithmic}
    \end{algorithm}

\section{Simulation}\label{Simulation}
In this section, we evaluate MATE in experiments under isotropic and anisotropic models with homogeneous and heterogeneous missingness (see Sections \ref{Homoscedasticity_case} and \ref{Heteroscedasticity_case}). Additionally, we estimate the number of factors from the imputed data using MATE and competing methods (see Section \ref{imputation}). The competing estimators are summarized in Section \ref{comparison}. From a data-analytic perspective, a given data matrix $X^o$ determines the blocks $\{B_{k \cdot}\}$ (feature-specific) and $\{B_{\cdot \ell}\}$ (sample-specific), either of which can characterize the missingness. Because of this duality, our simulations focus solely on feature-specific missingness. 

Observed data are generated from the model specified in \eqref{x_decomposition}, \eqref{gamma_decomposition} and \eqref{function_X^o}, that is, 
    \begin{equation}\label{simulation_X^o}
    X^o= \left(U 
    \begin{pmatrix}
        D_1 & \\
        & \sigma^2 I_{d-r}
    \end{pmatrix}^{1/2}V^{\top}\begin{bmatrix}F \\ e \end{bmatrix}\right)\circ B,
    \end{equation}
where the entries of $F\in \mathbb{R}^{r\times n}$ and $e\in \mathbb{R}^{d\times n}$ are independently drawn from $\mathcal{N}(0,1)$. Following \eqref{B_block} and \eqref{b_distribution}, we let $b_{ij}\stackrel{i.i.d.}{\sim}{\rm Bernoulli}(p_k)$ with $p_k\in(0,1]$ and $k\in[K]$. We consider two cases: $K=1$ and $K>1$. We set $r_{\max}=8$ for isotropic models, and $r_{\max}=10$ for anisotropic models. Each scenario is replicated 500 times to compute the \emph{mean squared error} (MSE).

\subsection{Parameter estimation for MATE}

In isotropic models, we estimate the number of factors using MATE $\hat{r}(v,\epsilon_n)$. For $K=1$, the threshold is ${v}=\hat{p} \hat{\sigma}^2 (1+\sqrt{\gamma_n})^2$, where $\hat{p}$ denotes the estimated non-missing rate, $\hat{\sigma}^2$ is an estimate of $\sigma^2$, and $\gamma_n=d/n$. For $K>1$, $v=\psi(\alpha_+)$ in \eqref{rightmost_edge_K>1}, where $\hat p_i$ is used as an estimate of $p_i$. 

Let $X^*$ be generated from the null model $\Sigma=\sigma^2 I_d$ with $p_i^* = \hat{p}_i$, where $p^*_i$ is the non-missing rate of the $i$-th feature in $X^*$, and let $\{X^*_m\}_{1\le m\le M}$ be $M$ independent copies. Define $\epsilon_n$ as the $(1-\beta)$-quantile of $\{\hat{\lambda}_{1(m)}^* - \frac{1}{M}\sum_{1\le m\le M}\hat{\lambda}_{1(m)}^* \}_{1 \le m \le M}$, where $\hat{\lambda}_{1(m)}^*$ is the largest eigenvalue of $X^*_mX^{*\top}_m/n$. See Algorithm \ref{Algorithm_epsilon_n} for details. In all simulations, we set $(\beta,M)=(0.1,500)$. 

    \begin{algorithm}
        \caption{Selection of the tuning parameter $\epsilon_n$.}
        \label{Algorithm_epsilon_n} 
        \algrenewcommand\alglinenumber[1]{\textbf{Step~#1}}
        \begin{algorithmic}[1]
        \renewcommand{\algorithmicrequire}{\textbf{Input:}}
        \renewcommand{\algorithmicensure}{\textbf{Output:}}
            \Require Observed incomplete data $X^o=(x_{ij}^o)$, $\beta\in (0,1)$ and $M \in \NN^+$.   
            \State Estimate $p_i$ as $\hat{p}_i= n^{-1}\sharp\{j: x_{ij}^o\neq 0\}$ for each $i\in [d]$.
            \State Generate $M$ copies of $X^*$, denoted by $\{X^*_m\}_{1\le m\le M}$.
            \State  Compute the largest eigenvalue $\hat{\lambda}_{1(m)}^*$ of $X^*_mX^{*\top}_m/n$, $1\le m\le M$. 
            \State  Take $\epsilon_n$ as the $(1-\beta)$-quantile of the sequence $\{\hat{\lambda}_{1(m)}^*- \frac{1}{M}\sum_{1\le m\le M}\hat{\lambda}_{1(m)}^*\}_{1 \le m \le M}$. 
            \Ensure   Parameter $\epsilon_n$. 
        \end{algorithmic}
    \end{algorithm}

In anisotropic models, we examine two Gamma-distributed cases and apply Algorithm \ref{Algorithm_MATE}, under both homogeneous and heterogeneous missingness. 

\subsection{Competing estimators}\label{comparison}
We compare MATE with seven competitors, including $\widetilde{{\rm CV}}$ \citep{jin2021factor}, M-ED, M-GR, M-ER, M-PC, M-IC, BEMA. These methods are described below. 
\begin{enumerate}
\item $\widetilde{\textbf{{CV}}}$ of \cite{jin2021factor}: we use $\widetilde{{\rm CV}}$ in Algorithm 1 of \cite{jin2021factor} with the number of draws $J=10$, subsets $K=8$, and leave-out probability $q=0.9$. To reduce computational cost, we set iteration number $\ell^*=\min\{\ln(0.002)/\ln(1-pq), 8\}$\footnote{Simulations show that when $p\ge 0.1$, the maximum number of iterations $\ell_{\max}\le 8$.}.
    
\item \textbf{PC} and \textbf{IC} of \cite{bai2002determining}:
\#
{\rm PC}(k)=V(k,\hat{F}^k)+k\hat{\sigma}^2 (d+n)/(dn) {\rm ln}[dn/(d+n)], \label{PC}\\
{\rm IC}(k)={\rm ln}(V(k,\hat{F}^k))+k(d+n)/(dn){\rm ln}[dn/(d+n)], \label{IC}
\#
where $V(k,\hat{F}^k)=d^{-1}\sum_{i=1}^d n^{-1}\hat{\varepsilon}_i\hat{\varepsilon}_i^{\top}=d^{-1}\sum_{r=k+1}^n \lambda_r (X^{\top}X/n)$, which is not directly computable due to missing entries in $X$.
    
\item \textbf{ED} of \cite{onatski2010determining}, \textbf{ER} and \textbf{GR} of \cite{ahn2013eigenvalue}:
\$
\hat{r}_{\rm ED}=\max\{i\le r_{\rm max}: \hat{\lambda}_i-\hat{\lambda}_{i+1} \ge \delta_0\},~
\hat{r}_{\rm ER}={\rm argmax}_{a \le i \le r_{\rm max}} \hat{\lambda}_i / \hat{\lambda}_{i+1},\\
\hat{r}_{\rm GR}={\rm argmax}_{a \le i \le r_{\rm max}} \ln [\textstyle\sum\nolimits_{j=i}^d\hat{\lambda}_j/\textstyle\sum\nolimits_{j=i+1}^d\hat{\lambda}_j]/ \ln [\textstyle\sum\nolimits_{j=i+1}^d\hat{\lambda}_j/\textstyle\sum\nolimits_{j=i+2}^d\hat{\lambda}_j],
\$
where $\delta_0$ is a fixed threshold, $r_{\rm max}$ is a prefixed upper bound, and $\{\hat{\lambda}_i\}$ are eigenvalues of $XX^{\top}/n$. Due to missing entries in $X$, $XX^{\top}/n$ cannot be computed directly, and thus these estimators are not applicable. 
    
\item \textbf{BEMA} of \cite{ke2023estimation}: Let $\Sigma={\rm diag}(\sigma_1^2,\cdots,\sigma_d^2)$, where $\sigma_i^2\stackrel{i.i.d}{\sim}{\rm Gamma}(\hat{\theta},\hat{\theta}/\hat{\sigma}^2)$, and $\{{\hat{\lambda}_i}\}$ denote eigenvalues of $XX^{\top}/n$. Set $(\alpha,\beta,M)=(0.2,0.1,500)$. \cite{ke2023estimation} estimate $(\hat{\sigma}^2, \hat{\theta})$ by $(\hat{\sigma}^2,\hat{\theta})=\arg \min\nolimits_{\theta}H(\theta)$, and estimate $r$ by $\hat{r}=\sharp\{i:\hat{\lambda}_i>\hat{T}\}$, where $H(\theta)=\min\nolimits_{\sigma^2}\left\{\textstyle \sum\nolimits_{\alpha (d\wedge n) \le k \le (1-\alpha) (d \wedge n)}[\hat{\lambda}_k-\sigma^2\bar{F}_{\gamma_n}^{-1}(k/(d \wedge n);1,\theta)]^2 \right\}$, $\bar{F}_{\gamma_n}^{-1}(k/(d \wedge n);1,\theta)$ is the $(k/(d \wedge n))$-upper-quantile of ${F}_{\gamma_n}(k/(d \wedge n);1,\theta)$, and where $\hat T$ is the $(1-\beta)$-quantile of $\{\hat \lambda_{1 (m)}\}$  under the null model. 
Since missingness alters the LSD of $XX^{\top}/n$, $(\hat{\sigma}^2, \hat{\theta})$ becomes biased, which in turn may induce bias in $\hat{r}$.
\end{enumerate}

Therefore, methods other than $\widetilde{\textbf{CV}}$ become infeasible for data with missingness. Following \cite{jin2021factor}, we modify \textbf{PC}, \textbf{IC}, \textbf{ED}, \textbf{ER}, and \textbf{GR}, and adapt \textbf{BEMA} so that they can be applied to incomplete  data. The specific modifications are as follows:
\begin{enumerate}
\item \textbf{M-ED}, \textbf{M-ER}, and \textbf{M-GR}: We compute ED, ER  and GR using eigenvalues of $X^oX^{o\top}/(\hat{p}^2n)$, where $\hat{p}$ is the estimated overall non-missing rate. 

\item \textbf{M-PC}, \textbf{M-IC}: We replace $V(k,\hat{F}^k)$ in (\ref{PC}) and (\ref{IC}) by $\widetilde{V}(k)=(\hat{p}dn)^{-1}\sum\nolimits_{(i,j)\in \Omega}(X^o-\hat{C}_{k,ij}^{(l^*)})$\footnote{We avoid replacing $XX^{\top}/n$ with $X^oX^{o\top}/(\hat{p}^2n)$ as this is inefficient and performs poorly under high missing rate.}, where $\hat{C}_{k,ij}^{(l^*)}$ is the iterated estimator proposed in Section 2 of \cite{jin2021factor}. 
        
\item \textbf{BEMA}: When $K=1$, let $\{\hat {\lambda}^o_i\}$ denote eigenvalues of $X^oX^{o\top}/n$, $\hat \sigma_o^2$ be its variance, and $\hat T^o$ correspond to $\hat{T}$. As shown in Section \ref{method}, $X^oX^{o\top}/n \overset{d}{=}pXX^{\top}/n$ and $\hat\lambda^o_i \overset{d}{=}p\hat{\lambda}_i$. Consequently, $\hat\sigma_o^2=p^2\hat{\sigma}^2$ and $\hat T^o=p\hat{T}$, and hence $(\hat{\sigma}^2/\hat{p}^2,\hat{\theta})=\arg \min\nolimits_{\theta}H(\theta)$ and $\hat{r}=\sharp\{i:\hat{\lambda}^o_i/\hat{p}>\hat{T}^o/\hat{p}\}$. Clearly, this updated $\hat{r}$ coincides with the original one. 
When $K>1$, the LSD of $X^oX^{o\top}/n$ is no longer explicitly related to $F_{\gamma_n}(x; \sigma^2, \theta)$, and a suitable adjustment is not available. We therefore retain the original BEMA. For simplicity of notation, both versions are referred to as BEMA.
\end{enumerate}

\subsection{Isotropic case}\label{Homoscedasticity_case}
Consider three models based on (\ref{simulation_X^o}) with $r=5$ and $e_i\stackrel{i.i.d.}{\sim}\mathcal{N}(0,I_d)$, as follows:
\begin{itemize}
\item[] \textbf{Model \ 1.} $\gamma=0.5$, $D_1={\rm diag}(3, 2.5, 2, 1.5, 1.1)$;
\item[] \textbf{Model \ 2.} $\gamma=1$, $D_1={\rm diag}(3.5, 3, 2.5, 2, 1.3)$;
\item[] \textbf{Model \ 3.} $\gamma=2$, $D_1={\rm diag}(3.5, 3, 2.5, 2, 1.6)$.
\end{itemize} 
For each model, we consider two settings: homogeneous missingness ($K=1$) and heterogeneous missingness ($K=3$). For $K=1$, we take $p\in \{1, 0.9, 0.7\}$, where $p=1$ means complete data (see Table~\ref{homoscedastic models with K=1}). For $K=3$, we set $(p_1,p_2,p_3)=(0.9,0.8,0.8),(0.8,0.7,0.6),(0.6,0.5,0.4)$ (see Table~\ref{homoscedastic models with K=3}). Table~\ref{homoscedastic models with K=1 and different p} further examines the sensitivity of  MATE and $\widetilde{{\rm CV}}$ to different missing rates.

Tables \ref{homoscedastic models with K=1}, \ref{homoscedastic models with K=1 and different p} and \ref{homoscedastic models with K=3} consistently show the superiority of MATE over the competitors. In general, MATE accurately estimates $r$ in both homogeneous and heterogeneous missingness regimes, even under high missing rates. In Table \ref{homoscedastic models with K=1}, MATE provides more accurate estimates, which  converge to the true values as $n$ increases, and remains stable across different missing rates, whereas the competing methods underestimate the factor number. Among these competitors, $\widetilde{\rm CV}$ performs best, with estimates converging as $n$ increases only in Model 2, while others do not exhibit clear convergence behavior. In Table \ref{homoscedastic models with K=1 and different p}, MATE remains stable in Models 1 and 2 and tends to underestimate $r$ as the non-missing rate decreases in Model 3, while $\widetilde{\rm CV}$ deteriorates as the missing rate increases across all models. Moreover, $\widetilde{\rm CV}$ is more computationally expensive than MATE since it requires iterations, and higher missing rates further increase the number of required iterations. 
In Table \ref{homoscedastic models with K=3}, MATE remains effective across three heterogeneous missingness settings, whereas the competitors still underestimate $r$. 

{\footnotesize
\setlength\tabcolsep{3pt} 
\begin{longtable}{ccccccccc}
\caption{Means (with MSEs in brackets) of the estimated number of common factors for isotropic models with homogeneous missing data ($K=1$) based on 500 simulations.}
\label{homoscedastic models with K=1}\\
\hline
\makebox[0.05\textwidth][c]{Model} &\makebox[0.08\textwidth][c]{$(d,n,p)$}  & \makebox[0.08\textwidth][c]{MATE}& \makebox[0.08\textwidth][c]{$\widetilde{\rm CV}$} & \makebox[0.08\textwidth][c]{M-ED} & \makebox[0.08\textwidth][c]{M-GR} & \makebox[0.08\textwidth][c]{M-ER} & \makebox[0.08\textwidth][c]{M-PC} & \makebox[0.08\textwidth][c]{M-IC}\\
\hline  
\endfirsthead

\multicolumn{9}{c}
{{\tablename\ \thetable{} -- Continued from previous page}} \\
\hline
 \makebox[0.05\textwidth][c]{Model} &\makebox[0.08\textwidth][c]{$(d,n,p)$}  & \makebox[0.08\textwidth][c]{MATE}& \makebox[0.08\textwidth][c]{$\widetilde{\rm CV}$} & \makebox[0.08\textwidth][c]{M-ED} & \makebox[0.08\textwidth][c]{M-GR} & \makebox[0.08\textwidth][c]{M-ER} & \makebox[0.08\textwidth][c]{M-PC} & \makebox[0.08\textwidth][c]{M-IC}\\
\hline  
\endhead

\hline \multicolumn{9}{r}{\textit{Continued on next page}} \\
\endfoot

\hline
\endlastfoot

\multirow{6}{*}{1} &(250,500,1) & $4.93_{(0.10)}$ &$ 3.10_{(3.69)} $&$ 1.18_{(16.04)} $&$ 2.30_{(7.96)}$ & $2.40_{(7.42)}$ & $1.93_{(9.50)}$ & $1.39_{(13.28)}$ \\
& (500, 1000, 1) &$ 5.02_{(0.03)}$ &$ 3.07_{(3.78)}$ & $0.61_{(20.30)}$ & $2.30_{(7.84)}$ & $2.39_{(7.37)}$& $1.04_{(15.75)}$ & $1.00_{(15.99)} $\\ 
&(250, 500, 0.9) & $4.90_{(0.15)} $& $3.00_{(3.99)}$ & $0.59_{(20.42)}$ &$ 2.25_{(8.16)}$ & $2.31_{(7.86)}$ &$ 1.99_{(9.10)}$ & $1.59_{(11.88)}$ \\
& (500, 1000, 0.9)&$ 4.99_{(0.07)}$ &$ 3.00_{(4.00)}$ & $0.22_{(23.30)}$ & $2.17_{(8.58)}$ & $2.22_{(8.29)}$ &$ 1.14_{(15.05)}$ & $1.00_{(15.97)}$\\
&(250, 500, 0.7) & $4.84_{(0.33)}$ & $2.81_{(4.97)}$ &$ 0.07_{(24.41)}$ & $2.00_{(9.59)}$ & $2.04_{(9.33)}$ & $2.04_{(8.78)}$ & $1.93_{(9.49)}$ \\
& (500, 1000, 0.7) & $4.85_{(0.28)}$ &$ 2.95_{(4.25)}$ & $0.01_{(24.95)}$ &$ 1.92_{(10.03)} $& $1.96_{(9.83)}$ & $1.76_{(10.68)}$ & $1.20_{(14.63)}$\\
\hline 
\multirow{6}{*}{2}&(250, 250, 1) & $4.85_{(0.16)}$& $3.80_{(1.60)}$ &$ 1.85_{(12.42)}$ & $2.92_{(5.46)}$ & $3.13_{(4.49)}$ & $2.22_{(7.91)}$& $1.88_{(9.87)}$ \\
& (500, 500, 1) & $4.98_{(0.03)}$ & $3.95_{(1.15)}$ &$ 1.10_{(17.56)}$ & $3.27_{(3.91)} $& $3.37_{(3.47)} $& $1.82_{(10.26)}$ & $1.36_{(13.47)}$ \\
&(250, 250, 0.9) & $4.74_{(0.29)}$ & $3.47_{(2.58)}$ &$ 1.07_{(17.38)}$ & $2.72_{(6.46)}$ & $2.90_{(5.54)}$ &$ 2.38_{(7.08)}$ &$ 1.95_{(9.40)}$ \\ 
&(500, 500, 0.9) &$ 4.94_{(0.11)}$& $3.64_{(2.07)}$ & $0.33_{(22.60)}$ & $3.12_{(4.54)}$ & $3.20_{(4.16)}$ & $1.94_{(9.45)}$ & $1.53_{(12.28)}$ \\
&(250, 250, 0.7) & $4.78_{(0.38)}$ & $2.93_{(4.39)} $& $0.11_{(24.14)} $& $2.25_{(8.49)}$ & $2.39_{(7.82)}$ &$ 2.92_{(4.40)}$ & $2.27_{(7.67)}$ \\
&(500, 500, 0.7)& $4.81_{(0.32)}$ & $3.00_{(4.00)}$ & $0.02_{(24.84)}$ & $2.36_{(7.95)} $&$ 2.42_{(7.62)}$ & $2.02_{(8.89)}$ & $1.92_{(9.57)} $\\
\hline
\multirow{6}{*}{3}& (500, 250, 1)& $ 4.93_{(0.08)}$ & $2.98_{(4.15)} $& $0.07_{(24.40)}$ & $2.03_{(9.61)} $& $2.09_{(9.23)}$ & $1.00_{(15.99)}$ & $1.00_{(16.00)}$ \\
&(1000, 500, 1) &$ 5.01_{(0.01)} $& $3.00_{(4.00)}$ &$ 0.00_{(25.00)}$ & $2.02_{(9.61)}$& $2.06_{(9.36)}$ & $1.00_{(16.00)}$ & $1.00_{(16.00)}$ \\
&(500, 250, 0.9) &$ 4.90_{(0.14)}$& $2.79_{(5.05)}$ & $0.03_{(24.77)}$ & $2.02_{(9.56)}$ & $2.07_{(9.30)}$ & $1.05_{(15.66)} $& $1.00_{(16.00)}$\\
&(1000, 500, 0.9)& $4.98_{(0.05)}$ &$ 2.92_{(4.41)}$ & $0.00_{(25.00)}$ & $1.91_{(10.23)}$ & $1.93_{(10.07)}$ & $1.00_{(16.00)}$ & $1.00_{(16.00)}$\\
&(500, 250, 0.7) & $4.60_{(0.46)}$ & $2.02_{(9.01)}$ & $0.00_{(24.97)}$ & $1.77_{(10.97)} $&$ 1.80_{(10.82)}$ & $1.47_{(12.74)}$ & $1.01_{(15.92)} $\\
&(1000, 500, 0.7)&$ 4.74_{(0.29)}$ & $2.05_{(8.75)}$ & $0.00_{(25.00)}$ & $1.78_{(10.92)}$ &$ 1.79_{(10.84)}$ & $1.00_{(16.00)}$ & $1.00_{(16.00)}$\\ 
\end{longtable} 
}

{\footnotesize
\setlength\tabcolsep{3pt} 
\begin{longtable}{c|ccccc|ccccc}
\caption{Means (with MSEs in brackets) of the estimated number of common factors for isotropic models with five different homogeneous missing data ($K=1$) based on 500 simulations.}
\label{homoscedastic models with K=1 and different p}\\
\hline
\multirow{2}{*}{\makebox[0.13\textwidth][c]{$(d,n, {\rm Model})$}} &\multicolumn{5}{c}{MATE}  &\multicolumn{5}{|c}{$\widetilde{\rm CV}$} \\ 
\cline{2-11}
& \makebox[0.05\textwidth][c]{$p$=$1$} & \makebox[0.05\textwidth][c]{$p$=$0.9$}& \makebox[0.05\textwidth][c]{$p$=$0.7$} & \makebox[0.05\textwidth][c]{$ p$=$0.5$}&\makebox[0.05\textwidth][c]{$p$=$0.3$} & \makebox[0.05\textwidth][c]{$p$=$1$} & \makebox[0.05\textwidth][c]{$p$=$0.9$}& \makebox[0.05\textwidth][c]{$p$=$0.7$} & \makebox[0.05\textwidth][c]{$p$=$0.5$}&\makebox[0.05\textwidth][c]{$p$=$0.3$} \\ 
\hline  
\endfirsthead

\multicolumn{9}{c}%
{{\tablename\ \thetable{} -- Continued from previous page}} \\
\hline
\multirow{2}{*}{\makebox[0.13\textwidth][c]{$(d,n, {\rm Model})$}} &\multicolumn{5}{c}{MATE}  &\multicolumn{5}{|c}{$\widetilde{\rm CV}$} \\ 
\cline{2-11}
& \makebox[0.05\textwidth][c]{$p$=$1$} & \makebox[0.05\textwidth][c]{$p$=$0.9$}& \makebox[0.05\textwidth][c]{$p$=$0.7$} & \makebox[0.05\textwidth][c]{$ p$=$0.5$}&\makebox[0.05\textwidth][c]{$p$=$0.3$} & \makebox[0.05\textwidth][c]{$p$=$1$} & \makebox[0.05\textwidth][c]{$p$=$0.9$}& \makebox[0.05\textwidth][c]{$p$=$0.7$} & \makebox[0.05\textwidth][c]{$p$=$0.5$}&\makebox[0.05\textwidth][c]{$p$=$0.3$} \\ 
\hline  
\endhead

\hline \multicolumn{9}{r}{\textit{Continued on next page}} \\
\endfoot

\hline
\endlastfoot

\multirow{2}{*}{(250, 500, 1)} & 4.92 & 4.89 & 4.73 & 4.74 & 4.84 & 3.08 & 3.00 & 2.77 & 1.99 & 1.00\\
& (0.08) & (0.16) & (0.37) & (0.45) & (0.58 )& (3.75) & (3.99) & (5.16) &(9.13) &(15.99)\\
\hline
\multirow{2}{*}{(500, 1000, 1)}& 5.02& 4.98 &4.79 &4.70 &4.40 &3.07 &3.00 &2.94 &2.00 &1.00 \\
& (0.02) &(0.05) &(0.31) &(0.45) & (0.74) &(3.79) & (3.99) &(4.31) &(9.00) &(15.99)\\
\hline
\multirow{2}{*}{(250, 250, 2)} & 4.85 &4.81 &4.81 &5.03& 5.13 &3.82 &3.49 &2.94 &2.03 &1.03 \\
&(0.16) &(0.24) &(0.34) &(0.43) &(0.64) &(1.55) &(2.52)& (4.35) &(8.99) & (15.78) \\
\hline
\multirow{2}{*}{(500, 500, 2)}& 4.97 &4.90 & 4.78 &4.78 &4.81 &3.95 &3.65 &3.00 &2.09 &1.02\\
&(0.05) &(0.14) &(0.32) & (0.42) &(0.51) &(1.15) &(2.04) &(4.02) &(8.55) &(15.87) \\
\hline
\multirow{2}{*}{(500, 250, 3)} &4.88 &4.86 &4.63 &4.30 &3.69 &2.96 &2.80 &2.04 &1.13 &1.00\\
& (0.12)& (0.16) & (0.41) &(0.83)	&(2.17) &(4.24) &(5.00) &(8.88) &(15.09) &(16.00)\\
\hline
\multirow{2}{*}{(1000, 500, 3)}&5.00 &5.01 &4.75 &4.37 &3.56 &3.00 &2.93 &2.03 &1.07 &1.00\\
&(0.01) &(0.04) &(0.30) &(0.73) &(2.43) &(3.99) &(4.34) &(8.87) &(15.48) &(16.00)\\
\end{longtable}
}

{\footnotesize
\setlength\tabcolsep{3pt} 
\begin{longtable}{ccccccccc}
\caption{Means (with MSEs in brackets) of the estimated number of common factors for isotropic models with heterogeneous missing data ($K=3$) based on 500 simulations.}
\label{homoscedastic models with K=3} \\
\hline
\makebox[0.05\textwidth][c]{Model} &\makebox[0.1\textwidth][c]{$(d,n,p_1,p_2,p_3)$}  & \makebox[0.08\textwidth][c]{MATE}& \makebox[0.08\textwidth][c]{$\widetilde{\rm CV}$} & \makebox[0.08\textwidth][c]{M-ED} & \makebox[0.08\textwidth][c]{M-GR} & \makebox[0.08\textwidth][c]{M-ER} & \makebox[0.08\textwidth][c]{M-PC} & \makebox[0.08\textwidth][c]{M-IC}\\
\hline  
\endfirsthead

\multicolumn{9}{c}%
{{\tablename\ \thetable{} -- Continued from previous page}} \\
\hline
 \makebox[0.05\textwidth][c]{Model} &\makebox[0.1\textwidth][c]{$(d,n,p_1,p_2,p_3)$}  & \makebox[0.08\textwidth][c]{MATE}& \makebox[0.08\textwidth][c]{$\widetilde{\rm CV}$} & \makebox[0.08\textwidth][c]{M-ED} & \makebox[0.08\textwidth][c]{M-GR} & \makebox[0.08\textwidth][c]{M-ER} & \makebox[0.08\textwidth][c]{M-PC} & \makebox[0.08\textwidth][c]{M-IC}\\
\hline  
\endhead

\hline \multicolumn{9}{r}{\textit{Continued on next page}} \\
\endfoot

\hline
\endlastfoot

\multirow{6}{*}{1} &(250, 500, 0.9, 0.8, 0.7)  & $4.71_{(0.33)}$ & $2.99_{(4.06)}$ & $0.27_{(24.82)}$ & $2.04_{ (9.39)}$ & $2.13_{(8.85)}$ & $2.00_{(9.03)}$ & $1.77_{(10.58)}$\\
&(500, 1000, 0.9, 0.8, 0.7) &$4.87_{(0.18)}$ & $3.00_{(4.00)}$ & $0.04_{(24.64)} $& $2.14_{(8.71)} $& $2.19_{(8.43)}$ & $1.38_{(13.37)}$& $1.03_{(15.78)}$ \\  
&(250, 500, 0.8, 0.7, 0.6) &$ 4.75_{ (0.33)}$& $2.75_{(5.23)} $& $0.12_{(24.06)}$ & $1.96_{(9.78)} $& $2.02_{(9.45)}$& $2.03_{(8.84)}$ & $1.94_{(9.39)}$  \\
& (500, 1000, 0.8, 0.7, 0.6)&$ 4.79_{(0.32)}$ & $2.95_{(4.26)}$ & $0.01_{(24.93)}$ & $1.93_{(10.00)} $& $1.96_{(9.77)}$ &$ 1.78_{ (10.53) }$ & $1.18_{(14.77)}$ \\
&(250, 500, 0.6, 0.5, 0.4) & $4.62_{(0.51)} $& $1.99_{(9.11)} $& $0.00_{(25.00)} $& $1.72_{ (11.20)}$ &$ 1.75_{(11.05)}$ & $2.99_{ (4.03) }$ & $2.33_{(7.36)}$\\
&(500, 1000, 0.6, 0.5, 0.4)&$4.66_{(0.46) }$ & $2.00_{(9.01)}$ & $0.00_{(25.00)}$ &$ 1.64_{(11.73) }$ & $1.65_{(11.64) }$ & $2.01_{(8.96)}$ & $1.99_{(9.04)}$\\
\hline  
\multirow{6}{*}{2}&(250, 250, 0.9, 0.8, 0.7) & $4.68_{(0.37)}$ & $3.15_{ (3.55)}$& $0.38_{(22.14)}$ & $2.50_{(7.46)}$ & $2.64_{(6.78)}$ & $2.68_{(5.62)}$ &$ 2.04_{(8.82)}$  \\
&(500, 500, 0.9, 0.8, 0.7) & $4.78_{(0.26)} $& $3.20_{(3.41)}$ & $0.07_{(24.48)}$ &$2.68_{(6.64)}$ & $2.79_{(6.08)} $& $1.99_{(9.08)}$ & $1.73_{(10.88)}$\\
&(250, 250, 0.8, 0.7, 0.6)  & $4.65_{(0.40)}$ & $2.93_{(4.41)}$ & $0.10_{(24.19)}$ & $2.23_{(8.64)}$ & $2.34_{(8.08)}$ &$2.91_{(4.44)}$ & $2.25_{(7.75)}$ \\
&(500, 500, 0.8, 0.7, 0.6) & $4.77_{(0.33)} $& $3.00_{(4.00)}$ & $0.03_{(24.76)}$ & $2.33_{(8.04)} $& $2.39_{(7.75)} $& $2.02_{(8.93)} $& $1.91_{(9.64) }$\\
&(250, 250, 0.6, 0.5, 0.4) & $4.78_{(0.45)}$ & $2.06_{(8.78)}$& $0.00_{(25.00)}$ & $1.85_{(10.58)}$ & $1.88_{(10.38)}$ & $4.06_{(0.99)}$ & $3.14_{(3.62)}$ \\
&(500, 500, 0.6, 0.5, 0.4)&$ 4.63_{(0.46)} $&$ 2.12_{(8.41)} $& $0.00_{(25.00)}$ &$ 1.75_{(11.12)}$ & $1.79_{(10.91)} $& $2.96_{(4.21)}$ & $2.34_{(7.32)}$ \\
\hline
\multirow{6}{*}{3}&(500, 250, 0.9, 0.8, 0.7)&  $4.65_{(0.35)}$ & $2.37_{(7.13)}$ &$ 0.00_{(24.96)} $& $1.87_{(10.42)}$ & $1.93_{(10.05)} $& $1.14_{(15.02)}$ & $1.00_{(15.99)}$ \\
&(1000, 500, 0.9, 0.8, 0.7) &$ 4.91_{(0.15)}$ & $2.50_{(6.51)}$ & $0.00_{(25.00)} $& $1.91_{(10.15) }$ & $1.92_{ (10.07)}$ &$ 1.00_{(16.00)} $& $1.00_{(16.00)}$ \\
&(500, 250, 0.8, 0.7, 0.6) &$ 4.55_{(0.48)}$ &$ 2.05_{(8.84)}$ &$ 0.00_{(24.98)}$ & $1.80_{(10.80)}$ &$ 1.83_{(10.61)}$ & $1.44_{(12.92)}$ & $1.02_{(15.89)}$  \\
&(1000, 500, 0.8, 0.7, 0.6)& $4.72_{(0.34)} $& $2.04_{(8.81)}$ & $0.00_{(25.00)} $& $1.81_{(10.75)} $& $1.82_{ (10.67)} $&$ 1.00_{(16.00)}$ & $1.00_{(16.00)}$\\ 
&(500, 250, 0.6, 0.5, 0.4) & $4.14_{(1.05)}$ & $1.13_{(15.12)} $& $0.00_{(25.00)}$& $1.60_{(12.01)}$ & $1.60_{(11.99)}$ &$ 2.59_{(6.06)}$ & $1.66_{(11.37)}$ \\
&(1000, 500, 0.6, 0.5, 0.4)& $4.26_{(0.84)}$ & $1.10_{(15.31)}$& $0.00_{(25.00)}$ & $1.55_{(12.32)}$ & $1.56_{(12.28)}$& $1.32_{(13.76)}$ & $1.01_{ (15.92)}$ \\
\end{longtable}
}

\subsection{Anisotropic case}\label{Heteroscedasticity_case}
Consider two models based on (\ref{simulation_X^o}) with $\varepsilon_i=\Sigma_{\varepsilon}^{1/2}e_i\stackrel{i.i.d.}{\sim}\mathcal{N}(0,\Sigma_{\varepsilon})$ as follows:
\begin{itemize}
\item[]\textbf{Model \ 4.} $r=0$, $\Sigma_{\varepsilon}={\rm diag}(\sigma_1^2,\cdots,\sigma_d^2)$ with $\sigma_i^2\overset{i.i.d.}{\sim} {\rm Gamma} (\theta,\theta/\sigma^2)$, where $\theta=3$, $\sigma^2=1$. 
\item[]\textbf{Model \ 5.}  $r=3$, $D_1=(3.5,3,2.4)$, $\Sigma_{\varepsilon}={\rm diag}(\sigma_1^2,\cdots,\sigma_d^2)$ with $\sigma_i^2\overset{i.i.d.}{\sim} {\rm Gamma} (\theta,\theta/\sigma^2)$, where $\theta=3$, $\sigma^2=1$. 
\end{itemize}

We first compare MATE with BEMA by estimating $\theta$ and $\sigma^2$ in Model 4 for $p\in \{1,0.9,0.7\}$ (see Table \ref{heteroscedastic Model 4 with K=1}). We then examine ($\hat{r}, \hat{\theta}, \hat{\sigma}^2$) for both methods in Model 5 under homogeneous ($K=1$, $p\in \{1,0.9,0.7\}$) and heterogeneous ($K=3$; $n_1=n_2=n_3=n/3$, $(p_1,p_2,p_3)=(0.9,0.8,0.7), (0.8,0.6,0.4)$) missingness (see Table \ref{heteroscedastic Model 5 with K=1} and \ref{heteroscedastic Model 5 with K=3}).

In Table \ref{heteroscedastic Model 4 with K=1}, both MATE and BEMA estimate $(\theta, \sigma^2)$ effectively in the pure noise setting, regardless of the missing rate. Table \ref{heteroscedastic Model 5 with K=1} shows that, when $K=1$, ($\hat{r}, \hat{\theta}, \hat{\sigma}^2$) produced by both methods remain accurate and efficient across different values of $d/n$ and $p$. In addition, the performance of both methods deteriorates as $p$ decreases. In Table \ref{heteroscedastic Model 5 with K=3}, MATE continues to perform well across all missingness levels. For BEMA, when $(p_1,p_2,p_3)=(0.9,0.8,0.7)$, $\hat{r}$ remains accurate while $(\hat{\theta},\hat{\sigma}^2)$ are slightly underestimated; when $(p_1,p_2,p_3)=(0.8,0.6,0.4)$, all three estimates exhibit underestimation.  

{\footnotesize
\setlength\tabcolsep{3pt} 
\begin{longtable}{c|cc|cc|cc|cc|cc|cc}
\caption{Means (with MSEs in brackets) of the estimated number of common factors for anisotropic Model 4 with homogeneous missing data ($K=1$) based on 500 simulations.}
\label{heteroscedastic Model 4 with K=1} \\
\hline
\multirow{3}{*}{\makebox[0.1\textwidth][c]{$(d,n)$}} &\multicolumn{4}{c}{$p$=$1$}  &\multicolumn{4}{|c}{$p$=$0.9$} &\multicolumn{4}{|c}{$p$=$0.7$}\\ 
\cline{2-13}
& \multicolumn{2}{c}{$\hat{\theta}$}& \multicolumn{2}{|c}{$\hat{\sigma}^2$} & \multicolumn{2}{|c}{$\hat{\theta}$}& \multicolumn{2}{|c}{$\hat{\sigma}^2$}& \multicolumn{2}{|c}{$\hat{\theta}$}& \multicolumn{2}{|c}{$\hat{\sigma}^2$}\\
\cline{2-13}
&{\scriptsize MATE} & {\scriptsize BEMA}&{\scriptsize MATE} & {\scriptsize BEMA} &{\scriptsize MATE} & {\scriptsize BEMA} &{\scriptsize MATE} & {\scriptsize BEMA} &{\scriptsize MATE} & {\scriptsize BEMA} &{\scriptsize MATE} & {\scriptsize BEMA} \\
\hline  
\endfirsthead

\multicolumn{13}{c}%
{{\tablename\ \thetable{} -- Continued from previous page}} \\
\hline
 \multirow{3}{*}{\makebox[0.1\textwidth][c]{$(d,n)$}} &\multicolumn{4}{c}{$p$=$1$}  &\multicolumn{4}{|c}{$p$=$0.9$} &\multicolumn{4}{|c}{$p$=$0.7$}\\ 
\cline{2-13}
& \multicolumn{2}{c}{$\hat{\theta}$}& \multicolumn{2}{|c}{$\hat{\sigma}^2$} & \multicolumn{2}{|c}{$\hat{\theta}$}& \multicolumn{2}{|c}{$\hat{\sigma}^2$}& \multicolumn{2}{|c}{$\hat{\theta}$}& \multicolumn{2}{|c}{$\hat{\sigma}^2$}\\
\cline{2-13}
&{\scriptsize MATE} & {\scriptsize BEMA}&{\scriptsize MATE} & {\scriptsize BEMA} &{\scriptsize MATE} & {\scriptsize BEMA} &{\scriptsize MATE} & {\scriptsize BEMA} &{\scriptsize MATE} & {\scriptsize BEMA} &{\scriptsize MATE} & {\scriptsize BEMA} \\
\hline  
\endhead

\hline \multicolumn{13}{r}{\textit{Continued on next page}} \\
\endfoot

\hline
\endlastfoot

\multirow{2}{*}{(250, 500)}&3.03	&3.06	&1.00	&1.00	&3.04	&3.07	&1.00	&1.00	&3.01	&3.08	&1.00	&1.00\\
&(0.09)	&(0.11)	&(0.00)	&(0.00)	&(0.09)	&(0.12)	&(0.00)	&(0.00)	&(0.11)	&(0.14)	&(0.00)	&(0.00)\\
\hline
\multirow{2}{*}{(500, 1000)}&3.01	&3.02	&1.00	&1.00	&3.00	&3.02	&1.00	&1.00	&3.00	&3.01	&1.00	&1.00\\
&(0.05)	&(0.06)	&(0.00)	&(0.00)	&(0.05)	&(0.05)	&(0.00)	&(0.00)	&(0.05)	&(0.05)	&(0.00)	&(0.00)\\
\hline
\multirow{2}{*}{(250, 250)} &2.97	&3.05	&1.00	&1.00	&3.01	&3.11	&1.00	&1.00	&2.95	&3.04	&1.00	&1.00\\
&(0.12)	&(0.15)	&(0.00)	&(0.00)	&(0.11)	&(1.03)	&(0.00)	&(0.00)	&(0.11)	&(0.17)	&(0.00)	&(0.00)\\
\hline
\multirow{2}{*}{(500, 500)}&3.02	&3.04	&1.00	&1.00	&3.00	&3.03	&1.00	&1.00	&2.95	&2.98	&1.00	&1.00\\
&(0.05)	&(0.06)	&(0.00)	&(0.00)	&(0.06)	&(0.07)	&(0.00)	&(0.00)	&(0.05)	&(0.05)	&(0.00)	&(0.00)\\
\hline
\multirow{2}{*}{(500, 250)} &2.99	&3.05	&1.00	&1.00	&2.95	&3.02	&1.00	&1.00	&2.93	&2.99	&1.00	&1.00\\
&(0.06)	&(0.09)	&(0.00)	&(0.00)	&(0.06)	&(0.08)	&(0.00)	&(0.00)	&(0.07)	&(0.07)	&(0.00)	&(0.00)\\
\hline
\multirow{2}{*}{(1000, 500)}&2.99	&3.01	&1.00	&1.00	&2.99	&3.00	&1.00	&1.00	&2.95	&2.98	&1.00	&1.00\\
&(0.02)	&(0.03)	&(0.00)	&(0.00)	&(0.03)	&(0.03)	&(0.00)	&(0.00)	&(0.03)	&(0.03)	&(0.00)	&(0.00)\\
\end{longtable}
}

{\footnotesize
\setlength\tabcolsep{3pt} 
\begin{longtable}{c|cc|cc|cc}
\caption{Means (with MSEs in brackets) of the estimated number of common factors for anisotropic Model 5 with homogeneous missing data ($K=1$) based on 500 simulations.}
\label{heteroscedastic Model 5 with K=1}\\
\hline
\multirow{2}{*}{\makebox[0.13\textwidth][c]{$(d,n,p)$}} &\multicolumn{2}{c}{$\hat{r}$}  &\multicolumn{2}{|c}{$\hat{\theta}$}& \multicolumn{2}{|c}{$\hat{\sigma}^2$} \\ 
\cline{2-7}
& \makebox[0.05\textwidth][c]{MATE} & \makebox[0.05\textwidth][c]{BEMA}& \makebox[0.05\textwidth][c]{MATE} & \makebox[0.05\textwidth][c]{BEMA}&\makebox[0.05\textwidth][c]{MATE} & \makebox[0.05\textwidth][c]{BEMA} \\
\hline  
\endfirsthead

\multicolumn{7}{c}%
{{\tablename\ \thetable{} -- Continued from previous page}} \\
\hline
\multirow{2}{*}{\makebox[0.13\textwidth][c]{$(d,n,p)$}} &\multicolumn{2}{|c}{$\hat{r}$}  &\multicolumn{2}{|c}{$\hat{\theta}$}& \multicolumn{2}{|c}{$\hat{\sigma}^2$} \\ 
\cline{2-7}
& \makebox[0.05\textwidth][c]{MATE} & \makebox[0.05\textwidth][c]{BEMA}& \makebox[0.05\textwidth][c]{MATE} & \makebox[0.05\textwidth][c]{BEMA}&\makebox[0.05\textwidth][c]{MATE} & \makebox[0.05\textwidth][c]{BEMA} \\
\hline  
\endhead

\hline \multicolumn{7}{r}{\textit{Continued on next page}} \\
\endfoot

\hline
\endlastfoot

(250, 500, 1)	&$3.13_{(0.19)}$	&$3.06_{(0.07)}$	&$3.09_{(0.11)}$	&$2.92_{(0.09)}$	&$0.99_{(0.00)}$	&$1.02_{(0.00)}$	\\
(500, 1000, 1)	&$3.11_{(0.11)}$	&$3.06_{(0.06)}$	&$3.06_{(0.05)}$	&$2.96_{(0.05)}$	&$1.00_{(0.00)}$	&$1.01_{(0.00)}$	\\
(500, 250, 1)	&$3.07_{(0.07)}$	&$3.02_{(0.03)}$	&$2.84_{(0.07)}$	&$2.88_{(0.07)}$	&$0.98_{(0.00)}$	&$1.02_{(0.00)}$	\\
(1000, 500, 1)	&$3.08_{(0.08)}$	&$3.04_{(0.04)}$	&$2.90_{(0.03)}$	&$2.91_{(0.03)}$	&$0.99_{(0.00)}$	&$1.01_{(0.00)}$	\\
(250, 250, 1)	&$3.10_{(0.11)}$	&$3.03_{(0.03)}$	&$2.99_{(0.10)}$	&$2.89_{(0.12)}$	&$0.98_{(0.00)}$	&$1.03_{(0.00)}$	\\
(500, 500, 1)	&$3.12_{(0.14)}$	&$3.06_{(0.07)}$	&$3.00_{(0.05)}$	&$2.94_{(0.05)}$	&$0.99_{(0.00)}$	&$1.01_{(0.00)}$	\\\hline
(250, 500, 0.9)	&$3.13_{(0.15)}$	&$3.07_{(0.07)}$	&$3.16_{(0.14)}$	&$2.99_{(0.11)}$	&$1.00_{(0.00)}$	&$1.03_{(0.00)}$	\\
(500, 1000, 0.9)	&$3.14_{(0.16)}$	&$3.08_{(0.08)}$	&$3.10_{(0.06)}$	&$3.00_{(0.05)}$	&$1.00_{(0.00)}$	&$1.02_{(0.00)}$	\\
(500, 250, 0.9)	&$3.03_{(0.07)}$	&$2.94_{(0.09)}$	&$2.84_{(0.07)}$	&$2.88_{(0.07)}$	&$0.99_{(0.00)}$	&$1.02_{(0.00)}$	\\
(1000, 500, 0.9)	&$3.08_{(0.08)}$	&$3.05_{(0.06)}$	&$2.93_{(0.03)}$	&$2.94_{(0.03)}$	&$0.99_{(0.00)}$	&$1.01_{(0.00)}$	\\
(250, 250, 0.9)	&$3.08_{(0.10)}$	&$3.03_{(0.13)}$	&$3.02_{(0.10)}$	&$2.98_{(1.66)}$	&$1.00_{(0.00)}$	&$1.04_{(0.00)}$	\\
(500, 500, 0.9)	&$3.12_{(0.14)}$	&$3.07_{(0.09)}$	&$3.03_{(0.05)}$	&$2.96_{(0.05)}$	&$0.99_{(0.00)}$	&$1.02_{(0.00)}$	\\\hline
(250, 500, 0.7)	&$3.07_{(0.09)}$	&$2.99_{(0.07)}$	&$3.25_{(0.19)}$	&$3.07_{(0.12)}$	&$1.02_{(0.00)}$	&$1.05_{(0.00)}$	\\
(500, 1000, 0.7)	&$3.09_{(0.12)}$	&$3.01_{(0.10)}$	&$3.13_{(0.07)}$	&$3.03_{(0.05)}$	&$1.01_{(0.00)}$	&$1.03_{(0.00)}$	\\
(500, 250, 0.7)	&$2.65_{(0.41)}$	&$2.47_{(0.54)}$	&$2.83_{(0.08)}$	&$2.92_{(0.07)}$	&$1.00_{(0.00)}$	&$1.03_{(0.00)}$	\\
(1000, 500, 0.7)	&$2.74_{(0.35)}$	&$2.62_{(0.43)}$	&$2.93_{(0.03)}$	&$2.95_{(0.03)}$	&$1.00_{(0.00)}$	&$1.02_{(0.00)}$	\\
(250, 250, 0.7)	&$2.97_{(0.19)}$	&$2.79_{(0.27)}$	&$3.08_{(0.13)}$	&$3.02_{(0.15)}$	&$1.01_{(0.00)}$	&$1.06_{(0.00)}$	\\
(500, 500, 0.7)	&$3.00_{(0.16)}$	&$2.87_{(0.20)}$	&$3.05_{(0.06)}$	&$3.01_{(0.05)}$	&$1.01_{(0.00)}$	&$1.03_{(0.00)}$	\\
\end{longtable}
}

{\footnotesize
\setlength\tabcolsep{3pt} 
\begin{longtable}{c|cc|cc|cc}
\caption{Means (with MSEs in brackets) of the estimated number of common factors for anisotropic Model 5 with heterogeneous missing data ($K=3$) based on 500 simulations.}
\label{heteroscedastic Model 5 with K=3}\\
\hline
\multirow{2}{*}{\makebox[0.13\textwidth][c]{$(d, n, p_1, p_2, p_3)$}} &\multicolumn{2}{c}{$\hat{r}$}  &\multicolumn{2}{|c}{$\hat{\theta}$}& \multicolumn{2}{|c}{$\hat{\sigma}^2$} \\ 
\cline{2-7}
& \makebox[0.05\textwidth][c]{MATE} & \makebox[0.04\textwidth][c]{BEMA}& \makebox[0.04\textwidth][c]{MATE} & \makebox[0.04\textwidth][c]{BEMA}&\makebox[0.04\textwidth][c]{MATE} & \makebox[0.04\textwidth][c]{BEMA} \\ 
\hline  
\endfirsthead

\multicolumn{7}{c}%
{{\tablename\ \thetable{} -- Continued from previous page}} \\
\hline
\multirow{2}{*}{\makebox[0.13\textwidth][c]{$(d, n, p_1, p_2, p_3)$}} &\multicolumn{2}{|c}{$\hat{r}$}  &\multicolumn{2}{|c}{$\hat{\theta}$}& \multicolumn{2}{|c}{$\hat{\sigma}^2$} \\ 
\cline{2-7}
& \makebox[0.05\textwidth][c]{MATE} & \makebox[0.04\textwidth][c]{BEMA}& \makebox[0.04\textwidth][c]{MATE} & \makebox[0.04\textwidth][c]{BEMA}&\makebox[0.04\textwidth][c]{MATE} & \makebox[0.04\textwidth][c]{BEMA} \\ 
\hline  
\endhead

\hline \multicolumn{7}{r}{\textit{Continued on next page}} \\
\endfoot

\hline
\endlastfoot

(250, 500, 0.9,0.8,0.7)	&$3.12_{(0.15)}$	&$3.04_{(0.04)}$	&$3.23_{(0.16)}$	&$2.99_{(0.09)}$	&$1.01_{(0.00)}$	&$0.83_{(0.03)}$	\\
(500, 1000, 0.9,0.8,0.7)	&$3.09_{(0.11)}$	&$3.05_{(0.05)}$	&$3.10_{(0.06)}$	&$2.95_{(0.05)}$	&$1.01_{(0.00)}$	&$0.82_{(0.03)}$	\\
(500, 250, 0.9,0.8,0.7)	&$2.96_{(0.14)}$	&$2.74_{(0.28)}$	&$2.88_{(0.07)}$	&$2.75_{(0.11)}$	&$1.00_{(0.00)}$	&$0.82_{(0.03)}$	\\
(1000, 500, 0.9,0.8,0.7)	&$3.03_{(0.10)}$	&$2.88_{(0.15)}$	&$2.94_{(0.03)}$	&$2.78_{(0.07)}$	&$1.00_{(0.00)}$	&$0.81_{(0.04)}$	\\
(250, 250, 0.9,0.8,0.7)	&$3.08_{(0.11)}$	&$2.92_{(0.12)}$	&$3.09_{(0.12)}$	&$2.88_{(0.11)}$	&$1.00_{(0.00)}$	&$0.84_{(0.03)}$	\\
(500, 500, 0.9,0.8,0.7)	&$3.10_{(0.10)}$	&$3.03_{(0.04)}$	&$3.04_{(0.05)}$	&$2.90_{(0.05)}$	&$1.00_{(0.00)}$	&$0.82_{(0.03)}$	\\\hline
(250, 500, 0.8,0.6,0.4)	&$3.08_{(0.21)}$	&$2.74_{(0.30)}$	&$3.32_{(0.25)}$	&$2.78_{(0.11)}$	&$1.03_{(0.00)}$	&$0.64_{(0.13)}$	\\
(500, 1000, 0.8,0.6,0.4)	&$3.04_{(0.19)}$	&$2.64_{(0.40)}$	&$3.16_{(0.09)}$	&$2.74_{(0.09)}$	&$1.01_{(0.00)}$	&$0.62_{(0.15)}$	\\
(500, 250, 0.8,0.6,0.4)	&$2.48_{(0.57)}$	&$1.86_{(1.46)}$	&$2.88_{(0.07)}$	&$1.92_{(1.19)}$	&$1.01_{(0.00)}$	&$0.62_{(0.14)}$	\\
(1000, 500, 0.8,0.6,0.4)	&$2.52_{(0.53)}$	&$1.93_{(1.22)}$	&$2.95_{(0.03)}$	&$1.94_{(1.12)}$	&$1.00_{(0.00)}$	&$0.61_{(0.15)}$	\\
 (250, 250, 0.8,0.6,0.4)	&$2.85_{(0.30)}$	&$2.35_{(0.67)}$	&$3.16_{(0.20)}$	&$2.47_{(0.37)}$	&$1.02_{(0.00)}$	&$0.64_{(0.13)}$	\\
(500, 500, 0.8,0.6,0.4)	&$2.88_{(0.26)}$	&$2.31_{(0.70)}$	&$3.08_{(0.07)}$	&$2.46_{(0.35)}$	&$1.01_{(0.00)}$	&$0.62_{(0.14)}$	\\
\end{longtable}
}

\subsection{Comparison with imputation-based methods}\label{imputation}
Imputation is a common approach for handling missing data. We first use complete data to select appropriate methods for estimating factor number in weak factor settings (see Table \ref{homoscedastic models with p=1}). Specifically, we consider seven approaches: MATE, $\widetilde{\rm CV}$, M-ED, M-GR, M-ER, M-PC, and M-IC. We then apply five imputation methods to incomplete data: zero imputation (ZI), mean imputation (MI), random forest (RF), expectation-maximization (EM) of \cite{jin2021factor}, and TP of \cite{cahan2023factor}. Next, we use the selected methods to estimate the number of factors on the imputed data, while MATE is directly applied to the original incomplete data. Finally, we compare the resulting estimates. We consider three models based on \eqref{simulation_X^o} with $r=5$ and $e_i\stackrel{i.i.d.}{\sim}\mathcal{N}(0,I_d)$: 
    \begin{itemize}
        \item[]\textbf{Model \ 6.} $\gamma=0.5$, $D_1={\rm diag}(3.5, 3, 2.5, 2.2, 1.8)$,
        \item[]\textbf{Model \ 7.} $\gamma=1$, $D_1={\rm diag}(4.5, 4, 3.5, 2.8, 2.2)$,
        \item[]\textbf{Model \ 8.} $\gamma=2$, $D_1={\rm diag}(4.5, 4, 3.5, 3, 2.5)$.
    \end{itemize}

{\footnotesize
\setlength\tabcolsep{3pt} 
\begin{longtable}{ccccccccc}
\caption{Means (with MSEs in brackets) of the estimated number of common factors for isotropic models with complete data based on 500 simulations.}
\label{homoscedastic models with p=1}\\
\hline
\makebox[0.05\textwidth][c]{Model} &\makebox[0.08\textwidth][c]{$(d,n)$}&\makebox[0.05\textwidth][c]{MATE} &\makebox[0.05\textwidth][c]{$\widetilde{\rm CV}$} &\makebox[0.05\textwidth][c]{M-ED}&\makebox[0.05\textwidth][c]{M-GR}&\makebox[0.05\textwidth][c]{M-ER}&\makebox[0.05\textwidth][c]{M-PC}&\makebox[0.05\textwidth][c]{M-IC}\\  
\hline  
\endfirsthead

\multicolumn{9}{c}%
{{\tablename\ \thetable{} -- Continued from previous page}} \\
\hline
\makebox[0.05\textwidth][c]{Model} &\makebox[0.08\textwidth][c]{$(d,n)$}&\makebox[0.05\textwidth][c]{MATE} &\makebox[0.05\textwidth][c]{$\widetilde{\rm CV}$} &\makebox[0.05\textwidth][c]{M-ED}&\makebox[0.05\textwidth][c]{M-GR}&\makebox[0.05\textwidth][c]{M-ER}&\makebox[0.05\textwidth][c]{M-PC}&\makebox[0.05\textwidth][c]{M-IC}\\ 
\hline  
\endhead

\hline \multicolumn{9}{r}{\textit{Continued on next page}} \\
\endfoot

\hline
\endlastfoot

\multirow{2}{*}{6}&(250, 500)&$5.02_{(0.02)}$ &$5.00_{(0.00)}$ &$4.98_{(0.04)}$ &$4.95_{(0.14)}$ &$4.97_{(0.07)} $&$2.96_{(4.24)} $&$2.43_{(6.83)}$\\
&(500, 1000)&$5.02_{(0.02)}$ &$5.00_{(0.00)}$	&$4.99_{(0.05)}$ &$5.00_{(0.00)}$	&$5.00_{(0.00)}$ &$2.04_{(8.81) }$	&$2.00_{(9.00)}$\\
\hline
\multirow{2}{*}{7}&(250, 250) &$5.02_{(0.02)}$ &$5.00_{(0.00)}$ &$4.99_{(0.02)} $&$4.81_{(0.34)}$ &$4.87_{(0.21)}$&$3.91_{(1.28)} $&$3.52_{(2.43)}$\\
& (500, 500)&$5.04_{(0.04)}$ &$5.00_{(0.00)}$ &$5.00_{(0.00)}$ &$4.94_{(0.11)}$ &$4.96_{(0.06)}$ &$3.11_{(3.66)}$ &$3.01_{(3.97)}$\\
\hline
\multirow{2}{*}{8}& (500, 250)&$5.02_{(0.02)}$ &$4.99_{(0.01)}$ &$4.69_{(1.10)}$ &$4.86_{(0.32)}$ &$4.89_{(0.23)}$ &$2.79_{(5.08)}$ &$2.29_{(7.55)}$\\
& (1000, 500)&$5.03_{(0.03)}$ &$5.00_{(0.00)}$	&$4.80_{(0.88)}$ &$4.99_{(0.04)}$	&$4.99_{(0.02)}$ &$2.03_{(8.89)}$	&$1.90_{(9.68)}$\\
\end{longtable}
}

From Table \ref{homoscedastic models with p=1}, we select four competitors ($\widetilde{\rm CV}$, M-ED, M-GR, M-ER) to examine the effects of the imputation on factor number estimation. Tables \ref{homoscedastic Model 6 with k=2}, \ref{homoscedastic Model 7 with k=2}, and \ref{homoscedastic Model 8 with k=2} report results for MATE in Models 6-8 with missing data ($K=2$, $(p_1,p_2)=(1,0.8),(1,0.6),(1,0.5)$\footnote{To ensure compatibility across all five imputation methods, we generate data with heterogeneous missingness ($K=2$), where the first half has no missing values ($p_1=1$), and the second half has non-missing rate $p_2$.}), and the results of the four competing methods applied to the corresponding rebalanced data imputed by the five imputation methods. Overall, MATE accurately estimates $r$ in all three models. $\widetilde{\rm CV}$ performs well in Models 6 and 7 but underestimates $r$ in Model 8. For M-ED, M-GR, and M-ER, the five imputation methods introduce varying degrees of errors, resulting in unreliable estimates. 
For Model 6 (Table \ref{homoscedastic Model 6 with k=2}), the performance of the imputation-based methods is notably affected by the missing rate, while MATE and $\widetilde{\rm CV}$ are stable. 
For Model 7 (Table \ref{homoscedastic Model 7 with k=2}), MATE and $\widetilde{\rm CV}$ remain accurate, while the others behave similarly to those in Model 6. 
For Model 8 (Table \ref{homoscedastic Model 8 with k=2}), MATE continues to perform well, whereas $\widetilde{\rm CV}$ underestimates $r$ in most cases. The remaining methods are relatively accurate only for TP-imputed data and deteriorate as $p_2$ decreases. Additionally, $\widetilde{\rm CV}$ is computationally expensive due to its iterative nature, whereas MATE computes the threshold directly and is therefore more efficient. Overall, in weak factor settings, MATE outperforms $\widetilde{\rm CV}$ in both accuracy and efficiency. 
Finally, we report the running times of these methods for Model~7 in Table~\ref{running_time}, which further confirms that MATE is the most efficient method.

{\footnotesize
\setlength\tabcolsep{3pt} 
\begin{longtable}{ccccccc}
\caption{Means (with MSEs in brackets) of the estimated number of common factors for isotropic Model 6 with heterogeneous missing data ($K=2$) based on 500 simulations.}
\label{homoscedastic Model 6 with k=2}\\
\hline
\makebox[0.05\textwidth][c]{Method} & \makebox[0.10\textwidth][c]{$(d,n,p_1,p_2)$} & \makebox[0.05\textwidth][c]{MATE} &\makebox[0.05\textwidth][c]{$\widetilde{\rm CV}$} & \makebox[0.05\textwidth][c]{M-ED} & \makebox[0.05\textwidth][c]{M-GR} & \makebox[0.05\textwidth][c]{M-ER}  \\   
\hline  
\endfirsthead

\multicolumn{7}{c}%
{{\tablename\ \thetable{} -- Continued from previous page}} \\
\hline
\makebox[0.05\textwidth][c]{Method} & \makebox[0.10\textwidth][c]{$(d,n,p_1,p_2)$} & \makebox[0.05\textwidth][c]{MATE} &\makebox[0.05\textwidth][c]{$\widetilde{\rm CV}$} & \makebox[0.05\textwidth][c]{M-ED} & \makebox[0.05\textwidth][c]{M-GR} & \makebox[0.05\textwidth][c]{M-ER}  \\  
\hline  
\endhead

\hline \multicolumn{7}{r}{\textit{Continued on next page}} \\
\endfoot

\hline
\endlastfoot

\multirow{6}{*}{ZI} & (250, 500, 1, 0.8) &$ 5.05_{(0.05)}$ & $5.00_{(0.00)}$ & $4.42_{(2.31)}$ & $4.76_{ (0.67)}$ & $4.87_{(0.32)}$  \\ 
&(500, 1000, 1, 0.8) & $5.04_{(0.04)}$ & $5.00_{(0.00)}$ & $4.42_{(2.68)}$ & $4.96_{(0.12)} $&$ 4.97_{(0.08)}$\\ 
&(250, 500, 1, 0.6) & $5.05_{(0.05)}$ & $5.00_{(0.00)}$ & $2.38_{(12.15)}$ &$ 4.21_{(2.47)}$ & $4.42_{(1.75)}$\\ 
&(500, 1000, 1, 0.6) & $5.08_{(0.08)}$ & $5.00_{(0.00)}$ &$ 1.88_{(15.32)}$ & $4.64_{(1.17)}$ & $4.73_{(0.85)}$ \\ 
& (250, 500, 1, 0.5)& $5.06_{(0.06)}$ & $4.99_{(0.01)}$ &$ 1.92_{(14.5)}$ &$ 3.95_{(3.39)}$ & $4.19_{(2.51)}$  \\ 
& (500, 1000, 1, 0.5)& $5.05_{(0.05)}$ &$ 5.00_{(0.00)}$ & $0.87_{(20.38)}$ & $4.42_{(1.91)} $& $4.53_{(1.54)}$ \\ 
\hline
\multirow{6}{*}{TP} &(250, 500, 1, 0.8)& $5.05_{ (0.05)}$ & $5.00_{(0.00)}$ & $4.97_{(0.10)}$ & $4.89_{(0.34)}$ & $4.93_{(0.19) }$ \\ 
& (500, 1000, 1, 0.8)& $5.06_{(0.06)} $& $5.00_{(0.00)}$ & $4.99_{(0.05)} $& $5.00_{ (0.00)}$ &$ 5.00_{(0.00)}$ \\ 
&(250, 500, 1, 0.6) & $5.06_{(0.06)}$ & $5.00_{(0.00)} $& $4.93_{(0.32)}$ & $4.78_{(0.64)} $& $4.88_{(0.32) }$ \\ 
& (500, 1000, 1, 0.6) & $5.07_{(0.07)} $& $5.00_{(0.00)} $& $4.98_{(0.08)}$& $4.96_{(0.12)}$ & $4.97_{(0.08) }$  \\ 
& (250, 500, 1, 0.5) & $5.07_{(0.07)}$ & $4.99_{(0.01)}$ & $6.50_{(5.07)}$ & $4.82_{(1.62)}$ & $5.07_{(1.14)}$ \\
& (500, 1000, 1, 0.5) & $5.06_{(0.06)}$ & $5.00_{(0.00)}$ & $6.19_{(3.94)}$ & $4.96_{(0.26)} $& $4.99_{(0.19)} $ \\ 
\hline
\multirow{6}{*}{EM} &(250, 500, 1, 0.8) &$ 5.04_{(0.04)}$ & $5.00_{(0.00)}$ & $4.81_{ (0.72)}$ & $4.75_{(0.81)}$ & $4.84_{(0.47)} $\\ 
&(500, 1000, 1, 0.8) & $5.05_{(0.05)}$ & $5.00_{(0.00)}$ & $4.89_{(0.47)}$ &$ 4.98_{(0.05)} $& $4.99_{(0.02)}$ \\ 
&(250, 500, 1, 0.6) & $5.07_{(0.07)}$ & $5.00_{(0.00)}$ & $7.96_{(8.96)}$ & $6.27_{(5.91)}$ & $6.80_{(6.26)}$\\ 
& (500, 1000, 1, 0.6) & $5.07_{(0.07)}$ & $5.00_{(0.00)}$ & $7.99_{(8.98)}$ & $6.58_{(5.07)}$ & $6.75_{(5.36)}$ \\ 
&(250, 500, 1, 0.5)& $5.07_{(0.07)}$ & $5.00_{(0.00)}$ & $8.00_{(9.00)}$ & $7.84_{(8.85)}$ & $7.92_{(8.91)} $\\ 
&(500, 1000, 1, 0.5)& $5.06_{(0.06)}$ &$ 5.00_{(0.00)}$ & $8.00_{(9.00)}$ & $7.96_{(8.94)}$ & $7.98_{(8.95)}$  \\ 
\hline
\multirow{6}{*}{RF} & (250, 500, 1, 0.8) & $5.06_{(0.06)}$ & $5.00_{(0.00)}$ & $4.71_{(0.96)}$ & $4.71_{(0.87)}$ &$ 4.81_{(0.50)}$ \\
&(500, 1000, 1, 0.8) &$ 5.05_{(0.05)}$ & $5.00_{(0.00)}$ & $4.74_{(1.12)}$ & $4.98_{(0.04)}$ & $4.99_{(0.02)}$ \\ 
& (250, 500, 1, 0.6) & $5.06_{(0.06)}$ & $5.00_{(0.00)}$ & $3.89_{(4.15)}$& $4.06_{(3.05)}$ & $4.29_{(2.19)}$  \\ 
&(500, 1000, 1, 0.6) & $5.07_{(0.07)}$ & $5.00_{(0.00)}$ & $2.71_{(10.49)}$ &$ 4.41_{(1.92)}$ & $4.54_{(1.46)}$  \\ 
& (250, 500, 1, 0.5)& $5.02_{(0.02)}$ & $5.00_{(0.00)} $& $3.16_{(6.98)}$ & $3.32_{(5.65)}$ & $3.67_{(4.32)}$ \\
& (500, 1000, 1, 0.5) & $5.05_{(0.05)}$ & $5.00_{(0.00)}$ & $1.98_{(14.00)}$ & $4.14_{(2.85)}$ &$ 4.25_{(2.44)}$\\ 
\hline
\multirow{6}{*}{MI} &(250, 500, 1, 0.8) & $5.05_{(0.05)}$ &$5.00_{(0.00)}$ & $4.46_{(2.13)} $& $4.69_{(0.98)}$ & $4.80_{(0.61)}$ \\ 
& (500, 1000, 1, 0.8) & $5.07_{(0.07)}$ & $5.00_{(0.00)}$ & $4.36_{(3.02)} $& $4.97_{(0.11 )}$ & $4.97_{(0.07)}$  \\ 
& (250, 500, 1, 0.6) & $5.08_{(0.08)}$ & $5.00_{(0.00)}$ &$ 2.53_{(11.35)} $& $4.18_{(2.58)}$ & $4.43_{(1.73)}$  \\ 
&(500, 1000, 1, 0.6) & $5.07_{(0.07)}$ & $5.00_{(0.00)}$ & $1.69_{(16.27)}$ & $4.69_{(0.93)} $& $4.79_{(0.57)}$  \\ 
& (250, 500, 1, 0.5) & $5.08_{(0.08)}$ & $5.00_{(0.00)} $& $1.79_{(15.09)}$ & $3.85_{(3.68)}$ & $4.11_{(2.75)}$ \\ 
& (500, 1000, 1, 0.5) & $5.07_{(0.07)}$ & $5.00_{(0.00)}$ & $0.99_{(19.77)}$ &$ 4.45_{(1.79)}$ & $4.55_{(1.41)}$\\ 
\end{longtable}
}

{\footnotesize
\setlength\tabcolsep{3pt} 
\begin{longtable}{ccccccc}
\caption{Means (with MSEs in brackets) of the estimated number of common factors for isotropic Model 7 with heterogeneous missing data ($K=2$) based on 500 simulations.}
\label{homoscedastic Model 7 with k=2}\\
\hline
\makebox[0.05\textwidth][c]{Method} & \makebox[0.10\textwidth][c]{$(d,n,p_1,p_2)$} & \makebox[0.05\textwidth][c]{MATE} &\makebox[0.05\textwidth][c]{$\widetilde{\rm CV}$} & \makebox[0.05\textwidth][c]{M-ED} & \makebox[0.05\textwidth][c]{M-GR} & \makebox[0.05\textwidth][c]{M-ER}  \\ 
\hline  
\endfirsthead

\multicolumn{7}{c}%
{{\tablename\ \thetable{} -- Continued from previous page}} \\
\hline
\makebox[0.05\textwidth][c]{Method} & \makebox[0.10\textwidth][c]{$(d,n,p_1,p_2)$} & \makebox[0.05\textwidth][c]{MATE} &\makebox[0.05\textwidth][c]{$\widetilde{\rm CV}$} & \makebox[0.05\textwidth][c]{M-ED} & \makebox[0.05\textwidth][c]{M-GR} & \makebox[0.05\textwidth][c]{M-ER}  \\  
\hline  
\endhead

\hline \multicolumn{7}{r}{\textit{Continued on next page}} \\
\endfoot

\hline
\endlastfoot

\multirow{6}{*}{ZI}	&(250, 250, 1, 0.8)	&$5.05_{(0.05)}$	&$4.96_{(0.04)}$	&$4.90_{(0.12)}$	&$4.56_{(0.87)}$	&$4.67_{(0.59)}$	\\
&(500, 500, 1, 0.8)	&$5.06_{(0.06)}$	&$5.00_{(0.00)}$	&$4.93_{(0.07)}$	&$4.80_{(0.33)}$	&$4.84_{(0.24)}$	\\
&(250, 250, 1, 0.6)	&$5.04_{(0.05)}$	&$4.86_{(0.14)}$	&$4.52_{(0.90)}$	&$4.25_{(1.56)}$	&$4.39_{(1.14)}$	\\
&(500, 500, 1, 0.6)	&$5.08_{(0.08)}$	&$4.97_{(0.03)}$	&$4.44_{(1.01)}$	&$4.50_{(0.94)}$	&$4.54_{(0.85)}$	\\
&(250, 250, 1, 0.5)	&$5.05_{(0.05)}$	&$4.76_{(0.24)}$	&$4.10_{(2.09)}$	&$3.95_{(2.38)}$	&$4.15_{(1.79)}$	\\
&(500, 500, 1, 0.5)	&$5.05_{(0.05)}$	&$4.94_{(0.06)}$	&$3.97_{(2.68)}$	&$4.26_{(1.44)}$	&$4.37_{(1.14)}$	\\ \hline
\multirow{6}{*}{TP}	&(250, 250, 1, 0.8)	&$5.07_{(0.07)}$	&$4.96_{(0.04)}$	&$4.98_{(0.02)}$	&$4.69_{(0.50)}$	&$4.77_{(0.34)}$	\\
&(500, 500, 1, 0.8)	&$5.05_{(0.05)}$	&$5.00_{(0.00)}$	&$5.00_{(0.00)}$	&$4.92_{(0.12)}$	&$4.94_{(0.10)}$	\\
&(250, 250, 1, 0.6)	&$5.11_{(0.11)}$	&$4.88_{(0.12)}$	&$4.96_{(0.04)}$	&$4.53_{(0.95)}$	&$4.70_{(0.53)}$	\\
&(500, 500, 1, 0.6)	&$5.05_{(0.05)}$	&$4.97_{(0.03)}$	&$5.00_{(0.00)}$	&$4.83_{(0.29)}$	&$4.85_{(0.24)}$	\\
&(250, 250, 1, 0.5)	&$5.04_{(0.04)}$	&$4.77_{(0.23)}$	&$5.73_{(2.39)}$	&$4.39_{(1.20)}$	&$4.58_{(0.83)}$	\\
&(500, 500, 1, 0.5)	&$5.06_{(0.06)}$	&$4.94_{(0.06)}$	&$5.10_{(0.33)}$	&$4.77_{(0.37)}$	&$4.82_{(0.28)}$	\\ \hline
\multirow{6}{*}{EM}	&(250, 250, 1, 0.8)	&$5.05_{(0.05)}$	&$4.99_{(0.01)}$	&$4.95_{(0.05)}$	&$4.59_{(0.76)}$	&$4.71_{(0.48)}$	\\
&(500, 500, 1, 0.8)	&$5.05_{(0.05)}$	&$5.00_{(0.00)}$	&$4.99_{(0.01)}$	&$4.84_{(0.27)}$	&$4.87_{(0.20)}$	\\
&(250, 250, 1, 0.6)	&$5.06_{(0.06)}$	&$4.99_{(0.01)}$	&$7.65_{(8.03)}$	&$4.50_{(2.29)}$	&$4.94_{(2.10)}$	\\
&(500, 500, 1, 0.6)	&$5.05_{(0.05)}$	&$5.00_{(0.00)}$	&$7.36_{(7.14)}$	&$4.70_{(1.04)}$	&$4.87_{(1.08)}$	\\
&(250, 250, 1, 0.5)	&$5.05_{(0.05)}$	&$4.99_{(0.01)}$	&$7.99_{(8.98)}$	&$5.97_{(5.49)}$	&$6.58_{(5.98)}$	\\
&(500, 500, 1, 0.5)	&$5.04_{(0.04)}$	&$5.00_{(0.00)}$	&$7.99_{(8.98)}$	&$6.12_{(4.90)}$	&$6.46_{(5.40)}$	\\ \hline
\multirow{6}{*}{RF}	&(250, 250, 1, 0.8)	&$5.06_{(0.06)}$	&$4.99_{(0.01)}$	&$4.93_{(0.08)}$	&$4.53_{(0.91)}$	&$4.65_{(0.60)}$	\\
&(500, 500, 1, 0.8)	&$5.07_{(0.07)}$	&$5.00_{(0.00)}$	&$4.97_{(0.03)}$	&$4.80_{(0.32)}$	&$4.82_{(0.28)}$	\\
&(250, 250, 1, 0.6)	&$5.09_{(0.09)}$	&$4.99_{(0.01)}$	&$4.77_{(0.28)}$	&$4.18_{(1.70)}$	&$4.36_{(1.15)}$	\\
&(500, 500, 1, 0.6)	&$5.06_{(0.06)}$	&$5.00_{(0.00)}$	&$4.79_{(0.27)}$	&$4.34_{(1.33)}$	&$4.42_{(1.14)}$	\\
&(250, 250, 1, 0.5)	&$5.06_{(0.06)}$	&$4.98_{(0.03)}$	&$4.62_{(0.53)}$	&$3.89_{(2.45)}$	&$4.12_{(1.77)}$	\\
&(500, 500, 1, 0.5)	&$5.08_{(0.08)}$	&$5.00_{(0.00)}$	&$4.63_{(0.47)}$	&$4.06_{(1.97)}$	&$4.17_{(1.63)}$	\\ \hline
\multirow{6}{*}{MI}	&(250, 250, 1, 0.8)	&$5.04_{(0.04)}$	&$4.97_{(0.03)}$	&$4.86_{(0.17)}$	&$4.49_{(1.00)}$	&$4.57_{(0.81)}$	\\
&(500, 500, 1, 0.8)	&$5.05_{(0.05)}$	&$5.00_{(0.00)}$	&$4.95_{(0.05)}$	&$4.84_{(0.26)}$	&$4.86_{(0.22)}$	\\
&(250, 250, 1, 0.6)	&$5.10_{(0.10)}$	&$4.94_{(0.06)}$	&$4.43_{(1.01)}$	&$4.20_{(1.65)}$	&$4.34_{(1.26)}$	\\
&(500, 500, 1, 0.6)	&$5.08_{(0.08)}$	&$4.99_{(0.01)}$	&$4.42_{(1.16)}$	&$4.41_{(1.09)}$	&$4.50_{(0.88)}$	\\
&(250, 250, 1, 0.5)	&$5.07_{(0.07)}$	&$4.90_{(0.10)}$	&$4.16_{(1.92)}$	&$3.94_{(2.51)}$	&$4.17_{(1.74)}$	\\
&(500, 500, 1, 0.5)	&$5.07_{(0.07)}$	&$4.99_{(0.01)}$	&$3.98_{(2.66)}$	&$4.31_{(1.27)}$	&$4.40_{(1.05)}$	\\
\end{longtable}
}

{\footnotesize
\setlength\tabcolsep{3pt} 
\begin{longtable}{ccccccc}
\caption{Means (with MSEs in brackets) of the estimated number of common factors for isotropic Model 8 with heterogeneous missing data ($K=2$) based on 500 simulations.}
\label{homoscedastic Model 8 with k=2}\\
\hline
\makebox[0.05\textwidth][c]{Method} & \makebox[0.10\textwidth][c]{$(d,n,p_1,p_2)$} & \makebox[0.05\textwidth][c]{MATE} &\makebox[0.05\textwidth][c]{$\widetilde{\rm CV}$} & \makebox[0.05\textwidth][c]{M-ED} & \makebox[0.05\textwidth][c]{M-GR} & \makebox[0.05\textwidth][c]{M-ER}  \\ 
\hline  
\endfirsthead

\multicolumn{7}{c}%
{{\tablename\ \thetable{} -- Continued from previous page}} \\
\hline
\makebox[0.05\textwidth][c]{Method} & \makebox[0.10\textwidth][c]{$(d,n,p_1,p_2)$} & \makebox[0.05\textwidth][c]{MATE} &\makebox[0.05\textwidth][c]{$\widetilde{\rm CV}$} & \makebox[0.05\textwidth][c]{M-ED} & \makebox[0.05\textwidth][c]{M-GR} & \makebox[0.05\textwidth][c]{M-ER}  \\  
\hline  
\endhead

\hline \multicolumn{7}{r}{\textit{Continued on next page}} \\
\endfoot

\hline
\endlastfoot

\multirow{6}{*}{ZI}	&(250, 250, 1, 0.8)	&$5.03_{(0.03)}$	&$4.82_{(0.18)}$	&$3.30_{(7.72)}$	&$4.57_{(1.17)}$	&$4.70_{(0.77)}$	\\
&(500, 500, 1, 0.8)	&$5.05_{(0.05)}$	&$4.93_{(0.07)}$	&$3.20_{(8.75)}$	&$4.92_{(0.22)}$	&$4.93_{(0.15)}$	\\
&(250, 250, 1, 0.6)	&$5.04_{(0.04)}$	&$4.60_{(0.40)}$	&$1.41_{(17.21)}$	&$4.09_{(2.68)}$	&$4.23_{(2.21)}$	\\
&(500, 500, 1, 0.6)	&$5.03_{(0.03)}$	&$4.72_{(0.28)}$	&$0.59_{(21.94)}$	&$4.60_{(1.12)}$	&$4.67_{(0.85)}$	\\
&(250, 250, 1, 0.5)	&$5.03_{(0.03)}$	&$4.43_{(0.57)}$	&$0.83_{(20.24)}$	&$3.79_{(3.63)}$	&$3.97_{(2.98)}$	\\
&(500, 500, 1, 0.5)	&$5.03_{(0.03)}$	&$4.49_{(0.51)}$	&$0.26_{(23.52)}$	&$4.31_{(2.00)}$	&$4.40_{(1.68)}$	\\
\hline
\multirow{6}{*}{TP}	&(250, 250, 1, 0.8)	&$5.04_{(0.04)}$	&$4.81_{(0.19)}$	&$4.71_{(1.03)}$	&$4.80_{(0.45)}$	&$4.85_{(0.32)}$	\\
&(500, 500, 1, 0.8)	&$5.04_{(0.04)}$	&$4.93_{(0.07)}$	&$4.78_{(0.98)}$	&$4.97_{(0.05)}$	&$4.97_{(0.04)}$	\\
&(250, 250, 1, 0.6)	&$5.03_{(0.03)}$	&$4.57_{(0.43)}$	&$4.53_{(1.62)}$	&$4.58_{(1.11)}$	&$4.67_{(0.84)}$	\\
&(500, 500, 1, 0.6)	&$5.04_{(0.04)}$	&$4.73_{(0.27)}$	&$4.73_{(1.11)}$	&$4.91_{(0.19)}$	&$4.93_{(0.13)}$	\\
&(250, 250, 1, 0.5)	&$5.03_{(0.03)}$	&$4.42_{(0.58)}$	&$4.30_{(2.86)}$	&$4.49_{(1.33)}$	&$4.63_{(0.91)}$	\\
&(500, 500, 1, 0.5)	&$5.03_{(0.03)}$	&$4.49_{(0.51)}$	&$4.48_{(2.37)}$	&$4.91_{(0.15)}$	&$4.92_{(0.14)}$	\\
\hline
\multirow{6}{*}{EM}	&(250, 250, 1, 0.8)	&$5.05_{(0.05)}$	&$4.94_{(0.06)}$	&$4.22_{(3.41)}$	&$4.63_{(1.05)}$	&$4.74_{(0.69)}$	\\
&(500, 500, 1, 0.8)	&$5.05_{(0.05)}$	&$5.00_{(0.00)}$	&$4.31_{(3.22)}$	&$4.91_{(0.17)}$	&$4.93_{(0.12)}$	\\
&(250, 250, 1, 0.6)	&$5.04_{(0.04)}$	&$4.92_{(0.08)}$	&$4.72_{(8.28)}$	&$4.75_{(2.75)}$	&$5.12_{(2.61)}$	\\
&(500, 500, 1, 0.6)	&$5.06_{(0.06)}$	&$4.99_{(0.01)}$	&$3.39_{(9.57)}$	&$5.13_{(1.31)}$	&$5.26_{(1.39)}$	\\
&(250, 250, 1, 0.5)	&$5.02_{(0.02)}$	&$4.97_{(0.03)}$	&$7.01_{(9.39)}$	&$6.22_{(5.49)}$	&$6.66_{(5.93)}$	\\
&(500, 500, 1, 0.5)	&$5.04_{(0.04)}$	&$4.99_{(0.01)}$	&$6.11_{(9.98)}$	&$6.16_{(4.35)}$	&$6.39_{(4.75)}$	\\
\hline
\multirow{6}{*}{RF}	&(250, 250, 1, 0.8)	&$5.04_{(0.04)}$	&$4.93_{(0.07)}$	&$4.02_{(4.11)}$	&$4.61_{(1.12)}$	&$4.70_{(0.81)}$	\\
&(500, 500, 1, 0.8)	&$5.05_{(0.05)}$	&$4.99_{(0.01)}$	&$3.51_{(7.08)}$	&$4.90_{(0.27)}$	&$4.91_{(0.22)}$	\\
&(250, 250, 1, 0.6)	&$5.03_{(0.03)}$	&$4.91_{(0.09)}$	&$2.89_{(9.05)}$	&$3.98_{(3.11)}$	&$4.19_{(2.29)}$	\\
&(500, 500, 1, 0.6)	&$5.05_{(0.05)}$	&$4.98_{(0.02)}$	&$1.48_{(17.12)}$	&$4.56_{(1.20)}$	&$4.62_{(0.96)}$	\\
&(250, 250, 1, 0.5)	&$5.02_{(0.02)}$	&$4.94_{(0.06)}$	&$2.63_{(9.66)}$	&$3.66_{(3.88)}$	&$3.88_{(3.13)}$	\\
&(500, 500, 1, 0.5)	&$5.02_{(0.02)}$	&$4.99_{(0.01)}$	&$0.74_{(20.81)}$	&$4.14_{(2.52)}$	&$4.26_{(2.12)}$	\\
\hline
\multirow{6}{*}{MI}	&(250, 250, 1, 0.8)	&$5.04_{(0.04)}$	&$4.83_{(0.17)}$	&$3.24_{(7.94)}$	&$4.54_{(1.35)}$	&$4.67_{(0.87)}$	\\
&(500, 500, 1, 0.8)	&$5.03_{(0.03)}$	&$4.95_{(0.05)}$	&$3.02_{(9.73)}$	&$4.92_{(0.21)}$	&$4.93_{(0.19)}$	\\
&(250, 250, 1, 0.6)	&$5.04_{(0.04)}$	&$4.75_{(0.25)}$	&$1.41_{(17.29)}$	&$4.12_{(2.63)}$	&$4.23_{(2.27)}$	\\
&(500, 500, 1, 0.6)	&$5.04_{(0.04)}$	&$4.87_{(0.13)}$	&$0.60_{(21.9)}$	&$4.62_{(1.04)}$	&$4.69_{(0.81)}$	\\
&(250, 250, 1, 0.5)	&$5.02_{(0.02)}$	&$4.69_{(0.31)}$	&$0.81_{(20.3)}$	&$3.90_{(3.18)}$	&$4.07_{(2.56)}$	\\
&(500, 500, 1, 0.5)	&$5.05_{(0.05)}$	&$4.75_{(0.25)}$	&$0.16_{(24.13)}$	&$4.28_{(2.04)}$	&$4.41_{(1.59)}$	\\
\end{longtable}
}

{\footnotesize
\setlength\tabcolsep{3pt} 
\begin{longtable}{ccccccc}
\caption{ The running time (in seconds) for isotropic Model 7 with heterogeneous missing data ($K=2$).}
\label{running_time}\\
\hline
\makebox[0.05\textwidth][c]{Method} & \makebox[0.10\textwidth][c]{$(d,n,p_1,p_2)$} & \makebox[0.05\textwidth][c]{MATE} &\makebox[0.05\textwidth][c]{$\widetilde{\rm CV}$} & \makebox[0.05\textwidth][c]{M-ED} & \makebox[0.05\textwidth][c]{M-GR} & \makebox[0.05\textwidth][c]{M-ER}  \\ 
\hline  
\endfirsthead

\multicolumn{7}{c}%
{{\tablename\ \thetable{} -- Continued from previous page}} \\
\hline
\makebox[0.05\textwidth][c]{Method} & \makebox[0.10\textwidth][c]{$(d,n,p_1,p_2)$} & \makebox[0.05\textwidth][c]{MATE} &\makebox[0.05\textwidth][c]{$\widetilde{\rm CV}$} & \makebox[0.05\textwidth][c]{M-ED} & \makebox[0.05\textwidth][c]{M-GR} & \makebox[0.05\textwidth][c]{M-ER}  \\  
\hline  
\endhead

\hline \multicolumn{7}{r}{\textit{Continued on next page}} \\
\endfoot

\hline
\endlastfoot

\multirow{2}{*}{ZI}	&(250, 250, 1, 0.5)	&$0.02$	&$4.21$	&$0.06$	&$0.06$	&$0.05$	\\
                  &(500, 500, 1, 0.5)	&$0.17$	 &$28.92$ &$0.30$ &$0.28$ &$0.30$	\\
\hline
\multirow{2}{*}{TP}	&(250, 250, 1, 0.5)	&$0.02$	&$4.56$	&$0.28$	&$0.30$	&$0.29$	\\
&(500, 500, 1, 0.5)	&$0.17$	 &$29.36$ &$0.94$ &$0.94$ &$0.94$	\\
\hline
\multirow{2}{*}{EM}	&(250, 250, 1, 0.5)	&$0.02$	&$4.20$	&$0.19$	&$0.18$	&$0.17$	\\
&(500, 500, 1, 0.5)	&$0.17$	 &$29.40$ &$1.14$ &$1.14$ &$1.14$	\\
\hline
\multirow{2}{*}{RF}	&(250, 250, 1, 0.5)	&$0.02$	&$50.08$ &$45.99$ &$36.70$ &$37.12$	\\
&(500, 500, 1, 0.5)	&$0.17$	 &$357.32$ &$333.02$ &$329.22$ &$329.70$	\\
\hline
\multirow{2}{*}{MI}	&(250, 250, 1, 0.5)	&$0.02$	&$4.11$	&$0.06$	&$0.05$	&$0.05$	\\
&(500, 500, 1, 0.5)	&$0.17$	 &$28.66$ &$0.30$ &$0.30$ &$0.28$	\\
\end{longtable}
}

\section{Real data examples}\label{real_data}
In this section, we use two Fama-French datasets to demonstrate the strong performance of MATE.

\subsection{Data description}
We use daily and monthly returns of the 100 portfolios in the dataset \textbf{``100 Portfolios Formed on Size and Book-to-Market (10 $\times$ 10)"} from Kenneth R. French's website\footnote{\url{https://mba.tuck.dartmouth.edu/pages/faculty/ken.french/data_library.html}}. The first dataset consists of monthly returns from 07/1963 to 02/2024, after removing periods with missing observations. We standardize the data to have zero mean and unit variance, obtaining a balanced panel $X$ with $(d,n)=(100, 620)$. The second dataset contains fully observed daily returns from 01/03/2019 to 29/02/2024. After the same standardization, we obtain a balanced panel $Y$ with $(d,n)=(100,1259)$. 

Figure \ref{realdata_eigenvalues} displays the eigenvalues of the sample covariance matrices of $X$ and $Y$. In the left panel, the largest eigenvalue is a clear outlier, while the others are closely clustered. In the right panel, the largest eigenvalue is also clearly separated, and the second and third eigenvalues are mildly separated from the bulk. Developing suitable models based on unbalanced panels to explain portfolio returns is of economic interest. We set the random seed to 10. Each entry of $X$ is independently masked with probability $0.3$, yielding an incomplete dataset $X^o$ with non-missing rate $p=0.7$. The same procedure is applied to $Y$ to obtain $Y^o$ with $p=0.7$. 

\begin{figure}[H]
\centering
\caption{$100$ eigenvalues of the sample covariance matrix of $X$(left), and those of $Y$(right). }\label{realdata_eigenvalues}
\includegraphics[width=0.4\linewidth]{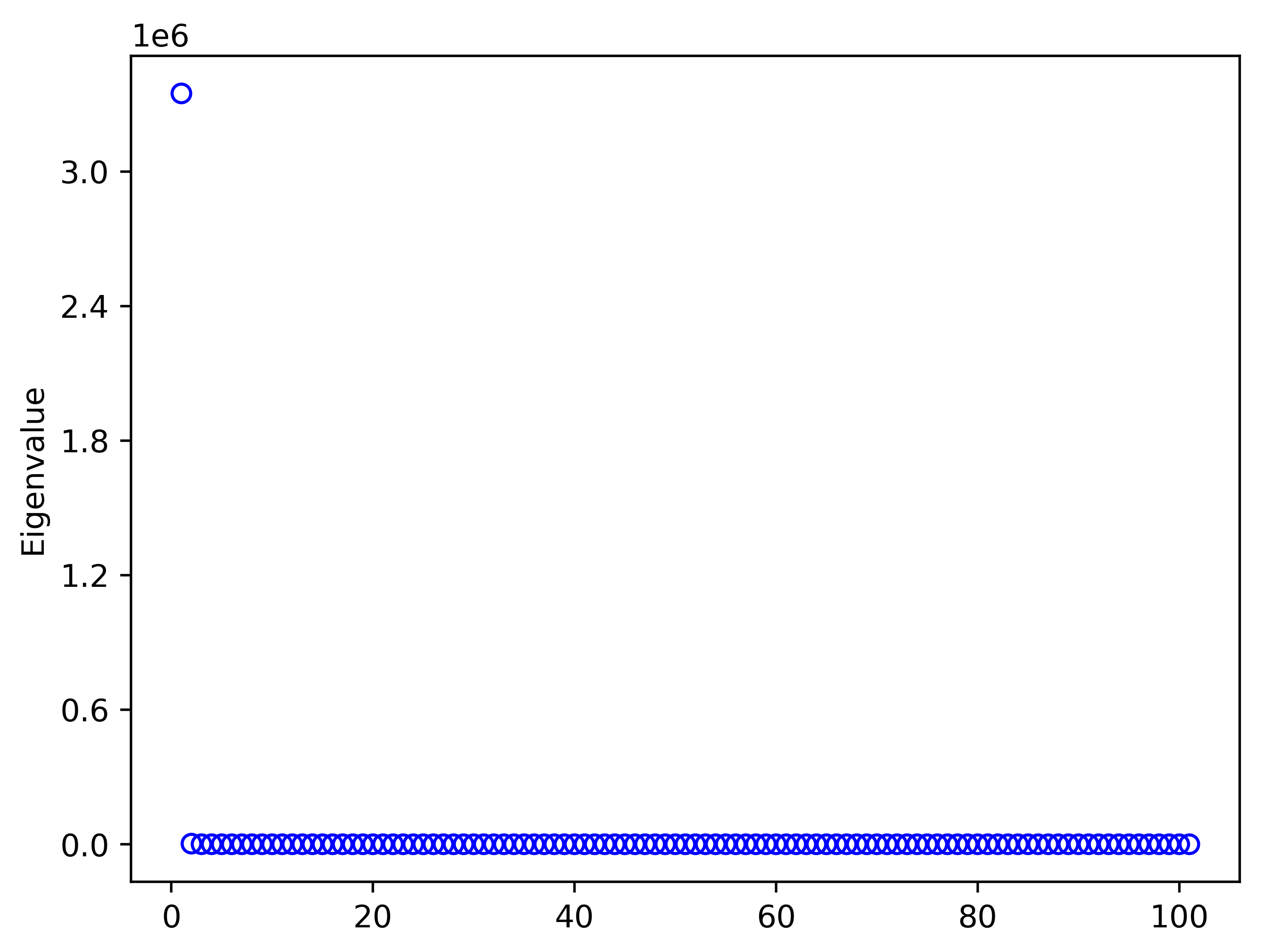}
\includegraphics[width=0.4\linewidth]{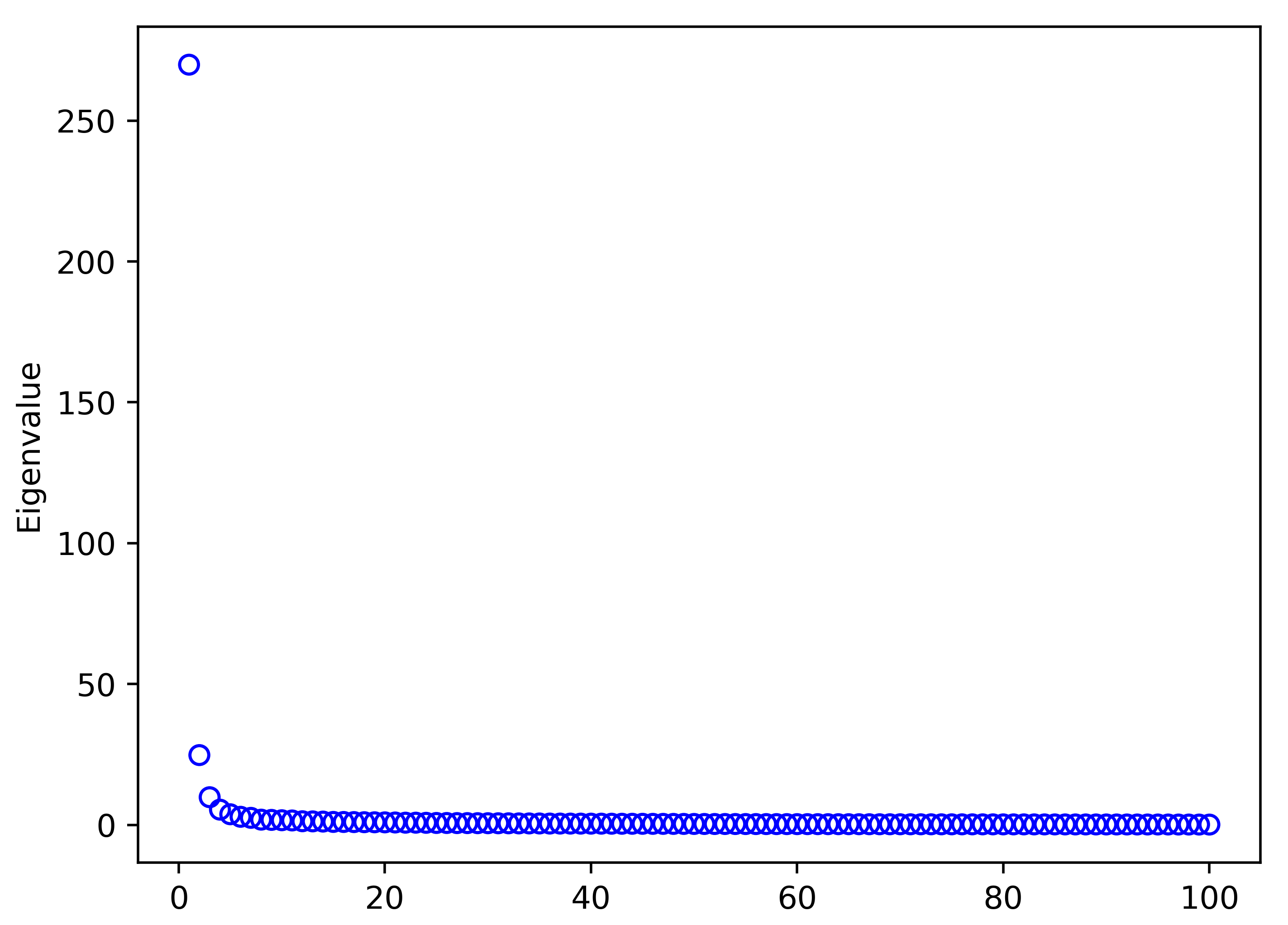}
\end{figure}

\subsection{Monthly data example $X$}\label{Monthly_data}
For $X^o$, we apply the methods used in the simulations to estimate the number of factors and record the computation time (in seconds) (Table \ref{realdataX_estimate_r}). MATE, $\widetilde{\rm CV}$ and M-ED yield the same estimate $\hat{r}=3$; M-GR and M-ER report only one factor; M-PC and M-IC report five factors; and BEMA reports four factors. Computationally, BEMA is the most expensive because it involves optimization and quantile calculations, whereas $\widetilde{\rm CV}$, M-PC and M-IC are relatively costly due to their iterations.

\begin{table}[H]
\setlength\tabcolsep{3pt} 
\centering
\footnotesize 
\caption{Estimated number of factors and the running time (in seconds) for panel $X^o$ across methods. Computation performed on an AMD Ryzen 5 4600U (6 cores, 2.10 GHz) with 16 GB RAM.}\label{realdataX_estimate_r}
\begin{tabular}{ccccccccc}
\hline
\makebox[0.03\textwidth][c]{} & \makebox[0.08\textwidth][c]{MATE} & \makebox[0.08\textwidth][c]{$\widetilde{\rm CV}$} &\makebox[0.08\textwidth][c]{M-ED} & \makebox[0.08\textwidth][c]{M-GR} & \makebox[0.08\textwidth][c]{M-ER} & \makebox[0.08\textwidth][c]{M-PC}& \makebox[0.08\textwidth][c]{M-IC} & \makebox[0.08\textwidth][c]{BEMA} \\ 
\hline
$\hat{r}$ & 3   &  3   & 3    & 1     & 1       & 5       & 5      & 4
\\	 	
\hline
Running time (s)   & 0.05 & 7.71 & 0.08 & 0.08 &0.06  & 1.00  & 1.00 &   251.19 \\
\hline
\end{tabular}
\end{table}  

We verify the number of factors in $X$ by measuring imputation performance. First, we replace missing values in $X^o$ with zeros. Second, given a candidate $\hat{r}$, we estimate $\hat{r}$-dimensional factors $\hat{f}_i$ and $\hat{r}$-dimensional factor loadings $\hat{\Lambda}_j$ from $X^o$, where $1 \le i \le d$ and $1 \le j \le n$, and impute missing entries by $\hat{x}_{ij}=\hat{\Lambda}_i^\top\hat{f_j}$. Third, we compute the imputation error ${\rm RMSE}=[(n^{\rm miss})^{-1}\sum\nolimits_{(i,j)\in \Omega_{\bot}} (x_{ij}-\hat{x}_{ij})^2]^{1/2}$, where $x_{ij}$ is the $(i,j)$-th entry of the ground truth $X$, $\Omega_{\bot}=\{(i,j)\in [d] \times [n]: x_{ij} \ {\rm is\ missing}\}$, and $n^{\rm miss}$ is the total number of missing entries. Finally, we repeat this procedure $100$ times and report the average RMSE (ARMSE). The left panel of Figure \ref{realdata_RMSE} shows ARMSEs for different $\hat{r}$. As $\hat{r}$ increases from 1 to 3, ARMSE decreases sharply and reaches its minimum at $\hat{r}=3$, after which it increases gradually. This suggests that $\hat{r}=3$ is a reasonable estimate for this dataset. 

\subsection{Daily data example $Y$}\label{Daily_data}

A long sample period may introduce structural changes. To address this, we use the daily dataset $Y$, which shares similar characteristics with $X$, to further validate $\hat{r}$. The factor number in $Y$ is estimated using the same procedure as in Section~\ref{Monthly_data} (Table \ref{realdataY_estimate_r}). MATE, along with $\widetilde{\rm CV}$ and M-ED, consistently yields $\hat{r}=3$. In contrast, M-GR and M-ER identify only one factor, while M-PC and M-IC estimate six factors. 
From a computational perspective, $\widetilde{\rm CV}$, M-PC, and M-IC are time-intensive due to their iterative procedures, whereas BEMA is particularly expensive because it involves optimization and quantile calculations. 

    \begin{table}[H]
    \setlength\tabcolsep{3pt} 
        \centering
        \footnotesize 
        \caption{Estimated number of factors and the running time (in seconds) for panel $Y^o$ across methods. Computation performed on an AMD Ryzen 5 4600U (6 cores, 2.10 GHz) with 16 GB RAM.}
       \label{realdataY_estimate_r}
        \begin{tabular}{ccccccccc}
        \hline
      \makebox[0.03\textwidth][c]{} & \makebox[0.08\textwidth][c]{MATE} & \makebox[0.08\textwidth][c]{$\widetilde{CV}$} & \makebox[0.08\textwidth][c]{M-ED} & \makebox[0.08\textwidth][c]{M-GR} & \makebox[0.08\textwidth][c]{M-ER} & \makebox[0.08\textwidth][c]{M-PC}& \makebox[0.08\textwidth][c]{M-IC} & \makebox[0.08\textwidth][c]{BEMA} \\ 
      \hline
       $\hat{r}$ & 3   &  3    & 3    & 1     & 1       & 6       & 6      &  5 \\	 	
     \hline
       Running time (s)   & 0.04 & 14.69 & 0.08 & 0.06 & 0.07  & 1.93  & 1.74 &  2452.02  \\
     \hline
    \end{tabular}
    \end{table}	

We apply the same procedure as in Section \ref{Monthly_data} to validate $\hat{r}$ for $Y$, and the ARMSEs (Figure \ref{realdata_RMSE}) exhibit a trend consistent with those for $X$. According to \cite{fama1993common}, the Fama-French three-factor model (Rm-Rf, SMB, and HML) effectively explains portfolio returns. Together with the findings of \cite{yu2024testing} and  the results from the two examples above, this supports $\hat{r}=3$ as a suitable estimate. 

    \begin{figure}[H]
        \centering
        \caption{ARMSEs over $100$ repetitions for $X^o$(left) and $Y^o$(right).} 
        \label{realdata_RMSE}
        \includegraphics[width=0.4\linewidth]{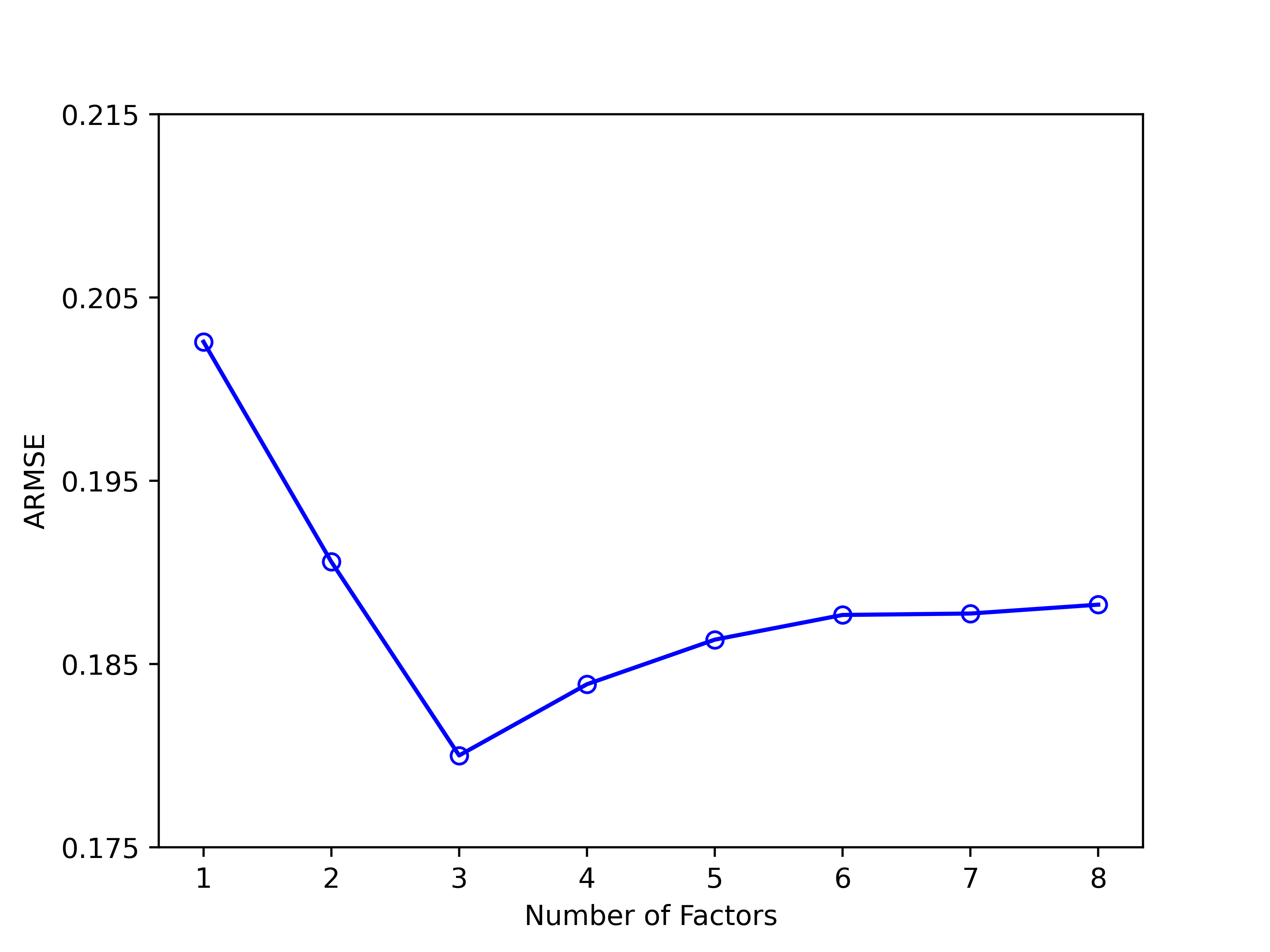}
        \includegraphics[width=0.4\linewidth]{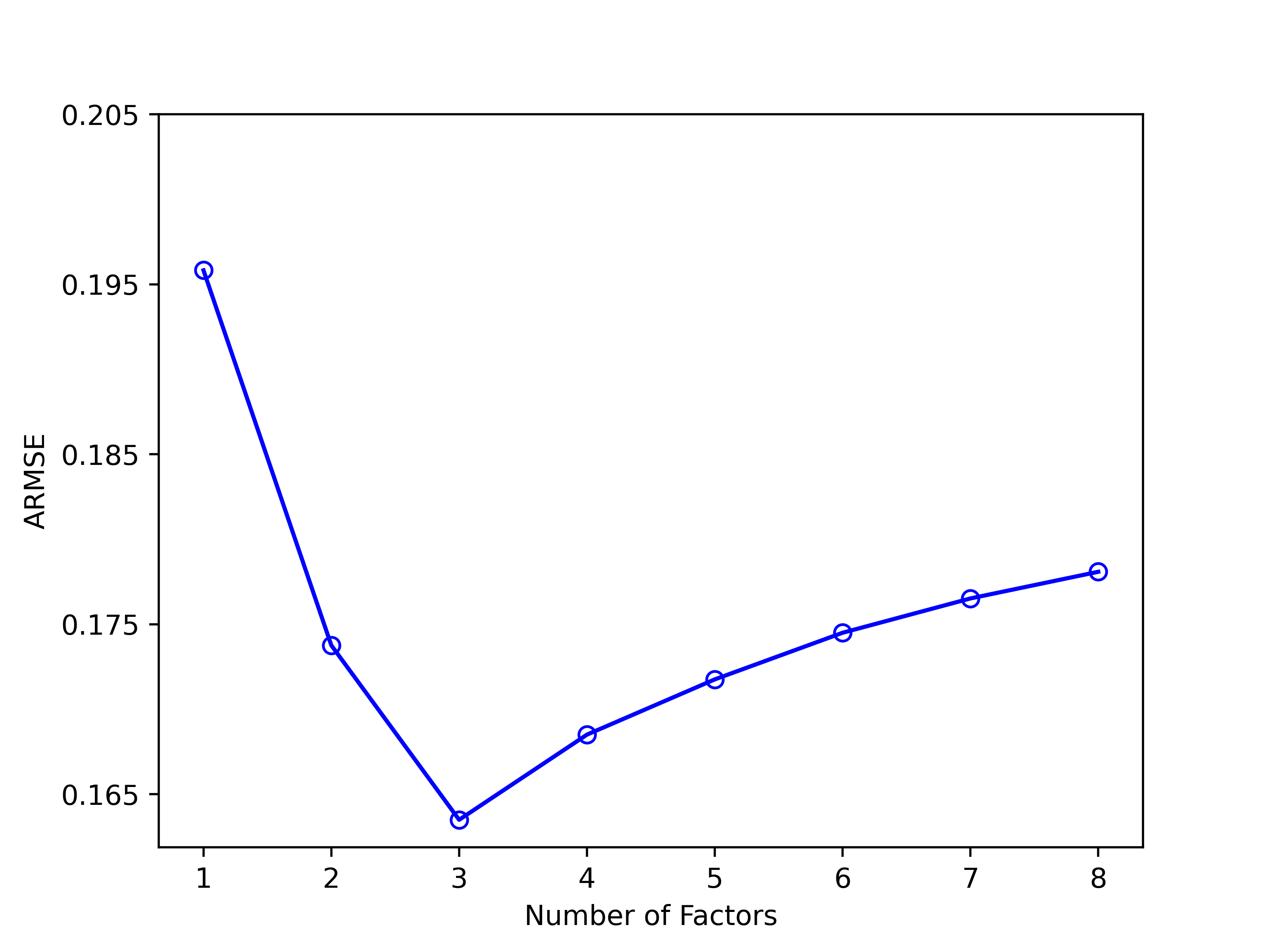}
    \end{figure}

\section{Conclusion and discussion}\label{Conclusion}
In this paper, we propose a novel estimator for determining the number of factors in high-dimensional factor models with random missing data and establish its consistency. Both homogeneous and heterogeneous missingness patterns are considered in isotropic and anisotropic settings. Simulations demonstrate that the proposed method outperforms existing methods, and real data examples confirm its practical utility. 

Several directions merit further investigation. First, covariance-based methods are not scale-invariant, which may lead to inconsistency in some cases. Replacing the covariance matrix with the correlation matrix could solve this scaling issue \citep{fan2022estimating}. Second, while our analysis focuses on MCAR, the method may be extended to more general missingness mechanisms. Third, this work considers static factor models, and extending the framework to dynamic factor models with missing data is a promising direction \citep{lam2012factor,li2017identifying}.

\bibliographystyle{chicago} 
\bibliography{ref} 

@book{yao2015sample,
  title={Sample covariance matrices and high-dimensional data analysis},
  author={Yao, Jianfeng and Zheng, Shurong and Bai, ZD},
  publisher={Cambridge University Press},
  address={New York},
  year={2015}
}

@incollection{bai2008clt,
  title={CLT for linear spectral statistics of large-dimensional sample covariance matrices},
  author={Bai, Zhidong D and Silverstein, Jack W},
  booktitle={Advances In Statistics},
  pages={281--333},
  year={2008},
  publisher={World Scientific}
}

@article{zeng2022order,
  title={Order Determination for Spiked Type Models},
  author={Zeng, Yicheng and Zhu, Lixing},
  journal={Statistica Sinica},
  volume={32},
  number={3},
  pages={1633--1659},
  year={2022},
  publisher={JSTOR}
}

@article{zeng2023order,
  title={Order determination for spiked-type models with a divergent number of spikes},
  author={Zeng, Yicheng and Zhu, Lixing},
  journal={Computational Statistics \& Data Analysis},
  volume={182},
  pages={107704},
  year={2023},
  publisher={Elsevier}
}

@article{baik2005phase,
  title={Phase transition of the largest eigenvalue for nonnull complex sample covariance matrices},
  author={Baik, Jinho and Ben Arous, G{\'e}rard and P{\'e}ch{\'e}, Sandrine},
  journal={The Annals of Probability},
  volume={33},
  number={5},
  pages={1643-1697},
  year={2005}
}

@article{bickel2008covariance,
  title={Covariance regularization by thresholding},
  author={Bickel, Peter J and Levina, Elizaveta},
  journal={The Annals of Statistics},
  volume={36},
  number={6},
  pages={2577-2604},
  year={2008}
}

@article{guan2011bayesian,
  title={Bayesian variable selection regression for genome-wide association studies and other large-scale problems},
  author={Guan, Yongtao and Stephens, Matthew},
  journal={The Annals of Applied Statistics},
  volume={5},
  number={3},
  pages={1780-1815},
  year={2011}
}

@article{richardson1997bayesian,
  title={On Bayesian analysis of mixtures with an unknown number of components (with discussion)},
  author={Richardson, Sylvia and Green, Peter J},
  journal={Journal of the Royal Statistical Society Series B: Statistical Methodology},
  volume={59},
  number={4},
  pages={731--792},
  year={1997},
  publisher={Oxford University Press}
}

@article{ding2021separable,
  title={Spiked separable covariance matrices and principal components},
  author={Ding, Xiucai and Yang, Fan},
  journal={The Annals of Statistics},
  volume={49},
  number={2},
  pages={1113-1138},
  year={2021}
}

@inproceedings{yu2025theory,
  title={A Theory-Driven Approach to Inner Product Matrix Estimation for Incomplete Data: An Eigenvalue Perspective},
  author={Yu, Fangchen and Zeng, Yicheng and Mao, Jianfeng and Li, Wenye},
  booktitle={Proceedings of the ACM on Web Conference 2025},
  pages={4077--4088},
  year={2025}
}

@article{bai2015unbalanced,
        author={Bai, Jushan and Liao, Yuan and Yang, Jisheng},
        year={2015},
	title={Unbalanced panel data models with interactive effects},
	  journal={The Oxford Handbook of Panel Data},
        pages={149--170},      
        publisher={Oxford University Press},
        doi = {https://doi.org/10.1093/oxfordhb/9780199940042.013.0005},
}

@article{stock2002macroeconomic,
  title={Macroeconomic forecasting using diffusion indexes},
  author={Stock, James H and Watson, Mark W},
  journal={Journal of Business \& Economic Statistics},
  volume={20},
  number={2},
  pages={147--162},
  year={2002},
  publisher={Taylor \& Francis},
  doi={https://doi.org/10.1198/073500102317351921}
}

@article{bai2002determining,
  title={Determining the number of factors in approximate factor models},
  author={Bai, Jushan and Ng, Serena},
  journal={Econometrica},
  volume={70},
  number={1},
  pages={191--221},
  year={2002},
  publisher={Wiley Online Library},
 doi={https://doi.org/10.1111/1468-0262.00273}
}

@article{su2017time,
  title={On time-varying factor models: estimation and testing},
  author={Su, Liangjun and Wang, Xia},
  journal={Journal of Econometrics},
  volume={198},
  number={1},
  pages={84--101},
  year={2017},
  publisher={Elsevier},
  doi={https://doi.org/10.1016/j.jeconom.2016.12.004}
}

@article{onatski2010determining,
  title={Determining the number of factors from empirical distribution of eigenvalues},
  author={Onatski, Alexei},
  journal={The Review of Economics and Statistics},
  volume={92},
  number={4},
  pages={1004--1016},
  year={2010},
  publisher={The MIT Press},
  doi={https://doi.org/10.1162/REST_a_00043}
}

@article{lam2012factor,
  title={Factor modeling for high-dimensional time series: inference for the number of factors},
  author={Lam, Clifford and Yao, Qiwei},
  journal={The Annals of Statistics},
  pages={694--726},
  year={2012},
  volume={40},
  number={2},
  publisher={JSTOR},
  doi={https://doi.org/10.1214/12-AOS970}

}

@article{li2017identifying,
  title={Identifying the number of factors from singular values of a large sample auto-covariance matrix},
  author={Li, Zeng and Wang, Qinwen and Yao, Jianfeng},
  journal={The Annals of Statistics},
  pages={257--288},
  year={2017},
  volume={45},
  number={1},
  publisher={JSTOR},
  doi={https://doi.org/10.1214/16-AOS1452}
}

@article{ahn2013eigenvalue,
  title={Eigenvalue ratio test for the number of factors},
  author={Ahn, Seung C and Horenstein, Alex R},
  journal={Econometrica},
  volume={81},
  number={3},
  pages={1203--1227},
  year={2013},
  publisher={Wiley Online Library},
  doi={https://doi.org/10.3982/ECTA8968}
}

@article{fan2022estimating,
  title={Estimating number of factors by adjusted eigenvalues thresholding},
  author={Fan, Jianqing and Guo, Jianhua and Zheng, Shurong},
  journal={Journal of the American Statistical Association},
  volume={117},
  number={538},
  pages={852--861},
  year={2022},
  publisher={Taylor \& Francis},
  doi={https://doi.org/10.1080/01621459.2020.1825448}
}

@article{ludvigson2007empirical,
  title={The empirical risk--return relation: a factor analysis approach},
  author={Ludvigson, Sydney C and Ng, Serena},
  journal={Journal of Financial Economics},
  volume={83},
  number={1},
  pages={171--222},
  year={2007},
  publisher={Elsevier},
  doi={https://doi.org/10.1016/j.jfineco.2005.12.002}
}

@article{jin2021factor,
  title={On factor models with random missing: {EM} estimation, inference, and cross validation},
  author={Jin, Sainan and Miao, Ke and Su, Liangjun},
  journal={Journal of Econometrics},
  volume={222},
  number={1},
  pages={745--777},
  year={2021},
  publisher={Elsevier},
  doi={https://doi.org/10.1016/j.jeconom.2020.08.002}
}

@article{bai2021matrix,
  title={Matrix completion, counterfactuals, and factor analysis of missing data},
  author={Bai, Jushan and Ng, Serena},
  journal={Journal of the American Statistical Association},
  volume={116},
  number={536},
  pages={1746--1763},
  year={2021},
  publisher={Taylor \& Francis},
  doi={https://doi.org/10.1080/01621459.2021.1967163}
}

@article{xiong2023large,
  title={Large dimensional latent factor modeling with missing observations and applications to causal inference},
  author={Xiong, Ruoxuan and Pelger, Markus},
  journal={Journal of Econometrics},
  volume={233},
  number={1},
  pages={271--301},
  year={2023},
  publisher={Elsevier},
  doi={https://doi.org/10.1016/j.jeconom.2022.04.005}
}

@article{cahan2023factor,
  title={Factor-based imputation of missing values and covariances in panel data of large dimensions},
  author={Cahan, Ercument and Bai, Jushan and Ng, Serena},
  journal={Journal of Econometrics},
  volume={233},
  number={1},
  pages={113--131},
  year={2023},
  publisher={Elsevier},
  doi={https://doi.org/10.1016/j.jeconom.2022.01.006}
}

@incollection{stock2016dynamic,
  title={Dynamic factor models, factor-augmented vector autoregressions, and structural vector autoregressions in macroeconomics},
  author={Stock, James H and Watson, Mark W},
  series={Handbook of Macroeconomics},
  volume={2},
  pages={415-525},
  year={2016},
  publisher={Elsevier},
  doi={https://doi.org/10.1016/bs.hesmac.2016.04.002}
}

@article{ke2023estimation,
  title={Estimation of the number of spiked eigenvalues in a covariance matrix by bulk eigenvalue matching analysis},
  author={Ke, Zheng Tracy and Ma, Yucong and Lin, Xihong},
  journal={Journal of the American Statistical Association},
  volume={118},
  number={541},
  pages={374--392},
  year={2023},
  publisher={Taylor \& Francis},
  doi={https://doi.org/10.1080/01621459.2021.1933497}
}

@article{johnstone2001distribution,
  title={On the distribution of the largest eigenvalue in principal components analysis},
  author={Johnstone, Iain M},
  journal={The Annals of Statistics},
  volume={29},
  number={2},
  pages={295--327},
  year={2001},
  publisher={Institute of Mathematical Statistics},
  doi={https://doi.org/10.1214/aos/1009210544}
}

@article{doz2011two,
  title={A two-step estimator for large approximate dynamic factor models based on Kalman filtering},
  author={Doz, Catherine and Giannone, Domenico and Reichlin, Lucrezia},
  journal={Journal of Econometrics},
  volume={164},
  number={1},
  pages={188--205},
  year={2011},
  publisher={Elsevier},
  doi={https://doi.org/10.1016/j.jeconom.2011.02.012}
}

@article{giannone2008nowcasting,
  title={Nowcasting: the real-time informational content of macroeconomic data},
  author={Giannone, Domenico and Reichlin, Lucrezia and Small, David},
  journal={Journal of Monetary Economics},
  volume={55},
  number={4},
  pages={665--676},
  year={2008},
  publisher={Elsevier},
  doi={https://doi.org/10.1016/j.jmoneco.2008.05.010}
}

@article{baik2006eigenvalues,
  title={Eigenvalues of large sample covariance matrices of spiked population models},
  author={Baik, Jinho and Silverstein, Jack W},
  journal={Journal of Multivariate Analysis},
  volume={97},
  number={6},
  pages={1382--1408},
  year={2006},
  publisher={Elsevier},
  doi={https://doi.org/10.1016/j.jmva.2005.08.003}
}

@article{marchenko1967distribution,
  title={Distribution of eigenvalues for some sets of random matrices},
  author={Marchenko, Vladimir Alexandrovich and Pastur, Leonid Andreevich},
  journal={Mathematics of the USSR-Sbornik},
  volume={1},
  number={4},
  pages={457--483},
  year={1967},
  publisher={Russian Academy of Sciences, Steklov Mathematical Institute of Russian~…},
  doi={https://doi.org/10.1070/SM1967v001n04ABEH001994}
}

@article{cai2020limiting,
  title={Limiting laws for divergent spiked eigenvalues and largest nonspiked eigenvalue of sample covariance matrices},
  author={Cai, T Tony and Han, Xiao and Pan, Guangming},
  journal={The Annals of Statistics},
  volume={48},
  number={3},
  pages={1255--1280},
  year={2020},
  doi={https://doi.org/10.1214/18-AOS1798}
}

@article{bloemendal2016principal,
  title={On the principal components of sample covariance matrices},
  author={Bloemendal, Alex and Knowles, Antti and Yau, Horng Tzer and Yin, Jun},
  journal={Probability Theory and Related Fields},
  volume={164},
  number={1},
  pages={459--552},
  year={2016},
  publisher={Springer},
  doi={https://doi.org/10.1007/s00440-015-0616-x}
}

@article{bai2012sample,
  title={On sample eigenvalues in a generalized spiked population model},
  author={Bai, Zhidong and Yao, Jianfeng},
  journal={Journal of Multivariate Analysis},
  volume={106},
  pages={167--177},
  year={2012},
  publisher={Elsevier},
  doi={https://doi.org/10.1016/j.jmva.2011.10.009}
}

@article{fama1993common,
  title={Common risk factors in the returns on stocks and bonds},
  author={Fama, Eugene F and French, Kenneth R},
  journal={Journal of Financial Economics},
  volume={33},
  number={1},
  pages={3--56},
  year={1993},
  publisher={Elsevier},
  doi={https://doi.org/10.1016/0304-405X(93)90023-5}
}

@article{yu2024testing,
  title={Testing the number of common factors by bootstrapped sample covariance matrix in high-dimensional factor models},
  author={Yu, Long and Zhao, Peng and Zhou, Wang},
  journal={Journal of the American Statistical Association},
  pages={1--22},
  year={2024},
  publisher={Taylor \& Francis},
  doi={https://doi.org/10.1080/01621459.2024.2346364}
}

@article{stock1989new,
  title={New indexes of coincident and leading economic indicators},
  author={Stock, James H and Watson, Mark W},
  journal={NBER Macroeconomics Annual},
  volume={4},
  pages={351--409},
  year={1989},
  publisher={MIT press},
  doi={https://doi.org/10.1086/654119}
}

@article{hsiao2012panel,
  title={A panel data approach for program evaluation: measuring the benefits of political and economic integration of {Hong} {Kong} with mainland {China}},
  author={Hsiao, Cheng and Steve Ching, H and Ki Wan, Shui},
  journal={Journal of Applied Econometrics},
  volume={27},
  number={5},
  pages={705--740},
  year={2012},
  publisher={Wiley Online Library},
  doi={ https://doi.org/10.1002/jae.1230}
}

@article{banbura2014maximum,
  title={Maximum likelihood estimation of factor models on datasets with arbitrary pattern of missing data},
  author={Ba{\'n}bura, Marta and Modugno, Michele},
  journal={Journal of Applied Econometrics},
  volume={29},
  number={1},
  pages={133--160},
  year={2014},
  publisher={Wiley Online Library},
  doi={ https://doi.org/10.1002/jae.2306}
}

@book{little2019statistical,
  title={Statistical analysis with missing data},
  author={Little, Roderick JA and Rubin, Donald B},
  volume={793},
  year={2019},
  publisher={John Wiley \& Sons},
  doi={}
}

@article{seaman2013meant,
  title={What is meant by “missing at random”?},
  author={Seaman, Shaun and Galati, John and Jackson, Dan and Carlin, John},
  journal={Statistical Science},
  volume={28},
  number={2},
  pages={257--268},
  year={2013},
  doi={https://doi.org/10.1214/13-STS415}
}

@article{johnstone2009consistency,
  title={On consistency and sparsity for principal components analysis in high dimensions},
  author={Johnstone, Iain M and Lu, Arthur Yu},
  journal={Journal of the American Statistical Association},
  volume={104},
  number={486},
  pages={682--693},
  year={2009},
  publisher={Taylor \& Francis},
  doi={https://doi.org/10.1198/jasa.2009.0121}
}

@article{dobriban2020optimal,
  title={Optimal prediction in the linearly transformed spiked model},
  author={Dobriban, Edgar and Leeb, William and Singer, Amit},
  journal={The Annals of Statistics},
  volume={48},
  number={1},
  pages={491--513},
  year={2020},
  publisher={JSTOR},
  doi={https://doi.org/10.1214/19-aos1819}
}

@article{jurczak2017spectral,
  title={Spectral analysis of high-dimensional sample covariance matrices with missing observations},
  author={Jurczak, Kamil and Rohde, Angelika},
  journal={Bernoulli},
  volume={23},
  number={4A},
  pages={2466--2532},
  year={2017},
  doi={https://doi.org/10.3150/16-BEJ815}
}

@article{nadakuditi2014optshrink,
  title={Optshrink: an algorithm for improved low-rank signal matrix denoising by optimal, data-driven singular value shrinkage},
  author={Nadakuditi, Raj Rao},
  journal={IEEE Transactions on Information Theory},
  volume={60},
  number={5},
  pages={3002--3018},
  year={2014},
  publisher={IEEE},
  doi={https://doi.org/10.1109/TIT.2014.2311661}
}

@article{cattell1966scree,
  title={The scree test for the number of factors},
  author={Cattell, Raymond B},
  journal={Multivariate Behavioral Research},
  volume={1},
  number={2},
  pages={245--276},
  year={1966},
  publisher={Taylor \& Francis},
  doi={https://doi.org/10.1207/s15327906mbr0102_10}
}

@book{hiai2006semicircle,
  title={The Semicircle Law, Free Random Variables and Entropy},
  author={Hiai, Fumio and Petz, Denes},
  year={2006},
  publisher={American Mathematical Society}
}

@article{wang2014limiting,
  title={Limiting spectral distribution of renormalized separable sample covariance matrices when p/n→ 0},
  author={Wang, Lili and Paul, Debashis},
  journal={Journal of Multivariate Analysis},
  volume={126},
  pages={25--52},
  year={2014},
  publisher={Elsevier}
}
\appendix 
%Notations for Appendices
\renewcommand{\theequation}{A.\arabic{equation}}
\renewcommand{\thetable}{A.\arabic{table}}
\renewcommand{\thefigure}{A.\arabic{figure}}
\renewcommand{\thesection}{A.\arabic{section}}
\renewcommand{\thelemma}{A.\arabic{lemma}}

\numberwithin{equation}{section}

\vspace{0.5in}

\noindent{\bf \LARGE Appendix}

\section{Proof of Main Results}

%\subsection{Proof of Theorem~\ref{consistency}, \ref{consistency2} and \ref{consistency3}}

\begin{proof}[Proof of Theorem~\ref{consistency}]

Under Assumptions~\ref{assumption_HD}-\ref{assumption_Sigma_x} with $K=1$, the ESD of the noise part converges to a Marčenko–Pastur law with upper edge $v=p\sigma^2(1+\sqrt{\gamma})^2$. 
Moreover, the largest noise eigenvalue satisfies 
\$
\hat{\lambda}_{r_1+1}^o-p\sigma^2(1+\sqrt{\gamma})^2=O_p(n^{-2/3}).
\$
For any $\epsilon_n =o(1)$ with $\epsilon_n n^{2/3}\rightarrow \infty$, it follows that
\$
\lim\limits_{n\to +\infty}\mathbb{P} \left\{\hat{\lambda}_{r_1+1}^o <p\sigma^2(1+\sqrt{\gamma})^2+\epsilon_n\right\} = 1
\$

On the other hand, for each identifiable factor $i\le r_1$, the corresponding sample eigenvalue separates from the bulk and converges to a limit strictly larger than $v=p\sigma^2(1+\sqrt{\gamma})^2$, so 
\$
\lim\limits_{n\to +\infty}\mathbb{P} \left\{\hat{\lambda}_{r_1}^o >p\sigma^2(1+\sqrt{\gamma})^2+\epsilon_n\right\} = 1. 
\$
Therefore, $\hat r(p(\sigma^2 (1+\sqrt \gamma)^2),\epsilon_n)$ consistently estimates the number $r_1$ of identifiable spikes, that is, 
\$
\lim_{n\to \infty}\mathbb{P} \left\{\hat{r}(p\sigma^2(1+\sqrt \gamma)^2,\epsilon_n)=r_1\right\} = 1.
\$
\end{proof}

\begin{proof}[Proofs of Theorems~\ref{consistency2} and \ref{consistency3}] 

The proofs of Theorems~\ref{consistency2} and \ref{consistency3} are similar to that of Theorem~\ref{consistency} and are  therefore omitted.

\end{proof}

\begin{proof}[Proof of Lemma~\ref{lemma_heterogeneous}]
(a) For $\Sigma$ generated from \eqref{TruncGamma_Sigma_spiked_model}, it can be verified that the ESD of $\Sigma P_d$ converges, and hence Assumption~\ref{assumption_P_d} is satisfied. Based on the decomposition of $X^o$ in \eqref{feature_specific_data}, Theorem~2.14 of \cite{yao2015sample} implies that the ESD of ${S_n^o}$ converges weakly. 
(b) Based on the decomposition of $X^o$ in \eqref{sample_specific_data} and Theorem~2.1 in \cite{wang2014limiting}, the ESD of ${S_n^o}$ converges weakly. 

\end{proof}

\begin{proof}[Proof of Lemma~\ref{LSD_moments}]

We first introduce some notations for convenience.
\begin{enumerate}
\item For feature-specific missingness, define $\Sigma_f=\Sigma P_d$ and $Z_f=P_d^{-1/2}Y^o$.
\item For sample-specific missingness, define $Z_s=Y^o Q_n^{-1/2}$;
\item Recall that for any $n\times n$ matrix $A$, 
\$
\beta_k (A)=\frac{1}{n}{\rm tr}(A^k),\ k\in \NN. 
\$ 
\item Recall that for sequences $\{a_n\}$ and $\{b_n\}$, we write $a_n\sim b_n$ if $a_n/b_n\rightarrow 1$, $a_n \overset{\rm a.s.}{\sim} b_n$ if $a_n/b_n \overset{\rm a.s.}{\rightarrow} 1$, and $a_n \overset{\rm p}{\sim} b_n$ if $a_n/b_n \overset{\rm p}{\rightarrow} 1$. 
\end{enumerate}

\noindent $(a)$. For the complete data matrix $X$, we write $X=\Sigma^{1/2} Y$. By Lemma~\ref{lemma_moment}, 
\$
\beta_1(S_n) 
&= \beta_1\left(\frac{1}{n} \Sigma^{1/2} Y Y^\top \Sigma^{1/2}\right)
\sim \beta_{1}(\Sigma) 
= \frac{1}{d}\sum_{1\le i\le d} \lambda_i (\Sigma) 
\sim \frac{1}{d}\sum_{1\le i\le d}\sigma_i^2 \sim \EE (\sigma_1^2) = \tilde\sigma^2,
\$
where the fourth step follows from the finite-rank structure in \eqref{Gamma_Sigma_spiked_model} and the boundedness of the spikes. Similarly, we have 
\$
\beta_2(S_n) 
& \sim \beta_{2}(\Sigma)+\gamma \beta_{1}^2(\Sigma) \sim \EE (\sigma_1^4) + \gamma \{\EE (\sigma_1^2)\}^2 
= \tilde\sigma^4 (1+\tilde \theta^{-1}) +\gamma \tilde\sigma^4.   
\$

\noindent $(b)$. For the feature-specific heterogeneous missingness case, according to \eqref{feature_specific_data}, the observed data matrix can be written as 
\$
X^o=\Sigma^{1/2} Y^o = \Sigma^{1/2} P_d^{1/2}P_d^{-1/2}  Y^o 
= \Sigma_f^{1/2} Z_f,
\$
where the entries of $Z_f$ are i.i.d. with zero mean and unit variance. 

We first compute $\beta_1(\Sigma_f)$ and $\beta_2(\Sigma_f)$: 
\begin{gather}
\beta_{1}(\Sigma_f)
= \frac{1}{d}\sum_{1\le i\le d} \lambda_i (\Sigma_f) 
= \frac{1}{d}\sum_{1\le k \le K}\sum_{i\in \mathcal I_k} p_k \sigma_i^2
\sim \EE \left( \frac{1}{d}\sum_{1\le k \le K}\sum_{i\in \mathcal I_k} p_k \sigma_i^2 \right)\nonumber \\
= \frac{1}{d}\sum_{1\le k \le K}\sum_{i\in \mathcal I_k} p_k \EE (\sigma_i^2)
= \frac{1}{d}\sum_{1\le k \le K}\sum_{i\in \mathcal I_k} p_k \tilde\sigma^2
= \tilde\sigma^2 \left( \frac{1}{d}\sum_{1\le k\le K} p_k d_k \right)\label{first_moment}
\end{gather}
and 
\begin{gather}
\beta_{2}(\Sigma_f)
= \frac{1}{d}\sum_{1\le i\le d} \lambda_i^2 (\Sigma_f) 
= \frac{1}{d}\sum_{1\le k \le K}\sum_{i\in \mathcal I_k} p_k^2 \sigma_i^4 
\sim \EE \left( \frac{1}{d}\sum_{1\le k \le K}\sum_{i\in \mathcal I_k} p_k^2 \sigma_i^4 \right)\nonumber \\
= \frac{1}{d}\sum_{1\le k \le K}\sum_{i\in \mathcal I_k} p_k^2 \EE (\sigma_i^4)
= \frac{1}{d}\sum_{1\le k \le K}\sum_{i\in \mathcal I_k} p_k^2 \tilde\sigma^4 (1+\tilde\theta^{-1}) 
= \tilde\sigma^4 (1+\tilde\theta^{-1}) \left( \frac{1}{d}\sum_{1\le k\le K} p_k^2 d_k\right).\label{second_moment}
\end{gather}

By Lemma~\ref{lemma_moment} and \eqref{first_moment}, 
\begin{gather*}
\beta_1(S_n^o) = \beta_1\left(\frac{1}{n} \Sigma_f^{1/2} Z_f Z_f^\top \Sigma_f^{1/2}\right)
\sim \beta_{1}(\Sigma_f)\sim \tilde \sigma^2 \left( \frac{1}{d}\sum_{1\le k\le K} p_k d_k  \right). 
\end{gather*}
Similarly, by Lemma~\ref{lemma_moment} together with \eqref{first_moment} and \eqref{second_moment}, 
\begin{gather*}
\beta_2 (S_n^o) 
\sim \beta_{2}(\Sigma_f)+ \gamma \beta_1^2(\Sigma_f) 
\sim \tilde \sigma^4 (1+\tilde\theta^{-1}) \left( \frac{1}{d}\sum_{1\le k\le K}p_k^2 d_k\right) + \gamma \tilde \sigma^4 \left(\frac{1}{d}\sum_{1\le k\le K}p_k d_k \right)^2.
\end{gather*}

\noindent $(c)$. For the sample-specific heterogeneous missingness case, according to \eqref{sample_specific_data}, the observed data matrix can be written as
\$
X^o=\Sigma^{1/2} Y^o = \Sigma^{1/2} Y^o Q_n^{-1/2}Q_n^{1/2}
= \Sigma^{1/2} Z_s Q_n^{1/2}, 
\$
where the entries of $Z_s=Y^o Q_n^{-1/2}$ are i.i.d. with zero mean and unit variance. 

For the first moment, we have
\begin{gather*}
\beta_1(S_n^o) = \beta_1\left( \frac{1}{n} \Sigma^{1/2} Z_s Q_n Z_s^\top \Sigma^{1/2}\right)
 = \beta_1 \left( \frac{1}{n} Z_s Q_n Z_s^\top\Sigma \right)
 \sim \beta_1 \left( \frac{1}{n} Z_s Q_n Z_s^\top \right) \beta_1 (\Sigma) \\
 = \beta_1 \left( \frac{1}{d} Z_s^\top Z_s Q_n \right) \beta_1 (\Sigma) 
 \sim \beta_1\left( \frac{1}{d}Z_s^\top Z_s \right) \beta_1(Q_n) \beta_1(\Sigma) 
 \sim \tilde\sigma^2\left( \frac{1}{n}\sum_{1\le \ell \le L}q_\ell n_\ell\right),
\end{gather*}
where the third and fifth steps hold by using Lemma~\ref{lemma_asymp_free} and \ref{lemma_free_prob_fourth_moment}. 

For the second moment, it holds that
\$
\beta_2(S_n^o)  
& =  \beta_1\left\{ \left(\frac{1}{n} \Sigma^{1/2} Z_s Q_n Z_s^\top \Sigma^{1/2} \right) \left(\frac{1}{n} \Sigma^{1/2} Z_s Q_n Z_s^\top \Sigma^{1/2} \right) \right\}\\
& = \beta_1\left\{ \left(\frac{1}{n} Z_s Q_n Z_s^\top \right)  \Sigma \left(\frac{1}{n} Z_s Q_n Z_s^\top \right)  \Sigma \right\}\\
& \sim \beta_1^2 \left( \frac{1}{n} Z_s Q_n Z_s^\top \right) \beta_2 \left(\Sigma\right)
+\beta_2  \left(\frac{1}{n} Z_s Q_n Z_s^\top \right) \beta_1^2(\Sigma) - \beta_1^2 \left(\frac{1}{n} Z_s Q_n Z_s^\top \right)\beta_1^2 (\Sigma)\\
& \sim \tilde\sigma^4 (1+\tilde\theta^{-1}) \beta_1^2 \left( \frac{1}{n} Z_s Q_n Z_s^\top \right) + \tilde\sigma^4 \beta_2 \left(\frac{1}{n} Z_s Q_n Z_s^\top \right) - \tilde\sigma^4 \beta_1^2 \left( \frac{1}{n} Z_s Q_n Z_s^\top \right),
\$
where we used Lemma~\ref{lemma_free_prob_fourth_moment} in the third step, and where
\$
\beta_1 \left( \frac{1}{n} Z_s Q_n Z_s^\top \right) = \beta_1 \left( \frac{1}{n}Z_s^\top Z_s Q_n \right)
\sim \beta_1 \left( \frac{1}{n}Z_s^\top Z_s \right) \beta_1 \left( Q_n \right)
\sim \frac{1}{n}\sum_{1\le \ell \le L}q_\ell n_\ell,
\$

And by Lemma~\ref{lemma_free_prob_fourth_moment}, we have
\$
& \quad \beta_2 \left(\frac{1}{n} Z_s Q_n Z_s^\top \right)\\
& =  \beta_1\left\{ \left( \frac{1}{n} Z_s Q_n Z_s^\top\right)\left( \frac{1}{n} Z_s Q_n Z_s^\top\right)\right\}\\
& =  \gamma^{-1} \beta_1 \left\{ Q_n  \left(\frac{1}{n} Z_s^\top Z_s  \right) Q_n \left( \frac{1}{n} Z_s^\top Z_s \right) \right\}\\
& \sim \gamma^{-1} \left\{\beta_1^2 \left(Q_n\right) \beta_2 \left(\frac{1}{n} Z_s^\top Z_s \right)  + \beta_2 \left(Q_n\right) \beta_1^2\left( \frac{1}{n} Z_s^\top Z_s\right) - \beta_1^2 \left(Q_n\right) \beta_1^2\left(\frac{1}{n} Z_s^\top Z_s\right) \right\}\\
& \sim \gamma^{-1} \left\{ \left( \frac{1}{n}\sum_{1\le \ell \le L}q_\ell n_\ell\right)^2 \gamma^2 (1+\gamma^{-1}) + \left( \frac{1}{n}\sum_{1\le \ell \le L}q_\ell^2 n_\ell\right) \gamma^2 - \left( \frac{1}{n}\sum_{1\le \ell \le L}q_\ell n_\ell\right)^2 \gamma^2 \right\}\\
& = \left( \frac{1}{n}\sum_{1\le \ell \le L}q_\ell n_\ell\right)^2 + \gamma \left( \frac{1}{n}\sum_{1\le \ell \le L}q_\ell^2 n_\ell\right),
\$
where the third step holds by using Lemma~\ref{lemma_asymp_free} and \ref{lemma_free_prob_fourth_moment}. 

All these together lead to 
\$
\beta_2(S_n^o) 
& \sim \tilde\sigma^4 (1+\tilde\theta^{-1}) \left( \frac{1}{n}\sum_{1\le \ell \le L}q_\ell n_\ell \right)^2 + \tilde\sigma^4 \left\{ \left( \frac{1}{n}\sum_{1\le \ell \le L}q_\ell n_\ell\right)^2 + \gamma \left( \frac{1}{n}\sum_{1\le \ell \le L}q_\ell^2 n_\ell\right)\right\} \\
& \quad - \tilde\sigma^4\left( \frac{1}{n}\sum_{1\le \ell \le L}q_\ell n_\ell \right)^2 \\
& = \tilde\sigma^4(1+\tilde\theta^{-1})\left(\frac{1}{n}\sum_{1\le \ell\le L}q_\ell n_\ell\right)^2 + \gamma\tilde\sigma^4\left(\frac{1}{n}\sum_{1\le \ell \le L}q_\ell^2 n_\ell\right).
\$
\end{proof}

\section{Technical lemmas}

\begin{lemma}\label{lemma_inner_hadamard}
For any diagonal matrix $A={\rm diag}(a_1,\cdots,a_d)\in\mathbb{R}^{d\times d}$, it holds that $(AB)\circ C = A (B\circ C)$.
\end{lemma}

\begin{proof}[Proof of Lemma~\ref{lemma_inner_hadamard}]
For the diagonal matrix $A$, we have 
\$    
((AB)\circ C)_{ij} = (AB)_{ij} C_{ij}
= a_i B_{ij}C_{ij}
= a_i (B\circ C)_{ij}
= ( A(B\circ C))_{ij}. 
\$
This completes the proof. 
\end{proof}

\begin{lemma}\label{lemma_moment}
Assume $Z\in \mathbb R^{d\times n}$ has i.i.d. entries with zero mean and unit variance, and that $\Sigma\in \mathbb R^{d\times d}$ is a nonnegative definite matrix whose ESD converges to a non-random probability measure. Then for the sample covariance matrix $\hat \Sigma=\Sigma^{1/2}ZZ^\top \Sigma^{1/2}/n$, we have 
\$
\beta_1(\hat \Sigma)\sim \beta_1(\Sigma),\quad \beta_2(\hat \Sigma)\sim \beta_2(\Sigma)+\gamma \beta_1^2(\Sigma). 
\$
Moreover, when $\Sigma=I_d$, it follows that 
\$
\beta_1(\hat \Sigma)\sim 1,\quad \beta_2(\hat \Sigma)\sim 1+\gamma.
\$
\end{lemma}

\begin{proof}[Proof of Lemma~\ref{lemma_moment}] 
This lemma follows directly from Lemma~2.16 in \cite{yao2015sample}. 
\end{proof}

\begin{definition}[Asymptotic freeness]
A sequence of random variables $\{a_1,\cdots,a_k\}$ is said to be {\it asymptotically free}  if 
\$ \beta_1 \left\{\prod_{1\le j\le m} \mathcal P_j\left(a_{i_j}-\beta_1 \left(\mathcal P_j(a_{i_j})\right)\right)\right\}\rightarrow 0
\$ 
for any positive integer $m$, any polynomials $\mathcal P_1,\cdots,\mathcal P_m$, and any indices $i_1,\cdots,i_m\in [k]$ such that no two adjacent indices are equal. 
\end{definition}

\begin{lemma}[Theorem 4.3.11 of \citep{hiai2006semicircle}]\label{lemma_asymp_free}
Let $ ({A_{N,1},\cdots,A_{N,n}})$ be a family of $N\times N$ complex bi-unitarily invariant random matrices, whose joint distribution is invariant under both left and right unitary multiplication, and let $({B_{N,1},\cdots,B_{N,m}})$ be a family of non-random diagonal matrices. Suppose that the empirical spectral distributions of $ (A_{N,i}A_{N,i}^*)$ and $(B_{N,j}B_{N,j}^*)$  converge almost surely to compactly supported limiting distributions as $N$ goes to infinity,  where $^*$ denotes the conjugate transpose.  Then the family
\$
\left\{ \{ A_{N,i}\}_{i\in \{1,\cdots,n\}}, \{A_{N,i}^*\}_{i\in \{1,\cdots,n\}}, \{B_{N,j}\}_{j\in \{1,\cdots,m\}}, \{B_{N,j}^*\}_{j\in \{1,\cdots,m\}} \right\}
\$
is asymptotically free almost surely as $N\to\infty$.
\end{lemma}

\begin{lemma}\label{lemma_free_prob_fourth_moment}
Suppose matrices $B\in \mathbb{R}^{d\times d}$ and $C\in \mathbb{R}^{d\times d}$ are asymptotically free. Then
\begin{gather*}
\beta_1 (BC)\sim \beta_1 (B) \beta_1(C),\\
\beta_1(B^2C)\sim \beta_2(B)\beta_1(C),\quad \beta_1(BC^2)\sim \beta_1(B)\beta_2(C),\\
\beta_1(BCBC) 
\sim \beta_1^2(B)\beta_2(C)+\beta_2(B)\beta_1^2(C)-\beta_1^2(B)\beta_1^2(C).
\end{gather*}
\end{lemma}

\begin{proof}[Proof of Lemma~\ref{lemma_free_prob_fourth_moment}]

By the definition of asymptotic freeness, 
\$
\beta_1\left\{ (B-\beta_1 (B))(C-\beta_1 (C))\right\}\rightarrow 0,
\$
where
\$
(B-\beta_1 (B))(C-\beta_1 (C)) = BC -B \beta_1 (C) + \beta_1(B) C -\beta_1(B)\beta_1(C). 
\$
It follows that 
\$
\beta_1 (BC)\sim \beta_1 (B) \beta_1 (C).
\$

Similarly, 
\$
\beta_1\left\{ \left(B -\beta_1(B^2)\right)^2\left(C-\beta_1 (C)\right) \right\}\rightarrow 0,
\$
which implies
\$
\beta_1 (B^2 C)
\sim & 2\beta_1 (B^2) \beta_1 (BC) +\beta_1 (C)\beta_1 (B^2) - 2\beta_1 (B^2)\beta_1 (B) \beta_1 (C)\\
= & 2\beta_1 (B^2) \beta_1 (B) \beta_1(C)+\beta_1 (C)\beta_1 (B^2) - 2\beta_1 (B^2)\beta_1 (B) \beta_1 (C)\\
= & \beta_1 (C)\beta_1 (B^2)\\
= & \beta_1 (C)\beta_2 (B).
\$
By symmetry, we have 
\$
\beta_1(BC^2)= \beta_1 (C^2 B)= \beta_1(B) \beta_1(C^2)
= \beta_1(B) \beta_2(C).
\$

Moreover, by asymptotic freeness,
\$
\beta_1\left\{ \left(B-\beta_1 (B)\right) \left(C-\beta_1 (C)\right) \left(B-\beta_1(B)\right) \left(C-\beta_1(C)\right) \right\}\rightarrow 0.
\$
Expanding the product and collecting terms yields
\$
& \left(B-\beta_1 (B)\right) \left(C-\beta_1 (C)\right) \left(B-\beta_1(B)\right) \left(C-\beta_1(C)\right) \\
=& BCBC - BCB\beta_1(C) - BC\beta_1(B)C + BC \beta_1(B)\beta_1(C)\\
& - B\beta_1 (C) BC + B\beta_1(C)B \beta_1(C) + B \beta_1 (C) \beta_1 (B) C- B \beta_1 (C) \beta_1 (B) \beta_1 (C)\\
& - \beta_1 (B) CBC +\beta_1 (B)CB \beta_1 (C) +\beta_1 (B)C\beta_1 (B)C - \beta_1 (B)C \beta_1 (B)\beta_1 (C)\\
& + \beta_1 (B) \beta_1 (C) BC -\beta_1 (B)\beta_1 (C) B \beta_1 (C) - \beta_1 (B) \beta_1 (C) \beta_1 (B) C +\beta_1 (B)\beta_1 (C)\beta_1 (B)\beta_1 (C).
\$
Substituting the previous relations, we obtain
\$
\beta_1 (BCBC) 
\sim & 2 \beta_1 (B^2 C) \beta_1 (C) + 2\beta_1 (BC^2) \beta_1 (B)- \beta_1 (B^2)\beta_1^2 (C) -\beta_1^2 (B)\beta_1 (C^2) \\
&- 4 \beta_1 (BC)\beta_1 (B)\beta_1 (C)+ 3\beta_1^2 (B)\beta_1^2 (C)\\
\sim & 2 \beta_1 (B^2) \beta_1^2 (C) + 2\beta_1^2 (B) \beta_1 (C^2)- \beta_1 (B^2)\beta_1^2 (C) - \beta_1^2 (B) \beta_1 (C^2)\\
&- 4 \beta_1^2 (B)\beta_1^2 (C)+ 3\beta_1^2 (B)\beta_1^2 (C)\\
= & \beta_1 (B^2) \beta_1^2 (C) +\beta_1^2 (B) \beta_1 (C^2) -\beta_1^2 (B)\beta_1^2 (C)\\
= & \beta_2 (B) \beta_1^2 (C) +\beta_1^2 (B) \beta_2 (C) -\beta_1^2 (B)\beta_1^2 (C).
\$
This completes the proof.

\end{proof}

\end{document}